\def\version{November 25, 2014}

\documentclass[12pt]{article}

\def\macrosPb{}
\def\macrosHarxiv{}
\usepackage[psamsfonts]{amsfonts}
\usepackage{amsmath,amssymb,amsthm}
\usepackage[dvips]{graphicx}
\usepackage{appendix}
\usepackage{bbm} 
\usepackage{amsbsy}
\usepackage{enumerate}
\usepackage{cite}

\ifdefined\macrosPa
  \usepackage[textwidth=465pt,textheight=650pt,centering]{geometry} 
\else\ifdefined\macrosPb
  \usepackage[textwidth=500pt,textheight=650pt,centering]{geometry} 
\fi\fi

\ifdefined\macrosS
  \makeatletter
  
  \def\boldsymbol{\pmb}
  \makeatother

  \usepackage{mathptmx}
  \DeclareMathAlphabet{\mathcal}{OMS}{cmsy}{m}{n}
\fi

\def\UseSection{
        \numberwithin{equation}{section}
	\theoremstyle{plain}
        \newtheorem{theorem}    {Theorem}[section]
        \DefineTheorems 
}

\def\DefineTheorems{
	
	\newtheorem{lemma}      [theorem] {Lemma}
	
	\newtheorem{prop}       [theorem] {Proposition}
	
	\newtheorem{cor}        [theorem] {Corollary}

	\theoremstyle{definition}
	\newtheorem{defn}       [theorem] {Definition}

	\theoremstyle{definition}

}

\newcommand{\bt}   {\begin{theorem}}
\newcommand{\et}   {\end  {theorem}}
\newcommand{\bl}   {\begin{lemma}}
\newcommand{\el}   {\end  {lemma}}
\newcommand{\bp}   {\begin{prop}}
\newcommand{\ep}   {\end  {prop}}
\newcommand{\bc}   {\begin{cor}}
\newcommand{\ec}   {\end  {cor}}
\newcommand{\bd}   {\begin{defn}}
\newcommand{\ed}   {\end  {defn}}

\newcommand{\ba}   {\begin{array}}
\newcommand{\ea}   {\end  {array}}
\newcommand{\be}   {\begin{enumerate}}
\newcommand{\ee}   {\end  {enumerate}}
\newcommand{\bi}   {\begin{itemize}}
\newcommand{\ei}   {\end  {itemize}}

\def\eq#1\en{\begin{equation}#1\end{equation}}  
\def\eqsplit#1\ensplit{
	\begin{equation}\begin{split}#1\end{split}\end{equation}
	}
\def\eqalign#1\enalign{
	\begin{align}#1\end{align}
	}
\def\eqmul#1\enmul{
	\begin{multline}#1\end{multline}
	}
\newcommand{\eqarrstar} {\begin{eqnarray*}} 
\newcommand{\enarrstar} {\end{eqnarray*}} 
\newcommand{\eqarray}   {\begin{eqnarray}} 
\newcommand{\enarray}   {\end{eqnarray}} 
\newcommand{\nnb}	{\nonumber \\} 

\newcommand{\lbeq}[1]  {\label{e:#1}}
\newcommand{\refeq}[1] {\eqref{e:#1}}    
\newcommand{\lbfg}[1]  {\label{fg: #1}}

\makeatletter
\newcommand{\labelcounter}[2]{{%
	\stepcounter{#1}
	\protected@write\@auxout{}%
	{\string\newlabel{#2}{{\csname the#1\endcsname}{\thepage}}}%
	{\ref{#2}}
	}}
\makeatother
\newcommand{\Ebold} {{\mathbb E}}

\newcommand{\Nbold} {{\mathbb N}}

\newcommand{\Pbold} {{\mathbb P}}

\newcommand{\Rbold} {{\mathbb R}}

\newcommand{\Zbold} {{\mathbb Z}}

\newcommand{\Bcal}   {\mathcal{B}} 
 
\newcommand{\Dcal}   {\mathcal{D}}

\newcommand{\Kcal}   {\mathcal{K}}

\newcommand{\Ncal}   {\mathcal{N}} 
 
\newcommand{\Pcal}   {\mathcal{P}}

\newcommand{\Scal}   {\mathcal{S}}

\newcommand{\Vcal}   {\mathcal{V}} 
\newcommand{\Wcal}   {\mathcal{W}}

\newcommand{\Zd}    {{ {\Zbold}^d }}

\newcommand{\spose}[1] {{\hbox to 0pt{#1\hss}} }
\newcommand{\ltapprox} {\mathrel{\spose{\lower 3pt\hbox{$\mathchar"218$}}
 \raise 2.0pt\hbox{$\mathchar"13C$}}}
\newcommand{\gtapprox} {\mathrel{\spose{\lower 3pt\hbox{$\mathchar"218$}}
 \raise 2.0pt\hbox{$\mathchar"13E$}}}

\newcommand{\dist} {{  \rm dist }}

\UseSection   
\usepackage[usenames]{color}

\definecolor{at}{rgb}{0.0, 0.5, 0.0} 

\newcommand{\LT}{{\rm Loc}  }

\newcommand{\DV}{\Dcal}
\newcommand{\DVa}{\alpha}

\renewcommand{\to} {\rightarrow}

\newcommand{\R}{\Rbold}
\newcommand{\Z}{\Zbold}

\newcommand{\N}{\Nbold}
\newcommand{\C}{\mathbb{C}}
\newcommand{\volume}{\mathbb{V}}
\newcommand{\Lambdabold}{\boldsymbol{\Lambda}}

\newcommand{\1}{\mathbbm{1}}

\newcommand{\psib}{\bar\psi}

\newcommand{\jm}{j_\Omega}

\newcommand{\Mext}{M_\mathrm{ext}}

\newcommand{\Ex}{\mathbb{E}}

\newcommand{\chicCov}{{\chi}}

\newcommand{\bubble}{{\sf B}}

\newcommand{\pt}{{\rm pt}}

\newcommand{\Vpt}{V_{\rm pt}}

\newcommand{\xch}{\check{x}}
\newcommand{\Kch}{\check{K}}
\newcommand{\Vch}{\check{V}}
\newcommand{\Rch}{\check{R}}
\newcommand{\gch}{\check{g}}
\newcommand{\zch}{\check{z}}
\newcommand{\much}{\check{\mu}}

\newcommand{\I}{\mathfrak{I}}

\newcommand{\z}{z}

\newcommand{\gbar}{\bar{g}}

\newcommand{\ggen}{\tilde{g}}
\newcommand{\sgen}{\tilde{s}}
\newcommand{\chigen}{\tilde{\chi}}
\newcommand{\mgen}{\tilde{m}}
\newcommand{\Iint}{\mathbb{I}}
\newcommand{\Igen}{\tilde{\mathbb{I}}}

\newcommand{\zbar}{\bar{z}}
\newcommand{\mubar}{\bar{\mu}}

\newcommand{\domRG}{\mathbb{D}}
\newcommand{\domRGch}{\check{\mathbb{D}}}

\newcommand{\betamax}{\beta_{\rm max}}

\newcommand{\pp}{a}
\newcommand{\qq}{b}

\newcommand{\half}{\textstyle{\frac 12}}

\newcommand{\ddp}[2]{\frac{\partial #1}{\partial #2}}

\newcommand{\phib}{\bar\phi}

\newcommand{\Kspace}{\Kcal}

\ifdefined\macrosH
  \usepackage{xr-hyper}
  \usepackage{hyperref}
  \hypersetup{hypertexnames=false}
  \hypersetup{colorlinks,citecolor=blue,linkcolor=blue}  

\else\ifdefined\macrosHarxiv
  \usepackage{xr-hyper}
  \usepackage{hyperref}
  \hypersetup{hypertexnames=false}

\else
  \newcommand{\texorpdfstring}[2]{#1}
  \usepackage{xr}
\fi\fi

\title{
 Logarithmic correction for the susceptibility of the
 4-dimensional weakly self-avoiding walk:
 \\
 a renormalisation group analysis}

\author{
 Roland Bauerschmidt\thanks{School of Mathematics,
   Institute for Advanced Study,
   Einstein Drive,
   Princeton, NJ 08540 USA.
   E-mail: {\tt roland@bauerschmidt.ca}.},\;
 David C.\ Brydges\thanks{Department of Mathematics,
   University of British Columbia,
   Vancouver, BC, Canada V6T 1Z2.
   E-mail: {\tt db5d@math.ubc.ca}, {\tt slade@math.ubc.ca}.}\;
 and Gordon Slade$^\dagger$}

\date\version

\begin{document}

\maketitle

\begin{abstract}
  We prove that the susceptibility of the continuous-time weakly
  self-avoiding walk on $\Zd$, in the critical dimension $d=4$, has a
  logarithmic correction to mean-field scaling behaviour as the
  critical point is approached, with exponent $\frac{1}{4}$ for the
  logarithm.  The susceptibility has been well understood previously
  for dimensions $d \ge 5$ using the lace expansion, but the lace
  expansion does not apply when $d=4$.
  The proof begins by rewriting the walk two-point function as the
  two-point function of a supersymmetric field theory.
  The field theory is then analysed via a
  rigorous renormalisation group method developed in a companion
  series of papers.  By providing a setting where the methods of the
  companion papers are applied together, the proof also serves as an
  example of how to assemble the various ingredients of the general
  renormalisation group method in a coordinated manner.
\end{abstract}

\section{Introduction and main result}

The critical behaviour of the self-avoiding walk depends on the
spatial dimension $d$.  The upper critical dimension is 4, and for $d
\geq 5$ the lace expansion has been used to prove that the
self-avoiding walk is governed by the same critical exponents as the
simple random walk \cite{BS85,Hara08,HS92a,Slad06}.  In this paper, we
apply a rigorous renormalisation group analysis to study the
susceptibility of the weakly self-avoiding walk in the critical
dimension $d=4$.

\subsection{Continuous-time weakly self-avoiding walk}
\label{sec:ctwsaw}

Let $X$ be the continuous-time simple random walk on the integer
lattice $\Zd$, with $d > 0$.  In more detail, $X$ is the stochastic
process with right-continuous sample paths that takes its steps at the
times of the events of a rate-$2d$ Poisson process.  Steps are
independent both of the Poisson process and of all other steps, and
are taken uniformly at random to one of the $2d$ nearest neighbours of
the current position.
The \emph{intersection local time up to time $T$} is defined by
\begin{equation} \label{e:ITdef}
  I(T) = \int_0^T \!\! \int_0^T \1_{X(S_1) = X(S_2)} \; dS_1 \, dS_2
  =
  \sum_{x\in\Z^d} (L_T^x)^2,
\end{equation}
where $L_T^x = \int_0^T \1_{X(S)=x} \; dS$ is the local time of $X$ at
$x$ up to time $T$.  Let $E_a$ denote the expectation for the process
with $X(0)=a \in \Zd$.

Given $g>0$ and $a,b \in \Zd$,
the continuous-time weakly self-avoiding walk \emph{two-point function} is then defined by
\begin{equation}
  \label{e:Gwsaw}
  G_{g,\nu}(a,b)
  =
  \int_0^\infty
  E_{a} \left(
    e^{-g I(T)}
    \1_{X(T)=b} \right)
  e^{- \nu T}
  dT,
\end{equation}
where $\nu$ is a parameter (possibly negative) chosen such that the
integral converges.  In \eqref{e:Gwsaw}, self-intersections are
suppressed by the factor $e^{-gI(T)}$.
In the limit $g \to \infty$, if $\nu$ is simultaneously sent to
$-\infty$ in a suitable $g$-dependent manner, it is known that the
limit of the two-point function \refeq{Gwsaw} is a multiple of the
two-point function of the standard discrete-time strictly
self-avoiding walk \cite{BDS12}.  Our analysis is for small $g>0$; the
model we study is predicted to be in the same universality class as
the strictly self-avoiding walk for all $g>0$.

We set $c_T = c_{g,T} = E_{a}(e^{-gI(T)})$, and define the
\emph{susceptibility} by
\begin{equation}
  \label{e:suscept-def}
  \chi(g,\nu) =   \sum_{b\in \Z^d}  G_{g,\nu}(a,b)
  =
  \int_0^\infty
  c_{g,T}
  e^{- \nu T}
  dT.
\end{equation}
By translation-invariance of the simple random walk and of
\eqref{e:ITdef}, $c_T$ and $\chi$ are independent of the point $a \in
\Z^d$.  In Lemma~\ref{lem:csub}, we apply a standard subadditivity
argument to prove that for all dimensions $d>0$ there exists a
\emph{critical value} $\nu_c=\nu_c(d,g) \in (-\infty,0]$ such that
\begin{equation}
  \label{e:chi-nuc}
  \text{$\chi(g,\nu) < \infty$ \; if and only if \; $\nu > \nu_c$}.
\end{equation}

The rate of divergence of $\chi(g,\cdot)$ is characterised by the
critical exponent $\gamma$ (assuming it exists) by
\begin{equation}
  \label{e:gamdef}
  \chi(g,\nu) \sim A_g (\nu - \nu_c)^{-\gamma}
  \quad \text{as $\nu \downarrow \nu_c$},
\end{equation}
where $A_g$ and $\gamma$ are $d$-dependent constants.
Throughout the paper, we write $p \sim q$ for asymptotic formulas,
i.e., when $\lim p/q =1$.  The exponent $\gamma$ is predicted to be
universal, i.e., dependent on the dimension $d$, but otherwise
independent of fine details of the model.  For $d=4$, a universal
logarithmic correction to this scaling has been predicted, and our
main result gives a rigorous proof of this logarithmic correction.
Logarithmic corrections for the scaling behaviour of the weakly
self-avoiding walk in dimension $4$ have been computed in the physics
literature using nonrigorous renormalisation group arguments, e.g.,
\cite{BGZ73,Dupl86}.  Early indications of the critical nature of the
dimension $d=4$ were given in \cite{ACF83,BFF84}, following proofs of
triviality of the $\varphi^4$ field theory above dimension 4
\cite{Aize82,Froh82}.  For $d=4$, the weakly self-avoiding walk also
coincides with the discrete-space continuous-time Edwards model (see
\cite[Section~10.1]{MS93}).

For $d \ge 5$, it is known that \refeq{gamdef} holds with critical
exponent $\gamma =1$, for the weakly and \emph{strictly} self-avoiding
walk \cite{BS85,HS92a}.  For $d=3$, the problem remains completely
unsolved from a mathematical point of view; a recent numerical
estimate for $d=3$ is $\gamma \approx 1.157$ \cite{SBB11}.  For $d=2$,
it is predicted that $\gamma = \frac{11}{32}$ \cite{Nien82}, and
recent work suggests that the scaling behaviour can be described by
${\rm SLE}_{8/3}$ \cite{LSW04}, but the existence neither of critical
exponents nor the scaling limit has yet been proved.  The case of
$d=1$ is of interest for weakly self-avoiding walk, where a fairly
complete understanding has been obtained \cite{Holl09}.  For a recent
survey of mathematical results about the self-avoiding walk, see
\cite{BDGS12}.

\subsection{Main result}

Let $\Delta$ denote the lattice Laplacian defined by
$\Delta f(x) = \sum_{e: |e|=1} (f(x+e)-f(x))$ on $\Zd$.
The lattice Green function is defined, for $m^2 > 0$, by
\begin{equation}
  \label{e:lGf}
  C_{m^2}(x)
  =
  G_{0,m^2}(0,x)
  = (-\Delta+m^2)^{-1}_{0x}.
\end{equation}
The inverse is bounded in $l^2(\Z^d)$-sense for $m^2>0$, and the limit
$m^2 \downarrow 0$ exists (pointwise in $x$) if $d>2$.  The
\emph{bubble diagram} for simple random walk is the squared $\ell^2$
norm $B_{m^2}=\sum_{x \in \Zd} C_{m^2}(x)^2$.  It follows from the
definition that
\begin{equation}
  B_{m^2} = \int_0^\infty \!\!\!
  \int_0^\infty P(X(T) = Y(S)) e^{-m^2T} e^{-m^2 S} \; dT \, dS
  ,
\end{equation}
where $X$ and $Y$ are two independent simple random walks starting at
$0$.  Hence the bubble diagram measures the expected amount of time
that two independent simple random walks killed at rate $m^2$
intersect each other.  The bubble diagram arises in an important way
in our analysis, and it is convenient to define $\bubble_{m^2} =
8B_{m^2}$.
By Parseval's formula and elementary calculus, as $m^2 \downarrow 0$,
\begin{equation}
  \label{e:freebubble}
  \bubble_{m^2}
  =
  8B_{m^2}
  =
  8
  \int_{[-\pi,\pi]^d}
  \left|\frac{1}{4 \sum_{j=1}^{d} \sin^2 (\frac{k_j}{2}) +m^2}\right|^2
  \frac{dk}{(2\pi)^d}
  \sim \begin{cases}
    {\sf b}  \log m^{-2}  & (d=4)
    \\
    \bubble_0 & (d>4),
  \end{cases}
\end{equation}
with ${\sf b} = 1/(2\pi^2)$ and a $d$-dependent constant $\bubble_0
\in (0,\infty)$.
In particular, the expected time that two independent simple random
walks, without killing, spend intersecting each other is finite in
dimension $d>4$, but infinite in $d=4$.

\medskip

Our main result is the following theorem.

\begin{theorem} \label{thm:suscept}
  Let $d= 4$ and let $g > 0$ be sufficiently small.
  There exists $A_g  >0$ such that, as $\varepsilon \downarrow 0$,
  \begin{equation} \label{e:chieps-asympt}
    \chi(g,\nu_c + \varepsilon)
    \sim A_g
      \varepsilon^{-1} (\log \varepsilon^{-1})^{1/4} .
  \end{equation}
  As $g\downarrow 0$,
  \begin{equation}
    \label{e:cgasy}
    A_g = ({\sf b} g)^{1/4}(1+O(g)) .
  \end{equation}
\end{theorem}

Let ${\sf a} = 2C_0(0)$ where $C_0(0)>0$ is the expected total time
spent at the origin by the simple random walk, $C_0(0) = \int_0^\infty
P(X(T)=0) \, dT = E_0(\int_0^\infty \1_{X(T)=0} \,dT)$.  By an
elementary application of Jensen's inequality, we prove in
Lemma~\ref{lem:csub} that $\nu_c(g) \in [-{\sf a}g,0]$ for all $d >
2$.  As a corollary of the proof of Theorem~\ref{thm:suscept}, we
obtain the following asymptotic formula for the critical value.

\begin{theorem} \label{thm:nuc}
  Let $d=4$ and ${\sf a} = 2C_0(0)$.
  As $g \downarrow 0$,
  \begin{equation} \label{e:nucasy}
    \nu_c(g) = - {\sf a} g + O(g^2).
  \end{equation}
\end{theorem}

Our method of proof of Theorems~\ref{thm:suscept}--\ref{thm:nuc} is
based on a rigorous renormalisation group analysis, and applies more
widely.  In particular, it is used in \cite{BBS-saw4} to prove that
the critical two-point function $G_{g,\nu_c(g)}(0,x)$ is asymptotic to
a multiple of $|x|^{-2}$ as $|x|\to\infty$ in dimension $d= 4$.  Also,
work is in progress to extend our methods to study the weakly
self-avoiding walk with nearest-neighbour contact attraction
\cite{BBS-saw-sa} in dimension $d = 4$.  In \cite{BBS-phi4-log}, we
apply the renormalisation group method to study the critical behaviour
of the 4-dimensional $n$-component $|\varphi|^4$ spin model, for all
positive integers $n \ge 1$.  The existence of logarithmic corrections
to scaling for certain critical 4-dimensional polymer networks, and
for various critical correlation functions for the $|\varphi|^4$ model
is proved in \cite{ST-phi4}.

\subsection{Discussion}
\label{sec-intro}

Since ${\sf a} = 2\int_0^\infty P(X(T)=0) \, dT > 0$, \eqref{e:nucasy}
shows that $\nu_c(g) < 0$ for $g>0$.  In addition, $\nu_c(g) \to 0$ as
$g\downarrow 0$, as expected since $\nu_c(0)=0$ is the critical point
of the simple random walk.

The factor $\varepsilon^{-1}$ in \eqref{e:chieps-asympt} corresponds
to the linear divergence of the simple random walk susceptibility:
\begin{equation}
  \chi(0,m^2) = \sum_{x \in \Z^d} C_{m^2}(x) = m^{-2},
\end{equation}
while the logarithmic factor in \eqref{e:chieps-asympt} arises from
the logarithmic divergence in \eqref{e:freebubble}.
Note that $A_g$ tends to $0$ as $g\downarrow 0$, as expected since
there is no logarithmic correction for the simple random walk ($g=0$).
It has been observed that if the Fourier transform of the critical
two-point function is bounded by a multiple of $|k|^{-2}$ as $k \to 0$
(as is known for $d \ge 5$ \cite{HS92a} and as is predicted to be true
in all dimensions), then the susceptibility can have \emph{at most} a
logarithmic correction for $d=4$ (see \cite{BFF84} and
\cite[Theorem~1.5.4]{MS93}).  The exponent $\frac{1}{4}$ in
\eqref{e:chieps-asympt} is predicted to be universal for models of
self-avoiding walk in four dimensions. In particular, it is predicted
to be the same for the usual discrete-time \emph{strictly}
self-avoiding walk \cite{MS93}.

For $d=4$, Theorem~\ref{thm:suscept} and a standard Tauberian theorem
\cite[Chapter~XIII]{Fell71} imply that
\begin{equation} \label{e:ct}
  \frac{1}{T}\int_0^T c_S e^{\nu_cS} \; dS
  \sim A_g (\log T)^{1/4}
  \quad (T\to \infty).
\end{equation}
It is believed that \eqref{e:ct} remains true without Ces\`aro average
in $T$, i.e., that
\begin{equation} \label{e:ct-noaverage}
  c_T
  \sim A_g e^{-\nu_c T}   (\log T)^{1/4}
  \quad (T\to\infty),
\end{equation}
but our present estimates do not suffice to prove \eqref{e:ct-noaverage}.
Furthermore, denoting by $E^{g,T}_a$ the expectation for the measure
of weakly self-avoiding walks of length $T$, i.e.,
\begin{equation}
  E^{g,T}_a(F(X)) = \frac{E_a(e^{-gI(T)} F(X))}{E_a(e^{-gI(T)})}
  ,
\end{equation}
it is believed that for $p \ge 1$,
\begin{equation}
  \label{e:EXtp}
  \left(E^{g,T}_0|X(T)|^p\right)^{1/p} \sim c_{g,p}\, T^{1/2} (\log T)^{1/8}
  \quad (T\to \infty),
\end{equation}
and that $(\lambda^{-\frac{1}{2}} (\log \lambda)^{-\frac{1}{8}}
X(\lambda T))_{T\geq 0}$ converges as a process to a multiple of
Brownian motion as $\lambda\to\infty$.
(Such convergence is known for
the $4$-dimensional loop-erased random walk, with exponent
$\frac{1}{6}$ rather than $\frac 18$ for the logarithm
\cite{Lawl86,Lawl95}.)

Renormalisation group methods have been used previously to study
weakly self-avoiding walk on a 4-dimensional \emph{hierarchical}
lattice.  The continuous-time version of the model has been studied in
the series of papers \cite{BEI92,GI95,BI03c,BI03d}.  In particular, in
\cite{BI03d}, the results of \cite{BI03c} were extended to prove a
result analogous to \eqref{e:EXtp} for the $4$-dimensional
hierarchical lattice.  This was achieved by a contour integral
analysis of $G_\nu(x)$, with $\nu$ \emph{complex}.  Our methods would
require further development to follow the same procedure for $\Z^4$.
More recently, a completely different renormalisation group approach
to the discrete-time weakly self-avoiding walk on a 4-dimensional
hierarchical lattice has been developed in \cite{Ohno13}.  A variant
of the 4-dimensional Edwards model was analysed in \cite{IM94} using a
renormalisation group method; this variant is not a model of walks
taking steps in a lattice, but it is presumably in the same
universality class as the 4-dimensional self-avoiding walk, 
and the results of \cite{IM94} bear some relation to our results.
Some steps towards an understanding of the behaviour in dimension
$d=4-\epsilon$ are taken in \cite{MS08} (the work of \cite{MS08} is
formulated in dimension $3$ but it mimics the behaviour of the
nearest-neighbour model in dimension $4-\epsilon$).

Our renormalisation group analysis has grown out of the methods of
\cite{BEI92,BI03d}, but in a much extended and generalised form.
It is based on an exact functional integral
representation of the two-point function of the continuous-time
self-avoiding walk as the two-point function of a
supersymmetric quantum field theory,
containing both boson and fermion fields.  Such integral representations
are summarised in \cite{BIS09}.  These representations
are inspired by the observation of de Gennes \cite{Genn72} that the
self-avoiding walk problem can be regarded as the $n=0$ limit of the
$n$-vector model (see also \cite[Section~2.3]{MS93}).  The basic
observation of de Gennes was that in a random walk representation of
the $n$-vector model every closed loop contributes a factor $n$.  When
$n=0$, closed loops do not contribute, leading to self-avoiding walks.
The $n$-vector model is closely related to the $n$-component $|\varphi|^4$
model.  For $n=1$, the critical 4-dimensional $\varphi^4$ model was
analysed using block spin renormalisation in \cite{GK85,GK86},
and via partially renormalised phase space expansion in \cite{FMRS87}.
In both approaches, the
critical two-point function was controlled.
Block spin methods were extended from the
critical point to its neighbourhood in
\cite{Hara87,HT87}, where logarithmic corrections for the
susceptibility and correlation length were derived for the
4-dimensional one-component $\varphi^4$ model
(in particular, the susceptibility has
exponent $\frac{1}{3}$ for the logarithm).  However, it is not clear
how to prove theorems about the scaling limit of the self-avoiding walk,
in a rigorous mathematical sense, via an analysis of an $n\to 0$ limit
of the $n$-vector model or $n$-component $|\varphi|^4$
theory.

On the other hand, the notion was developed in
\cite{BIS09,LeJa87,Lutt83,McKa80,PS80} that while an $n$-component
{\em boson} field $\phi$ associates $n$ to each closed loop, an
$n$-component {\em fermion} field $\psi$ associates $-n$.
With both fields present, the net
effect is to associate zero to each closed loop.  This provides a way
to realise de Gennes' idea, without any nonrigorous limit.  Moreover,
it was also understood that the fermion field can be regarded as
nothing more than the differential of the boson field, with the
anticommuting nature of fermions being represented by anticommuting
differential forms.  A representation of the self-avoiding walk
two-point function as the two-point function of a supersymmetric field
theory, sometimes referred to as the $\tau$-isomorphism, is central to
the analysis of \cite{BEI92,BI03c,BI03d}.  A self-contained derivation
of this integral representation is given in \cite{BIS09}, both for
weakly self-avoiding walk in continuous time and for strictly
self-avoiding walk in discrete time.

We use
the integral representation to rewrite the two-point function
of the continuous-time weakly self-avoiding walk as the two-point
function of a supersymmetric field theory, and apply a rigorous
renormalisation group argument to analyse the field theory.  Key
steps in the method are developed in the series of papers
\cite{BS11,BBS-rg-flow,BBS-rg-pt,BS-rg-norm,BS-rg-loc,BS-rg-IE,BS-rg-step}.
In the present paper, we rely heavily on those developments and show
how they can be combined to analyse the susceptibility
in
the critical dimension $d = 4$.
The proof here is thus not self-contained, but rather
reveals its general structure, with reliance on
substantial details obtained elsewhere.

In \cite{BBS-phi4-log}, we extend our analysis to the $n$-component
$|\varphi|^4$ model.  Among other results, we prove that its
susceptibility obeys
\begin{equation}
  \label{e:chin}
  \chi(n,g,\nu_c+\varepsilon)
  \sim A_{g,n} \varepsilon^{-1}(\log\varepsilon^{-1})^{(n+2)/(n+8)}.
\end{equation}
This confirms the logarithmic correction with exponent
$\frac{n+2}{n+8}$ predicted for $n \ge 1$ in \cite{LK69,BGZ73,WR73},
and generalises the rigorous result of \cite{Hara87,HT87} for $n=1$.
Setting $n=0$ in \refeq{chin} gives a formula consistent with
\refeq{chieps-asympt}, but our proof for the weakly self-avoiding walk
does not use any non-rigorous $n \to 0$ limit.

\subsection{Organisation}
\label{sec:org}

The proof of Theorem~\ref{thm:suscept}--\ref{thm:nuc} is divided into sections,
as follows.

In Section~\ref{sec:fv}, we show that the susceptibility of the weakly
self-avoiding walk on $\Z^d$ is well approximated by replacing $\Z^d$
by a sequence of finite tori $\Lambda_N = \Z^d/L^N\Z^d$, as
$N\to\infty$ with a fixed integer $L>1$.

In Section~\ref{sec:intrep}, we explain how the two-point function of
weakly self-avoiding walk on the finite set $\Lambda_N$ can be
represented as the two-point function of a supersymmetric field theory
on $\Lambda_N$.  The latter is a certain integral over differential
forms on the finite-dimensional linear manifold $\C^{\Lambda_N}$, and
has an interpretation as a Gaussian \emph{super-expectation}.  From
that point onward, we no longer consider walks, and instead focus
attention on the study of such integrals.

In Section~\ref{sec:chvar}, we express the two-point function of the
weakly self-avoiding walk in terms of the two-point function of a
simple random walk with \emph{renormalised} parameters.  We show that
Theorem~\ref{thm:suscept} follows if these renormalised parameters can
be chosen suitably.

In Section~\ref{sec:rg}, we explain how the integrals of
Section~\ref{sec:intrep} can be evaluated progressively.  This
progressive integration is the starting point for a multiscale
analysis and defines the \emph{renormalisation group map}, which we
parametrise by coordinates $(V,K)$.  The $V$-coordinate is
$3$-dimensional and describes the important (``relevant'' and
``marginal'') directions of the renormalisation group map.  The
$K$-coordinate is infinite dimensional and complements $V$ to a
complete description of the problem.  Thus the leading contributions
in \eqref{e:chieps-asympt}--\eqref{e:nucasy} will be determined by
$V$, but the control of $K$ is at the heart of obtaining a
mathematically rigorous result.

In Section~\ref{sec:rg-map}, we consider the definition and properties
of a single application of the renormalisation group map, by making
extensive use of results developed in the companion series of papers,
especially \cite{BBS-rg-pt,BS-rg-step}.  This includes the definition
of an infinite dimensional dynamical system which describes the
infinite volume limit $\Lambda_N \uparrow \Z^d$.

In Section~\ref{sec:bulk-flow}, we study the evolution of the
renormalisation group map, by applying a result of \cite{BBS-rg-flow}
concerning dynamical systems, together with the estimates given in
Section~\ref{sec:rg-map}.  We analyse the stability of the
renormalisation group map near the Gaussian fixed point corresponding
to the simple random walk $g=0$.  More precisely, we construct a
centre stable manifold near this fixed point.

In Section~\ref{sec:pfmr}, we prove
Theorems~\ref{thm:suscept}--\ref{thm:nuc} using the results of
Sections~\ref{sec:fv}--\ref{sec:bulk-flow}.  The centre stable
manifold constructed in Section~\ref{sec:bulk-flow} plays a crucial
role in the identification of the critical point.  To deduce the
logarithmic correction for the susceptibility for $d=4$, we study
infinitesimal deviations of the flow from the centre stable manifold,
which allows the derivative of $\chi$ with respect to $\nu$ to be
computed.  An elementary but important technical aspect is that the
results of the previous sections permit us to focus our attention on
this derivative for the case of finite volume $\Lambda_N$.

Throughout the remainder of this paper, we emphasise the source of the
power of the logarithm in \eqref{e:chieps-asympt} by writing
\begin{equation}
  \label{e:gamconv}
  \gamma = \frac{1}{4}.
\end{equation}
Thus $\gamma$ does not denote the critical exponent as in
\eqref{e:gamdef}, but rather the exponent of the logarithmic
correction to the critical exponent~$1$.

\section{Finite volume approximation}
\label{sec:fv}

The proof begins by approximation of $\Z^d$ by a sequence of finite
tori of period $L^N$, which we denote by $\Lambda=\Lambda_N = \Z^d/L^N\Z^d$.
Let $E^{\Lambda_N}_a$ denote the expectation of the continuous-time
simple random walk on $\Lambda_N$, started from $a \in \Lambda_N$.
Let
\begin{equation}
  \lbeq{cNTdef}
  c_{N,T}(a,b) = E^{\Lambda_N}_a(e^{-gI(T)}\1_{X(T)=b})
  ,
  \quad\quad
  c_{N,T} = E^{\Lambda_N}_a(e^{-gI(T)}).
\end{equation}
We define the two-point function for the torus by
\begin{equation}
  G_{N,\nu}(a,b) = \int_0^\infty c_{N,T}(a,b) e^{-\nu T} \; dT,
\end{equation}
and define the corresponding torus susceptibility by
\begin{equation}
  \chi_N(\nu)
  =
  \sum_{b\in\Lambda_N} G_{N,\nu}(a,b) = \int_0^\infty c_{N,T} e^{-\nu T} \; dT
  .
\end{equation}
By the Cauchy--Schwarz inequality,
$T= \sum_{x\in \Lambda} L_T^x \le (|\Lambda|I(T))^{1/2}$, and hence
\begin{equation}
  \chi_N(\nu)
  \leq
  \int_0^\infty e^{-gT^2/|\Lambda_N|} e^{-\nu T} \; dT < \infty
  \quad
  \text{for all $\nu \in \R$.}
\end{equation}
The following lemma shows that $\chi_N$ is a good approximation to
$\chi$.

\begin{lemma} \label{lem:suscept-finvol}
  Let $d>0$.
  For all $\nu \in \R$, $\chi_N(\nu)$ is non-decreasing
  in $N$, and obeys $\lim_{N\to\infty}\chi_N(\nu)=   \chi(\nu)$
  (with $\chi(\nu)=\infty$ for $\nu \le \nu_c$).
  The functions
  $\chi_N$ and $\chi$ are analytic on $\{\nu \in \C :{\mathrm{Re}}\nu > \nu_c\}$,
  and $\chi_N$ and all its derivatives
  converge uniformly on compact subsets of ${\mathrm{Re}}\nu > \nu_c$ to $\chi$ and its derivatives.
\end{lemma}

\begin{proof}
  Let $c_{N,T} = E^{\Lambda_N}(e^{-gI(T)})$ as in \refeq{cNTdef},
  and let $c_T = E(e^{-gI(T)})$ as in \refeq{suscept-def}.
  We first show that
  \begin{equation}
  \lbeq{ctmon}
    c_{N,T} \leq c_{N+1,T} \leq c_T
    ,\qquad
    \lim_{N \to \infty} c_{N,T} =c_T
  \end{equation}
  (for the inequality we assume $L^N \ge 3$).

  To see this, we observe that there is a one-to-one correspondence
  between nearest-neighbour walks on $\Z^d$ started at the origin and
  such walks on the finite torus $\Z_n^d$ if $n \geq 3$, by folding a
  walk on $\Z^d$ (the image under the canonical projection $\Z^d
  \twoheadrightarrow \Z_n^d$), and corresponding unfolding of walks on
  $\Z_n^d$ (unique for the nearest-neighbour walk when $n \geq 3$).
  Given a walk $X$ on $\Z^d$ starting at $0$, we denote the folded
  walk on $\Z_n^d$ by $X^n$.  This allows us to couple walks on tori
  of different diameter via the distribution of walks on $\Zd$.
  We denote the local time of a walk $X$ up to time $T$ by $L^x_T(X) =
  \int_0^T \1_{X(S)=x} \; dS$, and similarly the intersection local
  time by $I_T(X)$.
  Given $X$ and a positive integer $k$, we obtain
  \begin{align}
    I_T(X^{kn}) & = \sum_{x\in\Z_{kn}^d} \left(L^{x}_T(X^{kn}) \right)^2
    = \sum_{x\in\Z_n^d} \sum_{y\in\Z^d:  \|y\|_\infty < k}
    \left(L^{x+yn}_T(X^{kn}) \right)^2
    \nnb &
    \leq \sum_{x\in\Z_n^d}
    \left(\sum_{y\in\Z^d:  \|y\|_\infty < k} L^{x+yn}_T(X^{kn}) \right)^2
    = \sum_{x\in\Z_n^d}
    \left(L^{x}_T(X^{n})\right)^2
    = I_T(X^n).
  \end{align}
  Thus, with $n=L^N$, $k=L$, and with $X$ fixed,
  \begin{equation} \label{e:EI-cond}
    e^{-gI_{N+1,T}(X)}   \geq e^{-gI_{N,T}(X)}.
  \end{equation}
  Now we take the expectation over $X$ to obtain the first inequality
  of \refeq{ctmon}.
  This shows monotonicity in $N$ of $c_{N,T}$.
  Also, a folded walk can only have more intersections than its
  unfolding, so $I_T(X^n) \ge I_T(X)$ for any walk $X$ on
  $\Zd$ and for any $n$, so we also have $c_{N,T} \le c_{N+1,T} \le c_T$.

  For the convergence of $c_{T,N}$ to $c_T$, we first note that
  walks which do not reach distance $\frac 12 L^N$ from the origin
  do not contribute to the difference.
  Walks which do must take at least $\frac 12 L^N$ steps.  Therefore, since
  $I(T) \geq 0$, we conclude that $|c_T-c_{T,N}| \leq 2\Pbold(M_T > \frac 12 L^N)$,
  where $M_T$ is a rate-$2d$ Poisson process.
  The probability in the upper bound goes to zero as $N \to\infty$, so
  $\lim_{N\to\infty}c_{N,T} = c_T$, and \refeq{ctmon} is proved.

  With the monotone convergence theorem, \refeq{ctmon} gives
  \begin{equation}
    \chi(\nu)
    = \int_0^\infty \lim_{N\to\infty} c_{N,T} e^{-\nu T} \; dT
    = \lim_{N\to\infty} \chi_N(\nu)
    \quad  \text{for $\nu \in \R$}
    ,
  \end{equation}
  where both sides are finite if and only if $\nu > \nu_c$.
  Also,
  since
  $|c_{N,T} e^{-\nu T}| \le c_{N,T} e^{-({\rm Re}\nu) T} \le
  c_{T} e^{-({\rm Re}\nu) T}$,
  it follows from the
  dominated convergence theorem that
  \begin{equation}
  \chi(\nu) = \lim_{N\to\infty}\chi_N(\nu)
  \quad
  \text{for ${\rm Re} \nu > \nu_c$}.
  \end{equation}
  The analyticity of $\chi$ and $\chi_N$ follows from analyticity of
  Laplace transforms, 
  and the desired compact convergence of $\chi_N$ and all its derivatives
  then follows from Montel's theorem. 
\end{proof}

\section{Integral representation}
\label{sec:intrep}

The next step in the proof is to represent the two-point function for
walks on the finite set $\Lambda$ by an integral over the finite
dimensional space $\C^\Lambda$.  Before entering into the details of
the representation, we first briefly recall some basic facts about
integration of differential forms.

\subsection{Integration of differential forms}
\label{sec:intforms}

Let $\Lambda = \{1,\ldots, M\}$ be a finite set (e.g., a discrete torus)
of cardinality $M$.  Let $u_1,v_1,$ \ldots, $u_M,v_M$ be standard
coordinates on $\R^{2M}$. Then $du_1 \wedge dv_1 \wedge \cdots
\wedge du_M \wedge dv_M$ is the standard volume form on $\R^{2M}$,
where $\wedge$ denotes the usual anticommuting wedge product (see
\cite[Chapter~10]{Rudi76} for an introduction).
In the following, we drop the
wedge from the notation and write simply $du_idv_j$ in place of $du_i
\wedge dv_j$.  The one-forms $du_i$, $dv_j$ generate the Grassmann
algebra of differential forms on $\R^{2M}$.
A form which is a function of $u,v$ times a product of $p$ differentials is
said to have \emph{degree} $p$, for $p \ge 0$.
A sum of differential forms of even degree is called \emph{even}.
We will use the term \emph{form} as
an abbreviation for ``differential form.''

Any form $F$ of degree $2M$ can be written as $F=f(u,v)du_1dv_1\cdots
du_M dv_M$, and its integral is defined by
\begin{equation}
    \int F = \int_{\R^{2M}} f(u,v)du_1dv_1\cdots du_M dv_M,
\end{equation}
where the right-hand side is the usual Lebesgue integral of $f$ over
$\R^{2M}$.  For $p \ne 2M$, we define the integral over $\R^{2M}$ of a form of
degree $p$ to be zero.  Then the integral of an arbitrary
form, which is a linear combination of forms of degree $p$ for
different values of $p$, is defined by linearity; only its component
of top degree affects the value of the integral.

We introduce complex coordinates by setting $\phi_x = u_x + i v_x$,
$\phib_x = u_x-iv_x$ and $d\phi_x = du_x+idv_x$, $d\phib_x =
du_x-idv_x$, for $x \in \Lambda$.  Since the wedge product is
anticommutative, the following pairs all anticommute for every $x,y\in
\Lambda$: $d\phi_x$ and $d\phi_y$, $d\phib_x$ and $d\phi_y$,
$d\phib_x$ and $d\phib_y$.  In addition,
\begin{equation}
\label{e:duv}
    d\phib_x d\phi_x = 2i du_x dv_x.
\end{equation}

The integral of a function $f(\phi,\phib)$ (i.e., a \emph{$0$-form})
with respect to $\prod_{x\in \Lambda}d\phib_xd\phi_x$ is thus given by
$(2i)^M$ times the integral of $f(u+iv,u-iv)$ over $\R^{2M}$.  Note
that the product here can be taken in any order, since each factor
$d\phib_xd\phi_x$ has even degree (namely degree two).
It is convenient for us to define
\begin{equation}
\label{e:phipsi}
    \psi_x = \frac{1}{\sqrt{2\pi i}} d \phi_x,
    \quad
    \psib_x = \frac{1}{\sqrt{2\pi i}} d \phib_x,
\end{equation}
where we fix a choice of the square root and use this choice
henceforth.  Then
\begin{equation}
    \psib_x \psi_x =\frac{1}{2\pi i} d\phib_x d\phi_x = \frac{1}{\pi}du_x dv_x.
\end{equation}

By definition, a differential form $F$ can be written as
\begin{equation} \label{e:psipsib}
  \sum_{x,y} F_{x,y}(\phi,\bar\phi) \psi^x \bar\psi^{y}
  ,
\end{equation}
where the sum is over sequences $x$ and $y$ in $\Lambda$ (of any length)
and $\psi^x = \psi_{x_1}\cdots \psi_{x_p}$ for $x=(x_1,\ldots,x_p) \in\Lambda^p$.
Given a subset $X \subset \Lambda$, we denote by $\Ncal(X)$ the
algebra of \emph{even} forms $F$ with the restriction that the sum in \refeq{psipsib}
extends only over sequences in $X$ and the coefficients $F_{x,y}$ depend only
on $(\phi_x,\bar\phi_x)_{x\in X}$.  The algebra $\Ncal(X)$ is
commutative and associative.  For the special case $X=\Lambda$, we write
simply $\Ncal=\Ncal(X)$.
Later we discuss norms on $\Ncal$ that impose regularity conditions on its elements.

\subsection{Functions of forms}
\label{sec:formsfunc}

We refer to the variables $(\phi_{x})$ and the forms $(\psi_{x})$
in \refeq{psipsib}
as \emph{boson} and \emph{fermion} fields, respectively.
We view $\psi$ as an anticommuting analogue of $\phi$
and think of a differential form as a function of $\phi$ and $\psi$.
Differential forms have degree at most $2M$ and therefore, as
``functions'' of $\psi$, they are polynomial with degree at most $2M$.
We use terminology corresponding to this view of differential forms.
For example, by setting $\psib = \psi = 0$ in a form we mean to take
the projection to its degree-$0$ part.  In our context, the forms
$\psi, \bar\psi$ are often called \emph{Grassmann variables}, and
integrals of ``functions'' of $\psi$ in the standard sense of
differential forms, as explained in Section~\ref{sec:intforms}, are
the same as the Berezin integral \cite{Bere66,LeJa87}.

We also define functions of even forms.
Let $F=(F_j)_{j\in J}$ be a finite collection of even forms.
We say that $F$ is {\em even}.  Let $F_j^{0}$ denote the degree-$0$
part of $F_j$.  Given a $C^\infty$ function $f : \R^{J} \to\mathbb C$ we
define a form $f(F)$ by its Taylor series about the degree-$0$ part
of $F$, i.e.,
\begin{equation}
  \label{e:Fdef}
  f(F) = \sum_{\alpha} \frac{1}{\alpha !}
  f^{(\alpha)}(F^{0})
  (F - F^{0})^{\alpha}
\end{equation}
where $\alpha = (\alpha_j)_{j\in J}$ is a multi-index, with $\alpha ! = \prod_{j\in
J}\alpha_j !$, and $(F - F^{0})^{\alpha} =\prod_{j\in J} (F_{j} -
F_{j}^{0})^{\alpha_{j}}$.  Note that the summation terminates as
soon as $\sum_{j\in J}\alpha_j=M$ since higher order forms vanish, and
that the order of the product on the right-hand side is immaterial
when $F$ is even.  For example,
\begin{equation}
  e^{-\phi_{x}\phib_{x}
    - \psi_{x} \bar\psi_{x}}
  =  e^{-\phi_{x}\phib_{x}}
  \left(1 -  \psi_{x} \bar\psi_{x}\right).
\end{equation}

\subsection{Identity for two-point function}
\label{sec:Grep}

Let $\Lambda = \Lambda_N$ now denote the discrete torus in $\Zd$, as in Section~\ref{sec:fv}.
For $x \in \Lambda$, we define the forms
\begin{equation} \label{e:taudef}
  \tau_x
  = \phi_x \bar\phi_x
  + \psi_x  \wedge \bar\psi_x
  ,
\end{equation}
and
\begin{equation}
  \label{e:addDelta}
  \tau_{\Delta,x}
  =
  \frac 12 \Big(
  \phi_{x} (- \Delta \bar{\phi})_{x} + (- \Delta \phi)_{x} \bar{\phi}_{x} +
  \psi_{x} \wedge (- \Delta \bar{\psi})_{x} + (- \Delta \psi)_{x} \wedge \bar{\psi}_{x}
  \Big),
\end{equation}
where $\Delta$ is the lattice Laplacian on $\Lambda$ given by $\Delta
\phi_{x} = \sum_{e: |e|=1} (\phi_{x+e} - \phi_{x})$.
Atypically, we have left the wedge product $\wedge$ explicit in the above
definitions, for emphasis.  The following proposition is proved in
\cite{BEI92,BI03d}; see \cite[Theorem~5.1]{BIS09} for a
self-contained proof.  The integrand and integral on the right-hand
side of \eqref{e:Grep1} are as defined in Section~\ref{sec:intforms}.

\begin{prop}
  \label{prop:Grep}
  Let $d>0$.
  For $g>0$ and $\nu\in \R$, or $g = 0$ and $\nu > 0$,
  \begin{equation}
    \label{e:Grep1}
    G_{N,g,\nu}(a,b)
    =
    \int
    e^{-\sum_{x\in\Lambda}(\tau_{\Delta ,x} + g \tau_x^2 + \nu \tau_x)}
    \bar\phi_{\pp} \phi_{\qq}
    .
  \end{equation}
\end{prop}

Note that, by summation by parts on the torus $\Lambda$,
\begin{equation}
  \sum_{x\in\Lambda}\tau_{\Delta ,x}
  =
  \sum_{x\in\Lambda} \Big(
  \phi_{x} (- \Delta \bar{\phi})_{x} +
  \psi_{x} \wedge (- \Delta \bar{\psi})_{x}
  \Big)
  ,
\end{equation}
so the symmetrisation in \eqref{e:addDelta} is not needed for \eqref{e:Grep1}
(it is absent also in \cite[Theorem~5.1]{BIS09}).
However, later we need sums $\sum_{x\in X} \tau_{\Delta,x}$
with $X \subset \Lambda$ a proper subset, and there symmetrisation is important.

The right-hand side of \refeq{Grep1}
is the two-point function of a \emph{supersymmetric}
field theory with quartic self-interaction.
The partition function that results by dropping the factor $\bar\phi_{\pp}
\phi_{\qq}$ on the right-hand side of \eqref{e:Grep1} turns out to be
simply equal to $1$.  More generally, there is the remarkable
self--normalisation property that
\begin{equation}
  \label{e:self-norm}
  \int
  e^{
    -\sum_{x\in\Lambda}(p_{x}\tau_{\Delta ,x} + q_{x} \tau_x^2 + r_{x} \tau_x)
  }
  =
  1
  \quad\quad \text{for all  }
  p_{x}\ge 0,\, q_{x}> 0,\, r_{x} \in \Rbold
  ,
\end{equation}
as a result of supersymmetry.  Thus supersymmetry provides, in particular,
the simplification that
there is no need to conduct any analysis of a partition function.
A self-contained proof of \eqref{e:self-norm} and generalisations,
as well as a brief discussion of supersymmetry,
can be found in \cite{BIS09}.

In the remainder of this section, we discuss some of the
intuition associated to \eqref{e:Grep1}, but we do not make use of it afterwards.
Since $\tau^2 = |\phi|^4 + 2|\phi|^2|\psi|^2$,
the fermionic part of the right-hand side of \refeq{Grep1} is
\begin{equation}
  \lbeq{fermexp}
  e^{- \sum_{x,y\in \Lambda} \psi_x A_{xy}\bar\psi_y}
  \quad
  \text{with $A_{xy}= -\Delta_{xy}  + (\nu + 2g \phi_x\bar\phi_x )\delta_{xy} $}.
\end{equation}
One way to evaluate the integral on the right-hand side of \refeq{Grep1} is to expand
the exponential \refeq{fermexp}, keep only the top-degree part, and rearrange the order of the
differentials to obtain the standard volume form.  The anti-commutativity produces
a determinant, with the result that the top-degree part of
$e^{-\sum_{x\in\Lambda}(\tau_{\Delta ,x} + g \tau_x^2 + \nu \tau_x)}$ is equal to
\begin{equation}
  \lbeq{fermdet}
  \det(-\Delta + \nu + 2g|\phi|^2)
  e^{-\sum_{x\in\Lambda} (\phi_x(-\Delta \bar\phi)_x+g|\phi_x|^4+\nu|\phi_x|^2)}
  \prod_{x\in\Lambda} \frac{d\bar\phi_x \, d\phi_x}{2\pi i}.
\end{equation}
Thus there is an alternate version of the identity \refeq{Grep1}, in which the
exponential on its right-hand side is replaced by \refeq{fermdet}.
This procedure of integrating out of fermions has been useful in the study of other supersymmetric models
\cite{DS10,DSZ10},
but we do not follow this approach.  Instead, we work directly with the right-hand
side of \refeq{Grep1} throughout
the entire analysis.

If we drop the $\psi$ terms in the definitions of $\tau, \tau_{\Delta}$ in
\eqref{e:taudef}--\eqref{e:addDelta} (or equivalently drop the determinant
in \refeq{fermdet}), consider a real-valued
boson field $\varphi_x$ rather than a complex-valued one,
and integrate with respect to the
Lebesgue measure $\prod_{x\in\Lambda} d\varphi_x $ on $\R^\Lambda$,
the right-hand side of \eqref{e:Grep1} becomes the integral
\begin{equation}
  \lbeq{2comp}
  \int_{\R^\Lambda}
  \varphi_{\pp} \varphi_{\qq}
  e^{-\sum_{x\in\Lambda}(\frac12\varphi_x(-\Delta \varphi)_x + \frac14 g \varphi_x^4 + \frac12 \nu \varphi_x^2)}
  \prod_{x\in\Lambda} d\varphi_x
\end{equation}
which is the unnormalised two-point function of the $\varphi^4$-model.
The $\varphi^4$-model is a spin system in which each point
$x\in\Lambda$ carries a spin variable $\varphi_x\in \R$.  It is a
perturbation of the massless Gaussian free field, in which each spin
has a factor $e^{-U(\varphi_x)}$ with $U(t)=\frac14 gt^4+\frac12\nu
t^2$.  The term $\sum_{x\in\Lambda}\frac12 \varphi_x(-\Delta
\varphi)_x$ in the exponent can also be written as a sum of squares of
the gradient of $\varphi_x$, so fields with small gradient receive
larger weight, which is a ferromagnetic interaction. Completing the
square gives
\begin{equation}
  U(t) = \frac14 g\left(t^2 + \frac{\nu}{g} \right)^2 - \frac{\nu^2}{4g}.
\end{equation}
The term $-\nu^2/4g$ on the right-hand side cancels upon normalisation
of the integral \refeq{2comp} and is unimportant.  As in
Theorem~\ref{thm:nuc}, we are interested in negative values of $\nu$,
so that $U(t)$ has a double well shape, with minima at $t = \pm
\sqrt{|\nu|/g}$ separated by a barrier of height $\nu^2/g$.  As $\nu$
becomes increasingly negative, the two wells become increasingly deep
and increasingly separated, so it is plausible that there is a
critical value of $\nu$ at which long range correlations develop due
to the values of $\varphi_x$ becoming concentrated at the minima of
$U$.  The behaviour in dimension $d=4$ at and near this critical point
is the subject of the renormalisation group analyses of
\cite{FMRS87,GK85,GK86,Hara87,HT87}.  Although our method of proof is
not based on the above picture, the picture provide some intuition.
Our method is applied to the $n$-component $|\varphi|^4$ model in
\cite{BBS-phi4-log,ST-phi4}.

\subsection{Gaussian super-expectation}
\label{sec:gauss-sub}

The right-hand side of \eqref{e:Grep1} is an instance of a
\emph{Gaussian super-expectation}, defined as follows.
First, given a $\Lambda \times \Lambda$ matrix $A$, let
\begin{equation}
    S_{A} (\Lambda )
    =
    \sum_{x,y \in \Lambda}\Big(
    \phi_{x} A_{xy} \phib_{y} +
    \psi_{x} A_{xy} \psib_{y}
    \Big)
    .
\end{equation}

\begin{defn}
  \label{def:ExC}
  Given a positive-definite $\Lambda \times \Lambda$ matrix $C$, let $A=C^{-1}$.
  Given a form $F$ on $\C^\Lambda$, we define the
  \emph{Gaussian super-expectation} of $F$, with covariance $C$, by
  \begin{equation}
    \label{e:Ebolddef}
    \Ex_{C} F
    =
    \int
    e^{-S_A (\Lambda)} F
  \end{equation}
  where the integral on the right-hand side is defined as
  in Section~\ref{sec:intforms}.
\end{defn}

In the special case that $f$ is a
$0$-form, or in other words $f$ is a function of $(\phi,\bar\phi)$, $\Ebold_{C}f$
is equal to the standard Gaussian expectation for a complex-valued random
field $\phi$ with covariance $C$ (see \cite[Proposition~4.1]{BIS09}), i.e.,
\begin{equation}
  \label{e:intzeroform}
  \Ex_{C}f
  =
  \int_{\C^\Lambda} f(\phi,\bar\phi)
  d\mu_C(\phi,\bar\phi),
  \quad
  d\mu_C(\phi,\bar\phi)
  = Z_C^{-1} e^{-\sum_{x,y\in \Lambda}\phi_x C^{-1}_{xy}\bar\phi_y}
  \prod_{x \in \Lambda} d\bar\phi_x
  d\phi_x ,
\end{equation}
where $Z_C$ is a normalisation constant such that $\Ex_C 1 =1$.
In particular, $\Ebold_{C} \phib_{a} \phi_{b} = C_{ab}$, i.e.,
$C$ is the \emph{covariance} of the Gaussian measure
$\mu_C$. For a general differential form, $\Ex_C F$
loses the probabilistic interpretation, but nevertheless many of the
important properties of Gaussian integrals continue to hold for the Gaussian
super-expectation. For example,
$C$ is also the ``covariance'' of $\psi$,
i.e., $\Ex_C\bar\psi_a\psi_b = - \Ex_C\psi_b\bar\psi_a = C_{ab}$.
Also, as a result of supersymmetry and in contrast to the usual Gaussian
measure $\mu_C$ in \eqref{e:intzeroform},
there is no normalisation constant in \eqref{e:Ebolddef}.
Further discussion of the Gaussian super-expectation is given in
Section~\ref{sec:decomposition} below.

Formally, the right-hand side of \eqref{e:Grep1} looks like the Gaussian super-expectation of
$e^{-\sum_{x\in\Lambda} g\tau_x^2}$
with $A=-\Delta + \nu$ where $\Delta$ is the discrete Laplace operator on $\Lambda_N$.
However, the critical point $\nu_c$ is \emph{negative} according to \eqref{e:nucasy}.
Thus $-\Delta+ \nu$ is an \emph{indefinite} matrix for the values of $\nu$ of interest,
and therefore is not a proper covariance of a Gaussian expectation.
On the other hand, for any $\epsilon > 0$, \eqref{e:Grep1} can be written as
\begin{equation}
  \int
  e^{-\sum_{x\in\Lambda}(\tau_{\Delta ,x} + g \tau_x^2 + \nu \tau_x)}
  \bar\phi_{\pp} \phi_{\qq}
  =
  \Ex_C(e^{-\sum_{x\in\Lambda}(g\tau_x^2 + (\nu-\epsilon)\tau_x)} \bar\phi_\pp \phi_\qq)
\end{equation}
with the positive-definite covariance $C = (-\Delta+\epsilon)^{-1}$.
A careful division of the right-hand side of \eqref{e:Grep1} into a Gaussian expectation
and a perturbation is central to our analysis and is discussed next,
where we divide not just the $\tau$ term but also the $\tau_\Delta$ term.

\section{Reformulation in renormalised parameters}
\label{sec:chvar}

\subsection{Approximation by simple random walk}
\label{sec:ga}

An important notion in theoretical physics is that it is often
possible to approximate an interacting system by an \emph{effective}
non-interacting system with \emph{renormalised} parameters.  As a
first step towards implementing this, we write temporarily
\begin{align}
    V_{g,\nu,z;x}
    =
    V_{g,\nu,z;x}(\phi,\psi)
    &
    =
    g \tau_x^2 + \nu \tau_x + z\tau_{\Delta ,x}
    .
\end{align}
We fix $m^2 > 0$ and $z_0 > -1$, and set
$(\phi',\psi')=((1+z_0)^{1/2}\phi,(1+z_0)^{1/2}\psi)$ and similarly
for the conjugates.  By definition,
\begin{equation} \label{e:Vsplit}
  V_{g,\nu,1;x}(\phi',\psi')
  = V_{0,m^2,1;x}(\phi,\psi)
  + V_{g_0,\nu_0,z_0;x}(\phi,\psi)
  ,
\end{equation}
where
\begin{equation} \label{e:gg0}
  g_0 = g(1+z_0)^2, \quad \nu_0 = (1+z_0)\nu-m^2
  ,
\end{equation}
or equivalently,
\begin{equation}
\label{e:g0g}
    g = \frac{g_0}{(1+z_0)^2}, \quad
    \nu = \frac{\nu_0+m^2}{1+z_0}
  .
\end{equation}
For $X \subset \Lambda$, we define
\begin{align}
  \label{e:Vtil0def}
  V_{0} (X)
  &
  = \sum_{x\in X} V_{g_0,\nu_0,z_0;x}
  =
  \sum_{x\in X}
  \big(g_{0} \tau_x^2 + \nu_{0} \tau_x + z_{0}\tau_{\Delta ,x}\big)
  .
\end{align}

Let $C=(-\Delta+m^2)^{-1}$, with $\Delta$ the discrete Laplacian on $\Lambda_N$.
By making the change of variables $\phi_x \mapsto \phi'=(1+z_0)^{1/2}\phi_x$,
and writing $F' (\phi,\psi) = F(\phi',\psi')$, we obtain
\begin{equation}
  \lbeq{ExF}
  \int F e^{-\sum_{x\in \Lambda}(\tau_{\Delta,x}+g\tau_x^2 + \nu\tau_x)}
  = \Ex_C F' e^{-V_0(\Lambda)} .
\end{equation}
There is no explicit Jacobian factor, since we also make the change of variables in the
differentials $\psi,\bar\psi$.
Then, by \refeq{ExF},
\begin{equation}
  \label{e:GG2}
  G_N(g,\nu) = (1+z_0) \hat G_N (m^2, g_0, \nu_0, z_0),
\end{equation}
where
\begin{equation}
  \label{e:Gmgnzdef}
  \hat G_N(m^2,g_0,\nu_0,z_0)
  = \Ex_C(e^{-V_0(\Lambda)} \bar\phi_a \phi_b)
\end{equation}
(for $g_0>0$ the right-hand side is convergent for all $\nu_0 \in \R$
and $m^2 >0$).
In \refeq{GG2}, we write the two-point
function as $G_N(g,\nu)$ instead of
$G_{N,g,\nu}(\pp,\qq)$, with $a,b$ now suppressed, since
the lattice points $\pp, \qq$ do not play a primary role for the moment
and it is rather the dependence on $g,\nu$ that
we wish to emphasise.
If we set $V_0=0$ on the right-hand side of \refeq{GG2}, the result is a multiple
of the simple random walk two-point function
with mass $m^2$.  We regard the factor $e^{-V_0}$
as a perturbation of the simple random walk.
Much of our effort lies in the choice of
the renormalised parameters $(m^2,z_0)$ and the determination of their
relation to the original (or ``bare,'' in terminology of quantum field theory) parameters
$(g,\nu)$, so that the perturbation is small and the approximation
by simple random walk is a good one.

We define an analogous quantity for the susceptibility by
\begin{equation}
  \label{e:chiNhatdef}
  \hat\chi_N (m^2, g_0, \nu_0, z_0)
  = \sum_{x\in \Lambda} \Ex_C(e^{-V_0(\Lambda)} \bar\phi_0\phi_x)
  .
\end{equation}
Then, as in \refeq{GG2}, $\chi$ and $\hat\chi$ are related by
\begin{equation}
\label{e:chichibar}
  \chi_N\left(g,\nu\right)
  = (1+z_0)\hat\chi_N(m^2, g_0, \nu_0, z_0)
  .
\end{equation}
The limit
\begin{equation}
  \label{e:chihatdef}
  \hat\chi(m^2, g_0, \nu_0, z_0)
  =
  \lim_{N \to \infty}\hat\chi_N (m^2, g_0, \nu_0, z_0)
\end{equation}
exists by Lemma~\ref{lem:suscept-finvol} and \refeq{chichibar}, and
\begin{equation}
  \label{e:chichihat}
  \chi(g,\nu) = (1+z_0) \hat\chi (m^2,g_0,\nu_0,z_0).
\end{equation}
Also, for $\nu > \nu_c$, it follows from Lemma~\ref{lem:suscept-finvol}
and the chain rule that
\begin{equation}
  \label{e:dchichihat}
  \ddp{\chi}{\nu}(g,\nu)
  = (1+z_0)^2 \ddp{\hat\chi}{\nu_0}(m^2,g_0,\nu_0,z_0)
  =
  (1+z_0)^2 \lim_{N \to \infty} \ddp{\hat\chi_N}{\nu_0}(m^2,g_0,\nu_0,z_0).
\end{equation}

The finite volume susceptibility $\hat\chi_N$ can be conveniently
re-expressed in terms of a generating functional, as follows.
Given an \emph{external field} (or \emph{test function}) $J: \Lambda \to \C$, we write
\begin{equation}
  (J,\bar\phi) = \sum_{x\in \Lambda} J_x \bar\phi_x,
  \quad
  (\bar J,\phi) = \sum_{x\in \Lambda} \bar J_x \phi_x.
\end{equation}
Let $1$ denote the constant test function $1_x = 1$ for all $x\in \Lambda$.
Let $m^2>0$, $g_0>0$, $\nu_0 \in \R$, $z_0>-1$.
By \eqref{e:chiNhatdef} and translation invariance,
\begin{align}
  \hat\chi_{N}
  &=
  \hat\chi_{N}(m^2, g_0, \nu_0, z_0)
  =|\Lambda|^{-1}\Ex_{C}((1,\bar\phi)(1,\phi) Z_0),
\end{align}
where $C = (-\Delta+m^2)^{-1}$ and $Z_0 = Z_0(g_0,\nu_0,z_0)=e^{-V_0(\Lambda)}$.
We define the bosonic \emph{generating functional}
$\Sigma: \C^\Lambda 
\to \C$ by
\begin{align}
  \Sigma(J, \bar J)
  & =
  \Ex_C(e^{(J,\bar\phi)+(\phi, \bar J)} Z_0).
\end{align}
Then
\begin{equation}
\label{e:chibarG}
  \hat\chi_{N}
  = |\Lambda|^{-1} D^2\Sigma(0,0; (0,1),(1,0))
  = |\Lambda|^{-1} D^2\Sigma(0,0; 1,1),
\end{equation}
where the right-hand sides involve the directional derivative
with directions equal to the constant function $1$ in the first and second argument, respectively,
for which we use the short-hand notation of the last equality.
The evaluation of $\hat\chi_{N}$ now becomes reduced to
the evaluation of $D^2\Sigma$ on the right-hand side
of \refeq{chibarG}.

As a first step, we  complete the square in the exponent to obtain
\begin{align} \label{e:Gamma-sq}
  \Sigma(J, \bar J)
  &=
  e^{(J,C\bar J)} \Ex_C(Z_0(\phi+C J, \bar\phi+C\bar J, \psi, \bar\psi)).
\end{align}
In more detail, with $A = -\Delta + m^2 = C^{-1}$,
\begin{equation}
\begin{aligned}
  \sum_{x\in\Lambda} \left( \tau_{\Delta,x} + m^2\tau_x \right)
  - (J,\bar\phi) - (\phi, \bar J)
  &=
  (\phi, A \bar\phi) + (\psi, A\bar\psi) - (J, \bar\phi)- (\phi, \bar J)
  \\
  &=
  (\phi-CJ, A (\bar\phi-C\bar J)) + (\psi, A\bar\psi) - (J,C \bar J)
\end{aligned}
\end{equation}
and \eqref{e:Gamma-sq} follows with the translation $\phi \mapsto \phi + CJ$
of the integration variable in the integral \eqref{e:Ebolddef} defining the
super-expectation $\Ex_C$.
This translation of $\phi$ leaves $\psi = (2\pi i)^{-1/2} d\phi$ unchanged.
As we explain next,
\eqref{e:Gamma-sq} can be expressed conveniently in terms of a notion of convolution.

By definition, any form $F$ in the algebra $\Ncal = \Ncal(\Lambda)$ of
differential forms (see Section~\ref{sec:intforms}) is a linear
combination of products of factors $\psi_{x_i}$ and $\bar\psi_{y_i}$,
with $x_i,y_i\in \Lambda$ and with coefficients given by functions of
$\phi$ and $\bar\phi$.
We also define an algebra $\mathcal{N}^\times$ with twice as many
fields as $\mathcal{N}$, namely with boson fields $(\phi,\xi)$ and
fermion fields $(\psi,\eta)$, where $\psi = (2\pi i)^{-1/2} d\phi$,
$\eta = (2\pi i)^{-1/2} d\xi$.
There are also the four corresponding conjugate fields.
Given a form $F=f(\phi, \bar\phi) \psi^x \bar\psi^y$
(where $\psi^x$ denotes a product
$\psi_{x_1}\cdots\psi_{x_j}$), we define
\begin{equation}
\label{e:thetadef}
    \theta F = f(\phi+\xi,\bar\phi+\bar\xi) (\psi+\eta)^x (\bar\psi+\bar\eta)^y,
\end{equation}
and extend this to a map $\theta : \mathcal{N} \to \mathcal{N}^\times$
by linearity.
Then we understand the map $\Ebold_C \circ \theta : \mathcal{N}
\to \mathcal{N}$ as the integration with respect to the
\emph{fluctuation fields} $\xi$ and $\eta$, with the fields $\phi$ and $\psi$
left fixed.  This is like a conditional expectation.
For example, if $F = f(\phi)$ is of degree-$0$, then
\begin{equation}
  \Ex_C\theta F = \mu_C * f
  = E_C(f(\phi+\xi) | \phi),
\end{equation}
where the right-hand side is the usual conditional expectation with respect to
the Gaussian measure $d\mu_C$ defined in \refeq{intzeroform}.
However, in general, this
is not usual probability theory, since $\Ebold_C\circ\theta$
acts on the algebra of forms.

The expectation with the translated boson field $\phi$ in \eqref{e:Gamma-sq}, with no corresponding translation
of the fermion field $\psi$,
can be expressed succinctly in terms of $\Ex_C \theta$
and projection onto the subspace of degree zero forms.
To this end, let
\begin{equation} \label{e:ZNdef}
  Z_N = \Ex_C\theta Z_0
\end{equation}
which is a differential form on $\C^{\Lambda}$.
The subscript $N$ on $Z_N$ refers to the parameter defining the finite volume $\Lambda_N$.
We denote the degree-$0$
part of $Z_N$ by $Z_N^0$, i.e., $Z_N^0$ is a function
$ Z_N^0: \C^\Lambda \to \C$.
Then
\begin{align}
  \lbeq{GamZ0}
  \Sigma(J, \bar J)
  &
  =
  e^{(J,C\bar J)} Z_N^0(C J, C\bar J).
\end{align}
As mentioned already in  \eqref{e:self-norm},
the partition function $\Ex_C Z_0$
is equal to $1$
by supersymmetry.  However it is much
more challenging to understand the convolution
$\Ex_C \theta Z_0$.

To express the susceptibility conveniently in terms of $Z_N^0$,
observe that $1$ is an eigenfunction of the covariance matrix $C$, namely
\begin{equation} \label{e:C1m2}
  C1  = (-\Delta + m^2)^{-1} 1 = \frac{1}{m^{2}} 1.
\end{equation}
With \refeq{GamZ0}, this and $(1,1) = |\Lambda|$ imply that
\begin{equation}
  D^2\Sigma(0, 0; 1, 1)
  =  (1,C 1) + D^2Z_N^0(0, 0; C1, C1)
  = \frac{1}{m^2} |\Lambda| + \frac{1}{m^4} D^2Z_N^0(0, 0; 1, 1).
\end{equation}
Thus, with \eqref{e:chibarG}, we obtain
\begin{equation}
  \label{e:chibarm}
  \hat\chi_N
  = \frac{1}{m^2}  + \frac{1}{m^4} \frac{1}{|\Lambda|} D^2Z_N^0(0, 0; 1, 1)
  .
\end{equation}

The representation
\eqref{e:chichihat}
has the two free parameters $m^2$ and $z_0$, which define a division of the
quadratic $\tau$ and $\tau_{\Delta}$ terms between
the Gaussian expectation $(1+z_0) \Ex_C$ and
the perturbation $V_0(\Lambda)$.
The formula \refeq{chibarm} holds for \emph{any} choice of $m^2,z_0$
(and also any choice of $g_0,\nu_0$), but to study the behaviour near
the critical point $\nu_c$, it is essential that this split be made
exactly right.  The correct \emph{critical} choice is given in the
following theorem, whose proof occupies the remainder of the paper.
With this choice, we have $|\Lambda|^{-1} D^2Z_N^0(0,0;1,1) \to 0$ as
$N\to\infty$ in \refeq{chibarm}.  In Section~\ref{sec:pfsuscept}
below, we show that Theorem~\ref{thm:suscept} is a consequence of
Theorem~\ref{thm:suscept-diff}.

\begin{theorem}
\label{thm:suscept-diff} Let $d = 4$, and
let $\delta > 0$ be sufficiently small.
There are continuous real-valued functions $\nu_0^c,z_0^c$, defined for
$(m^2, g_0) \in [0,\delta)^2$ and
continuously differentiable in $g_0$,
and there is a continuous function $c(g_0)=1+O(g_0)$, such that
for all $m^2,g_0,\hat g_0 \in (0,\delta)$,
\begin{align} \label{e:chi-m}
  \hat\chi \left( m^2,g_0,\nu_0^c(m^2,g_0),z_0^c(m^2,g_0) \right)
  &= \frac{1}{m^2} ,
  \\
  \label{e:chiprime-m}
  \ddp{\hat\chi}{\nu_0} \left(m^2,g_0,\nu_0^c(m^2,g_0),z_0^c(m^2,g_0) \right)
  &\sim - \frac{1}{m^4} \frac{c(\hat g_0)}{(\hat g_0\bubble_{m^2})^{\gamma}}
  \quad \text{as $(m^2,g_0) \to (0,\hat g_0)$}.
\end{align}
The functions $\nu_0^c,z_0^c$ obey
\begin{align}
  &\nu_0^c(m^2,0) = z_0^c(m^2,0) = 0,
  \quad
  \ddp{\nu_0^c}{g_0}(m^2,g_0) = O(1),
  \quad
  \label{e:z0est}
  \ddp{z_0^c}{g_0}(m^2,g_0) = O(1),
\end{align}
where $O (1)$ means that these derivatives are bounded on their
whole domain by constants uniform in $(m^2,g_0)$.
\end{theorem}

\subsection{Change of parameters}
\label{sec:changeofvariables}

We have introduced six real variables $\{g,\nu,m^2,g_0,z_0,\nu_0\}$, and it
is convenient in different contexts to switch perspective on which are dependent and which are independent variables.
In particular, to deduce Theorem~\ref{thm:suscept} from Theorem~\ref{thm:suscept-diff},
we relate the parameters $(m^2,g_0,\nu_0^c,z_0^c)$ of $\hat \chi$ in \eqref{e:chihatdef} 
with the parameters $(g,\nu)$ of $\chi$.
We summarise the different perspectives now.
The six variables are constrained to satisfy
the two equations in \eqref{e:gg0} as well as
\begin{equation} \label{e:nu0cz0c}
  \nu_0 = \nu_0^c(m^2,g_0),
  \quad
   z_0 = z_0^c(m^2,g_0),
\end{equation}
with $\nu_0^c,z_0^c$ the functions of Theorem~\ref{thm:suscept-diff}.
In particular, \eqref{e:GG2} and \eqref{e:chichihat} hold for such a choice.

Given $(m^2,g_0)$, \eqref{e:gg0} and \refeq{nu0cz0c} determine $g^c,\nu^c$ such that
\begin{equation}
\lbeq{cccrit}
    (g,\nu,\nu_0,z_0) = (g^c(m^2,g_0),\nu^c(m^2,g_0),\nu_0^c(m^2,g_0),z_0^c(m^2,g_0)),
\end{equation}
and \refeq{cccrit} is continuous in $(m^2,g_0) \in [0,\delta)^2$ by Theorem~\ref{thm:suscept-diff}.
In Proposition~\ref{prop:changevariables}(i) below,
we show that, given $(m^2,g)$, we can determine
\begin{equation}
\lbeq{ccstar}
    (\nu,g_0,\nu_0,z_0)
    = (\nu^*(m^2,g), g_0^*(m^2,g), \nu_0^*(m^2,g), z_0^*(m^2,g)),
\end{equation}
continuously in $(m^2,g) \in [0,\delta_1)^2$,
with \eqref{e:gg0} and \eqref{e:nu0cz0c} satisfied.
In Proposition~\ref{prop:changevariables}(ii) below,
we show that given $(g,\nu_c+ \varepsilon)$, we can
determine
\begin{equation}
\lbeq{tilpars}
    (m^2,g_0,\nu_0,z_0)
    =
    (\tilde m^2(g,\varepsilon), \tilde g_0(g,\varepsilon), \tilde \nu_0(g,\varepsilon), \tilde z_0(g,\varepsilon))
\end{equation}
so that \eqref{e:gg0} and \eqref{e:nu0cz0c} hold, with
\emph{right-continuity} as $\varepsilon \downarrow 0$.  (We also
expect continuity in $\varepsilon > 0$, but have not proved it.)

Theorem~\ref{thm:suscept-diff} is stated in terms of parameters $(m^2,g_0)$ rather
than the original parameters $(g,\nu)$ in \eqref{e:suscept-def}.
Let $\nu=\nu_c(g)+\varepsilon$, and
let $(\tilde m^2, \tilde g_0, \tilde \nu_0, \tilde z_0)$ be given by
\refeq{tilpars}.
By \eqref{e:chichihat} and Theorem~\ref{thm:suscept-diff},
we have the \emph{equality}
\begin{equation}
\lbeq{chi-mtil0}
    \chi(g,\nu)
    =
    (1+ \tilde z_0) \hat\chi (\tilde m^2,\tilde g_0,\tilde \nu_0,\tilde z_0)
    =
    (1+ \tilde z_0) \frac{1}{\tilde m^2}
    .
\end{equation}
The right-hand side is $(1+\tilde z_0)$ times the susceptibility
$\tilde m^{-2}$ of the non-interacting
walk with killing rate $\tilde m^2$.
Thus we have implemented the goal mentioned
at the beginning of Section~\ref{sec:ga},
i.e., to represent the susceptibility of
the interacting model by that of a non-interacting model with renormalised parameters.
Equation~\eqref{e:chi-mtil0} is reminiscent of the \emph{renormalisation conditions} imposed
in the physics literature (e.g., \cite[Chapter~5]{ID89a}). However, we stress that here
\eqref{e:chi-mtil0} arises not by defining $\tilde m^2,\tilde z_0$
by the requirement  that the equality holds,
but rather that our method computes $\tilde m^2,\tilde z_0$ and we subsequently verify
that the equality holds.
The equality \refeq{chi-mtil0} is an identity, not merely an asymptotic
formula, and thus it contains all information about the
susceptibility, including not just the leading asymptotic behaviour but
also all higher-order corrections.

\begin{prop} \label{prop:changevariables}
  There exists $\delta_1>0$ such that the following hold.

  \smallskip\noindent
  (i)
  For $(m^2,g) \in [0,\delta_1)^2$, there exist
  \begin{equation}
    \lbeq{ccstar2}
    (\nu,g_0,\nu_0,z_0)
    = (\nu^*(m^2,g), g_0^*(m^2,g), \nu_0^*(m^2,g), z_0^*(m^2,g)),
  \end{equation}
  continuous in $(m^2,g)$,
  such that
  \eqref{e:gg0} and \eqref{e:nu0cz0c} hold, and
  \begin{gather}
    \label{e:nustarbd}
    \nu^*(0,g) = \nu_c(g), \quad
    \nu^*(m^2,g) > \nu_c(g) \quad (m^2 > 0),
    \\
    \label{e:gznustarbd}
    g_0^*(m^2,g) = g + O(g^2),
    \quad
    \nu_0^*(m^2,g) = O(g),
    \quad
    z_0^*(m^2,g) = O(g).
  \end{gather}
  (ii)
  For $(g,\varepsilon) \in [0,\delta_1)^2$, there exist
  \begin{equation}
    \lbeq{cctilde}
    (m^2,g_0,z_0,\nu_0) = (\tilde m^2(g,\varepsilon),\tilde g_0(g,\varepsilon),\tilde z_0(g,\varepsilon),\tilde \nu_0(g,\varepsilon)),
  \end{equation}
  right-continuous as $\varepsilon \downarrow 0$ (with $g$ fixed),
  such that \eqref{e:gg0} and \eqref{e:nu0cz0c} hold,
  and
  \begin{gather} \label{e:mtildebd}
    \tilde  m^2(g,0) = 0, \quad
    \tilde m^2(g, \varepsilon) > 0 \quad (\varepsilon>0),
    \\
    \label{e:gznutildebd}
    \tilde  g_0(g,\varepsilon) = g + O(g^2),
    \quad
    \tilde  \nu_0(g,\varepsilon) = O(g),
    \quad
    \tilde  z_0(g,\varepsilon) = O(g).
  \end{gather}
\end{prop}

In the following proof, the construction of the maps \refeq{ccstar2} and \refeq{cctilde}
involves only elementary calculus.
The proof does not use \refeq{chiprime-m}, but
the identification of the critical point
(which occurs in \eqref{e:nustarbd} and \eqref{e:mtildebd}),
as well as the proof of right-continuity of $(\tilde m^2, \tilde g_0, \tilde \nu_0, \tilde z_0)$
as $\varepsilon \downarrow 0$, uses \eqref{e:chi-m}.

\begin{proof}
(i)
For $(m^2,g_0) \in [0,\delta)^2$, set
\begin{equation}
  \lbeq{smg}
  s(m^2,g_0) = \frac{g_0}{(1+ z_0^c (m^{2},g_{0}))^2}
  .
\end{equation}
For fixed $m^2$, we show below that it is possible to
construct an inverse to the map $g_0 \mapsto s$,
i.e., to show that there exists $\delta_1 >0$ such that the equation $g=s(m^2,g_0)$
can be solved for $g_{0}=g_0^*(m^2,g)$ as a continuous function of $(m^2,g) \in
[0,\delta_1)^2$. Given its existence, we set
\begin{equation}
  \nu_0^* = \nu_0^*(m^2,g) = \nu_0^c(m^2,g_0^*(m^2,g))
  ,
  \quad\quad
  z_0^* = z_0^*(m^2,g) = z_0^c(m^2,g_0^*(m^2,g))
  ,
\end{equation}
and also define
\begin{equation}
  \label{e:mu-m-def}
  \nu^* = \nu^*(m^2 ,g) = \frac{\nu_0^* + m^2}{1+z_0^*}.
\end{equation}
Since $\nu_0^c$ and $z_0^c$ are continuous in $(m^2,g_0)$,
and since $g_0^*$ is continuous in $(m^2,g)$,
it is also the case that $z_0^*$, $\nu_0^*$, and $\nu^*$
are continuous in $(m^2,g)$, and it is immediate
that \eqref{e:gg0} and \eqref{e:nu0cz0c} hold,
and also that \eqref{e:gznustarbd} holds.
To show \eqref{e:nustarbd}, it follows from \refeq{chichihat} and \refeq{chi-m} that
\begin{equation}
  \chi(g,\nu^*(m^2,g)) = (1+z_0^*) \hat\chi(m^2,g_0^*,\nu_0^*,z_0^*)
    = (1+z_0^* )\frac{1}{m^2}.
\end{equation}
In particular, $\chi(g, \nu^*(m^2,g)) < \infty$ if $m^2 > 0$ and
therefore $\nu^*(m^2,g) > \nu_c(g)$ for $m^2>0$.  Also, since
$\chi(g,\nu^*) \uparrow \infty$ as $m^2 \downarrow 0$, it follows from
the continuity of $\nu^*$ that
\begin{equation} \label{e:nucinf}
    \nu_c(g) = \nu^*(0,g) = \inf\{ \nu^*(m^2,g) : m^2\in [0,\delta_1)\}.
\end{equation}

To construct the inverse to $g_0 \mapsto s$, we proceed as follows.
By Theorem~\ref{thm:suscept-diff}, $\nu_0^c,z_0^c$ are continuous in
$(m^2,g_0) \in [0,\delta)^2$, differentiable in $g_0$, and satisfy
\eqref{e:z0est}. Thus, with $z_0' = \ddp{}{g_0} z_0(m^2,g_0)$,
\begin{align} \label{e:sdiff}
  \ddp{}{g_0}s(m^2,g_0)
  = \frac{(1+z_0)^2 - 2g_0 (1+z_0) z_0'}{(1+z_0)^4}
  = 1+O(g_0) > 0.
\end{align}
For sufficiently small $g_*>0$, $s$ is therefore a strictly increasing
continuous function of $g_{0} \in [0,g_*)$ such that $|s (m^2,u) - s
(m^2,v)| \ge (1 - O (g_{*})) |u - v|$ and hence, for $m^2$ fixed,
$s(m^2, \cdot)$ is a continuously invertible map from $[0,g_{*})$ onto
the interval $[0,s(m^2,g_{*}))$.  Let $[0,\delta_1)$ be in the
intersection over $m^{2}>0$ of the latter intervals. Since
\eqref{e:sdiff} holds uniformly in $m^2$, it follows that
$\delta_1>0$. Therefore, $g=s(m^2,g_0)$ can be solved for $g_{0}$ as a
function $g_0^*(m^{2},g)$ for $g \in [0,\delta_1)$ and $g_{0}^*$ is
continuous in $g$ for $m^2$ fixed.
To prove that $g_{0}^*$ is jointly continuous in $(m^2,g)$ it suffices
to show that when $(\hat m^2,\hat g) \rightarrow (m^2,g)$, then $\hat
g_{0} \rightarrow g_{0}$, where $\hat g_{0},g_{0}$ solve $s(\hat
m^2,\hat g_{0}) - \hat g = 0 = s(m^2,g_{0}) - g$.  This follows from
$(1 - O(g_{*})) |\hat g_{0} - g_{0}| \le |s (\hat m^2,\hat g_{0}) - s
(\hat m^2,g_{0})| = |(s (m^2,g_{0}) - s (\hat m^2,g_{0})) + (\hat g-
g)| \rightarrow 0$, since $s(\cdot,g_0)$ is continuous by \refeq{smg}
and the continuity of $z_0^c$.

\smallskip\noindent (ii)
We set
\begin{equation} \lbeq{mtildef}
  \tilde m^2 = \tilde m^2 (g,\varepsilon) = \inf \{m^2 >0 : \nu^*(m^2,g)=\nu_c(g) + \varepsilon \}.
\end{equation}
By continuity of $\nu^*$, the infimum is attained and
\begin{equation} \lbeq{nucstar}
  \nu_c(g) + \varepsilon = \nu^*(\tilde m^2(g,\varepsilon),g).
\end{equation}
The left-hand side of \refeq{nucstar} converges to $\nu_c$ as $\varepsilon \downarrow 0$,
and hence
$\tilde m^2(g,\varepsilon) \downarrow 0$ as $\varepsilon \downarrow 0$
(if $\tilde m^2(g,\varepsilon)$ had a nonzero accumulation point the right-hand
side of \refeq{nucstar} would not converge to $\nu_c$).
We set
\begin{equation}
    \tilde g_0 = \tilde g_0(g,\varepsilon) = g_0^*(\tilde m^2 ,g),
    \quad
    \tilde \nu_0 = \tilde \nu_0(g,\varepsilon) = \nu_0^*(\tilde m^2 ,g),
    \quad
    \tilde z_0 = \tilde z_0(g,\varepsilon) = z_0^*(\tilde m^2 ,g),
\end{equation}
and conclude that \eqref{e:gg0} and \eqref{e:nu0cz0c} hold.
The desired properties of $\tilde \nu_0$ and $\tilde z_0$
follow from those of $\tilde m$ and $\tilde g_0$,
and continuity of $\nu_0^*$ and $z_0^*$,
and the proof is complete.
\end{proof}

\subsection{Proof of Theorem~\ref{thm:suscept} assuming Theorem~\ref{thm:suscept-diff}}
\label{sec:pfsuscept}

The following elementary lemma is used in the proof of
Theorem~\ref{thm:suscept}.

\begin{lemma} \label{lem:ODE}
  Let $\delta > 0$.
  Suppose that $u: [0,\delta) \to [0,\infty)$ is continuous,
  differentiable on $(0,\delta)$, that $u(0)=0$ and $u(t)>0$ for $t>0$, and that
  \begin{equation}
    u'(t) = (-\log u(t))^{-\gamma} ( 1+ o(1)) \quad \text{(as $t\downarrow 0$)} .
  \end{equation}
  Then
  \begin{equation}
    u(t) = t(-\log t)^{-\gamma} (1+o(1))
    \quad \text{(as $t\downarrow 0$)}
    .
  \end{equation}
\end{lemma}

\begin{proof}
By hypothesis,
\begin{equation}
    \int_0^t u'(t) (-\log u(t))^{\gamma} \; dt
    =
    \int_0^t (1+o(1)) \; dt
    =
    t(1+o(1))
    .
\end{equation}
  Note that $u(t) > 0$ implies that $u'(t) > 0$ for small $t$, so $u$ is monotone.
  A change of variables, followed by integration by parts, gives
  \begin{equation}
  \begin{aligned}
    \int_0^t u'(t) (-\log u(t))^{\gamma} \; dt
    &
    =
    \int_0^{u(t)} (-\log v)^{\gamma} \; dv
    \\&
    =  u(t) (-\log u(t))^\gamma
    (1+ O((-\log u(t))^{-1}))
    .
  \end{aligned}
  \end{equation}
  The equality of the above two right-hand sides then gives
  \begin{align}
    u(t) (-\log u(t))^\gamma = t(1+o(1)).
  \end{align}
  Let $f(x) = x(-\log x)^\gamma$ and $g(y) = y(-\log y)^{-\gamma}$. Then
  $f$ and $g$ are approximate inverses in the sense that
  $f(g(y)) = y(1+o(1))$.
  Thus
    $u(t) = t(-\log t)^{-\gamma}(1+o(1))$,
  and the proof is complete.
\end{proof}

\begin{proof}[Proof of Theorem~\ref{thm:suscept}]
Let $\nu=\nu_c(g)+\varepsilon$.
We have observed already in \refeq{chi-mtil0} that
\begin{equation}
\lbeq{chi-mtil}
    \chi(g,\nu)
    =
    (1+ \tilde z_0) \hat\chi (\tilde m^2,\tilde g_0,\tilde \nu_0,\tilde z_0)
    =
    (1+ \tilde z_0) \frac{1}{\tilde m^2}
    .
\end{equation}
For the derivative, we apply \refeq{dchichihat}, \refeq{chiprime-m}
and \refeq{chi-mtil} to obtain, as $\nu \downarrow \nu_c$,
\begin{equation}
\lbeq{dchi-mtil}
    \ddp{\chi}{\nu}(g,\nu)
    =
    (1+\tilde z_0)^2 \ddp{\hat\chi}{\nu_0}
    (\tilde m^2,\tilde g_0,\tilde \nu_0,\tilde z_0)
    \sim
    -  \chi^2(g,\nu)
    \frac{c_0(g)}{(\tilde g_0 \bubble_{\tilde m^2})^\gamma}
    .
\end{equation}
The constant on the right-hand side is $c_0(g)= \lim_{\varepsilon
\downarrow 0} c(\tilde g_0(g,\varepsilon))$ which exists by
right-continuity of $\tilde g_0$ at $\varepsilon=0$ for fixed $g$ and
by continuity of $c$.
Equation~\eqref{e:dchichihat} relates
a directional derivative of $\chi$ at the point $(g,\nu)$ to a
directional derivative of $\hat{\chi}$ at any point
$(m^{2},g_{0},\nu_{0},z_{0})$ related to $(g,\nu)$ by
\eqref{e:gg0}. Equation \eqref{e:dchi-mtil} is valid because it
specialises \eqref{e:dchichihat} to $(\tilde m^2,\tilde g_0,\tilde
\nu_0,\tilde z_0)$ which is related to $(g,\nu)$ by
\eqref{e:gg0}.

Now we drop $g$ from the notation.
For $\nu > \nu_c$, let
\begin{equation} \label{e:F-def}
  F(\nu) = \frac{1}{\chi(\nu)}.
\end{equation}
By Proposition~\ref{prop:changevariables}(ii),
$\tilde z_0(\varepsilon)$ is continuous
as $\varepsilon \downarrow 0$, and hence, by \refeq{chi-mtil},
\begin{equation} \label{e:F-m}
  F(\nu) \sim (1+ \tilde z_0 (0))^{-1} \tilde m^2
  \quad\quad
  \text{(as $\nu \downarrow \nu_c$)}.
\end{equation}
We set $F(\nu_c)=0$.
By \eqref{e:dchi-mtil}, together with \eqref{e:freebubble},
\refeq{F-m} and \eqref{e:gznutildebd},
\begin{equation}
  \label{e:Fprime}
  \ddp{F}{\nu}
  = - \frac{1}{\chi^{2}} \ddp{\chi}{\nu}
  \sim \frac{c_0(g)}{(\tilde g_0 \bubble_{\tilde m^2})^{\gamma}}
  \sim
  \frac{c_0(g)}{(\tilde g_0(0) {\sf b})^\gamma}(-\log F(\nu))^{-\gamma}
  \quad\quad
  \text{(as $\nu \downarrow \nu_c$)}
  .
\end{equation}
We apply Lemma~\ref{lem:ODE} with $u(t)=(\tilde g_0(0) {\sf b})^\gamma c_0(g)^{-1} F(\nu_c+t)$
to conclude from \eqref{e:Fprime} that
\begin{equation} \label{e:F-asymp}
  F(\nu_c+\varepsilon) \sim A_g^{-1}\varepsilon(-\log \varepsilon)^{-\gamma}
\end{equation}
with
\begin{equation}
  A_g = \frac{(\tilde g_0 (0) {\sf b})^\gamma}{c_0(g)}.
\end{equation}
This is equivalent to \eqref{e:chieps-asympt}.
Since $c_0(g) = 1+O(\tilde g_0 (0)) = 1+O(g)$ as $g \downarrow 0$ by \eqref{e:gznutildebd},
we obtain
\begin{equation}
  A_g =
  (g {\sf b})^\gamma (1+O(g)),
\end{equation}
which proves \eqref{e:cgasy} and completes the proof.
\end{proof}

As a consequence of Theorems~\ref{thm:suscept} and
\ref{thm:suscept-diff}, we note in passing that the effective killing
rate $m^2$ vanishes as the critical point is approached, i.e., as
$\varepsilon = \nu - \nu_c(g) \downarrow 0$, according to the
asymptotic formula
\begin{equation}
  \label{e:m-mu}
  m^2
  \sim
  \frac{1+z_0^c(0,g_0)}{A_g} \varepsilon (\log \varepsilon^{-1})^{-\gamma}
  \quad \text{as $\varepsilon \downarrow 0$}.
\end{equation}

\section{Renormalisation group}
\label{sec:rg}

In this section, we gather together some of the important ingredients
of the renormalisation group analysis used in the proof of Theorem~\ref{thm:suscept-diff}.

\subsection{Progressive integration and covariance decomposition}
\label{sec:decomposition}

The proof of Theorem~\ref{thm:suscept-diff} is based on
\emph{renormalisation group analysis} of the integral on the
right-hand side of \eqref{e:chiNhatdef}.
We now explain a fundamental mechanism of our renormalisation group method:
progressive Gaussian integration which enables a multi-scale analysis.

It is an elementary fact that if $X_1 \sim N(0, \sigma_1^2)$ and
$X_2 \sim N(0,\sigma_2^2)$ are independent centred Gaussian random variables
with respective variances $\sigma_1^2,\sigma_2^2$,
then $X_1+X_2 \sim N(0, \sigma_1^2+\sigma_2^2)$
(here $\sim$ denotes equality in distribution).
In particular,
if $X \sim N(0,\sigma_1^2+\sigma_2^2)$ then we can evaluate an expectation
$E(f(X))$ by performing iterated Gaussian integrals, as
\begin{equation}
\label{e:EX1X2}
    E(f(X)) = E( E( f(X_1+X_2) \, | \, X_2) ).
\end{equation}
The inner expectation on the right-hand side is conditional on $X_2$, and
the outer expectation then averages over $X_2$.
The elementary identity \refeq{EX1X2} extends to the
super-expectation \eqref{e:Ebolddef},
as we describe next.  Recall the map $\theta$ defined in \refeq{thetadef}.

\begin{prop}\label{prop:Gaussian-conv}
Let $F \in \mathcal{N}(\Lambda)$, and suppose that $C_1$ and $C'$
are $\Lambda\times\Lambda$ matrices with
positive-definite Hermitian parts.  Then
  \begin{equation} \label{e:Gaussian-conv}
    \Ex_{C'+C_1}\theta F = (\Ex_{C'}\theta \circ \Ex_{C_1}\theta) F.
  \end{equation}
\end{prop}

For a proof, see \cite[Proposition~\ref{norm-prop:conv}]{BS-rg-norm}.
Equations~\eqref{e:Gaussian-conv} and \eqref{e:EX1X2} are closely related.
To see this, consider the special case in which we set
$\phi=\bar\phi=0$ and $\psi=\bar\psi=0$ in \eqref{e:Gaussian-conv}.  In this case,
\refeq{Gaussian-conv} gives
\begin{equation}
    \Ex_{C'+C_1}  F = \Ex_{C'}  ( \Ex_{C_1}\theta  F),
\end{equation}
which is an abbreviation for the lengthier formula
(in which we suppress the conjugate fields) expressed in the conventional notation
for conditional expectations as
\begin{equation}
    \Ex_{C'+C_1}  F(\xi,\eta)
    = \Ex_{C'}  ( \Ex_{C_1}  ( F(\xi'+\xi_1,\eta'+\eta_1)  \mid \xi',\eta') ).
\end{equation}

As in Section~\ref{sec:gauss-sub},
we set $C =(-\Delta +m^2)^{-1}$ where $\Delta$ is the
Laplacian on
the torus $\Lambda_N$ of period $L^N$.
Our goal is to prove Theorem~\ref{thm:suscept-diff}, which concerns the asymptotic
behaviour of $\hat\chi_N$ of \eqref{e:chiNhatdef}.
In particular,
we are interested in the limits  $\Lambda_N \uparrow \Z^d$ and $m^2 \downarrow 0$,
so $C$ is an approximation to the operator $(-\Delta_{\Z^d})^{-1}$.
The kernel $[(-\Delta_{\Z^d})^{-1}]_{xy}$ decays as $|x-y|^{-2}$ as $|x-y|\to\infty$ in dimension
$d=4$, and such long-range correlations make the analysis difficult.
The renormalisation group approach takes the long-range correlations
into account progressively, by making a good decomposition of the
covariance $C$ into a sum of terms with \emph{finite range},
together with \emph{progressive integration}.
The covariance $(-\Delta_\Zd + m^2)^{-1}$ with $m^2 > 0$ also arises in our
analysis.

The progressive integration is based on finite-range decompositions
of the covariances $C$ and $(-\Delta_\Zd + m^2)^{-1}$.
The decompositions we use are defined and discussed in detail
in \cite[Section~\ref{pt-sec:Cdecomp}]{BBS-rg-pt}, based on
\cite{Baue13a} (see also \cite{Bryd09,BGM04}).
In \cite{Baue13a}, it is shown
that there are positive-definite
covariances $(C_j)_{1 \le j < \infty}$ on $\Zd$ such that
\begin{equation}
    (-\Delta_\Zd + m^2)^{-1} = \sum_{j=1}^\infty C_j.
\end{equation}
These covariances have the \emph{finite-range} property that
\begin{equation} \label{e:finite-range}
  C_{j;x,y} = 0 \quad \text{if $|x-y| \geq \half L^j$},
\end{equation}
and therefore $C_j$ can also be identified with a covariance on
the torus of period $L^N$ when $N>j$.
In addition, there is a positive-definite covariance $C_{N,N}$ on $\Lambda$ such that
\begin{equation}
\lbeq{CCN}
    C=(-\Delta_\Lambda + m^2)^{-1} = \sum_{j=1}^{N-1} C_j + C_{N,N}.
\end{equation}

The finite-range property
gives rise to a factorisation property which is essential for the analysis
of \cite{BS-rg-step},
upon which our results depend.  To state the property,
given $X \subset \Lambda$,
recall that $\Ncal(X)$ was defined in Section~\ref{sec:intforms} as the
subalgebra of $\Ncal(\Lambda)$ consisting of forms that depend
on the boson and fermion fields only at points in $X$.
It follows from \cite[Proposition~\ref{norm-prop:factorisationE}]{BS-rg-norm}
that for forms $F \in \Ncal(X)$
and $G \in \Ncal(Y)$ such that $\dist(X,Y) > \frac 12 L^j$,
\begin{equation}
\label{e:Exfac}
  \Ex_{C_j}\theta(FG) = (\Ex_{C_j}\theta F)(\Ex_{C_j}\theta G),
\end{equation}
so $F$ and $G$ are uncorrelated with respect to $\Ex_{C_j}$.
This is an extension to the fermionic setting of the
standard fact that uncorrelated Gaussian random
variables are independent.

By \refeq{CCN} and Proposition~\ref{prop:Gaussian-conv},
\begin{equation}
    \label{e:progressive}
    \Ex_{C}\theta F
    =
    \big( \Ex_{C_{N,N}}\theta \circ \Ex_{C_{N-1}}\theta \circ \cdots
    \circ \Ex_{C_{1}}\theta\big) F
    ,
\end{equation}
and this expresses the expectation on the left-hand side as a progressive
integration.
The calculation of the expectation on the right-hand side of
\refeq{chibarm}
provides the basis for our proof of Theorem~\ref{thm:suscept-diff}.
For simplicity, we sometimes write $C_N = C_{N,N}$.
To compute the expectation \eqref{e:ZNdef}, we use
\refeq{progressive}
to evaluate it progressively.  Namely, if we define
\begin{equation}
  \label{e:Z0def}
  Z_0 = e^{-V_0(\Lambda)}, \quad
  Z_{j+1} = \Ex_{C_{j+1}}\theta Z_j \;\;\; (j<N),
\end{equation}
then, consistent with \eqref{e:ZNdef},
\begin{equation}
  \label{e:ZN}
  Z_N = \Ex_C\theta Z_0.
\end{equation}
Thus we are led to study the recursion $Z_j \mapsto Z_{j+1}$.
To simplify the notation, we use the short-hand notation $\Ex_{j} = \Ex_{C_j}$,
and leave implicit the dependence of the covariance $C_j$
on of the mass $m$.

In our analysis, the value $m^2=0$ corresponds to $\nu=\nu_c$.
The fact that we work with covariances with $m^2>0$
is an important feature of our method because it
permits the renormalisation
group step to be carried out with good estimates indefinitely, not only for $\nu=\nu_c$
but also for $\nu > \nu_c$.

\subsection{Typical size of fields}
\label{sec:fieldsize}

To control the progressive
integration,
we need
good estimates on $C_j$.
To state the estimates, we fix some $\delta>0$ and define mass intervals
\begin{equation}
\lbeq{massint}
    \Iint_j = \begin{cases}
    [0,\delta) & (j<N)
    \\
    [\delta L^{-2(N-1)},\delta) & (j=N).
    \end{cases}
\end{equation}
It is shown in
\cite[Proposition~\ref{pt-prop:Cdecomp}]{BBS-rg-pt}
that  for multi-indices $\alpha,\beta$ with
$\ell^1$ norms $|\alpha|_1,|\beta|_1$ at most
some fixed value $p$, for $j \le N$, for $m^2 \in \Iint_j$, and for any $k \in \N$,
\begin{equation}
  \label{e:scaling-estimate}
  |\nabla_x^\alpha \nabla_y^\beta C_{j;x,y}|
  \leq c(1+m^2L^{2(j-1)})^{-k}
  L^{-(j-1)(2[\phi]+(|\alpha|_1+|\beta|_1))},
\end{equation}
where $c=c(k)$ depends on $k$ and $\delta$ but is independent of $j$ and $m^2$,
and the estimate is for $C_{N,N}$ when $j=N$.
Here $\nabla_x^\alpha=\nabla_{x_1}^{\alpha_1} \dotsb \nabla_{x_d}^{\alpha_d}$ for a
multi-index $\alpha=(\alpha_1,\dotsc,\alpha_d)$, where
$\nabla_{x_k}$ denotes the
finite-difference operator $\nabla_{x_k}f(x,y)=f(x+e_k,y)-f(x,y)$.
The number $[\phi]=\frac 12
(d-2)$ is referred to as the \emph{scaling dimension} or
\emph{engineering dimension} of the field, or, more briefly, simply as
the field's \emph{dimension}.
Throughout this paper, the parameter $\delta>0$ in \eqref{e:massint} is fixed
and all constants are allowed to depend on it.
For other reasons, we sometimes assume that $\delta>0$ is small.
We take limits in the order $N\to\infty$, and then $m^2 \downarrow 0$,
so that the condition $m^2 \geq \delta L^{-2(N-1)}$
required for the validity of \eqref{e:scaling-estimate} is satisfied.
Moreover, throughout the remainder of this paper, we tacitly
assume that $L$ is sufficiently large. This is needed,
in particular, for the results of \cite{BS-rg-IE,BS-rg-step}, on which our
results rely.

In the expectation $Z_{j+1}=\Ex_{j+1}\theta Z_j$, on the right-hand
side we may write $\phi_j = \phi_{j+1}+\xi_{j+1}$, as in
\eqref{e:thetadef}, and similarly for $\bar\phi_{j}, d\phi_{j},
d\bar\phi_{j}$.  The expectation $\Ex_{j+1}\theta$ integrates out
$\xi_{j+1}, \bar\xi_{j+1}, d\xi_{j+1}, d\bar\xi_{j+1}$ leaving
dependence of $Z_{j+1}$ on $\phi_{j+1}$, $\bar\phi_{j+1}$,
$d\phi_{j+1}$, $d\bar\phi_{j+1}$.  This process is repeated.  The
$\xi_j$ fields that are integrated out are called \emph{fluctuation
fields}.  The field $\phi_j$ is thus the variable in a Gaussian
super-expectation with the remaining covariance $C_{\ge
j+1}=\sum_{k=j+1}^N C_k$.  By \eqref{e:intzeroform}, $\Ex_{\ge
j+1}|\phi_{j,x}|^2 = C_{\ge j+1;x,x}$.  The bounds of
\eqref{e:scaling-estimate} suggest that $C_{\ge j+1;x,x}\approx
C_{j+1;x,x}$, so that the typical size of the field $\phi_{j;x}$ is of
order $L^{-j[\phi]}$.  Moreover, \eqref{e:scaling-estimate} also
indicates that the derivative $\nabla_{x_k} \phi_{j;x}$ is typically
smaller than the field itself by a factor $L^{-j}$, so that $\phi_j$
remains approximately constant over a distance $L^{j}$. It is less
familiar to think of the size of a differential form, but motivated by
\begin{equation}
  \Ex_{j+1}(\bar\psi_x\psi_y) = \Ex_{j+1}(\bar\phi_x\phi_y) = C_{j+1;xy},
\end{equation}
we can apply the same heuristics as for $\phi$ to the fermion field $\psi$.

\subsection{Polymers and relevant directions}
\label{sec:polymers}

To make a systematic analysis of the fluctuation fields with
covariance satisfying \eqref{e:finite-range} and
\eqref{e:scaling-estimate}, we introduce nested pavings of the torus
$\Lambda = \Lambda_N = \Z^d/L^N\Z^d$ by sets of \emph{blocks}
$\mathcal{B}_j$ on \emph{scales} $j=0,\ldots,N$. See
Figure~\ref{fig:RG_hierarchy1}. The blocks in $\mathcal{B}_0$ are
simply the points in $\Lambda$. The blocks in $\mathcal{B}_1$ form a
disjoint paving of $\Lambda$ by boxes of side $L$.  More generally,
each block in $\mathcal{B}_j$ has side length $L^j$ and consists of
$L^d$ disjoint blocks in $\mathcal{B}_{j-1}$. The following definition
makes this more formal, and also introduces the concept of a
\emph{polymer}, which has a long history in statistical mechanics
going back to the important paper \cite{GK71} (these polymers have
nothing to do with long chain molecules or random walks).

\begin{figure}
\begin{center}
\includegraphics[scale = 0.4]{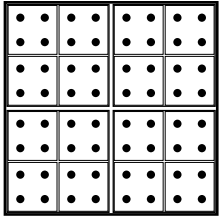}
\end{center}
\caption{\lbfg{reblock} Illustration of ${\cal B}_j(\Lambda)$ for
$j=0,1,2,3$ when $d=2$, $N=3$, $L=2$.}
\label{fig:RG_hierarchy1}
\end{figure}

\begin{defn}
\label{def:blocks}
For each $j=0,1,\ldots,N$, the torus $\Lambda$ is paved in a natural
way by $L^{N-j}$ disjoint $d$-dimensional cubes of side $L^j$.  The cube
that contains the origin at the corner has the form
\begin{equation}
    \{x\in \Lambda:  x_{i} = 0,1,\dots ,L^{j}-1 \; \forall \, i=1,\dots ,d\}
\end{equation}
and all the other cubes are translates of this one by vectors in
$L^{j} \Zd$.  We call these cubes $j$-{\em blocks}, or {\em blocks}
for short, and denote the set of $j$-blocks by ${\cal B}_j= {\cal
B}_j(\Lambda)$.
A union of $j$-blocks is called a {\em polymer} or $j$-{\em
polymer}, and the set of $j$-polymers is denoted ${\cal P}_j={\cal
P}_j(\Lambda)$.
\end{defn}

For a block $B \in \mathcal{B}_j$, the considerations
in Section~\ref{sec:fieldsize} concerning
the typical size of $\phi_j$ suggest that, at each of the $L^{dj}$
points $x\in B$, $\phi_{j,x}$ has typical size of order $L^{-j[\phi]}$, and
hence
\begin{equation} \label{e:monomrel}
  \sum_{x\in B} \phi_{j,x}^p \approx L^{dj} L^{-pj[\phi]}
  = L^{(d-p[\phi])j}.
\end{equation}
The above right-hand side is \emph{relevant} (growing exponentially in $j$) for
$p[\phi]<d$, \emph{irrelevant} (decaying exponentially in $j$) for
$p[\phi]>d$, and \emph{marginal} (neither growing or decaying) for
$p[\phi]=d$.  Since $\tau_x = \phi_x\bar\phi_x +\psi_x\bar\psi_x$ is
quadratic in the fields, it corresponds to $p=2$, so
$p[\phi]=2[\phi] = d-2 < d$ and $\tau_x$ is relevant in all
dimensions. Similarly, $\tau_x^2$ corresponds to $p=4$ with
$p[\phi]=4[\phi]=2d-4$, so $\tau_x^2$ is irrelevant for $d>4$,
marginal for $d=4$, and relevant for $d<4$.
Also, $\tau_{\Delta,x}$ is marginal in all dimensions.
In fact, by \cite[Lemma~\ref{pt-lem:ss}]{BBS-rg-pt},
any sum of local field monomials in $\phi$ and $\psi$ that are relevant or marginal,
Euclidean invariant,
and that obey the additional symmetry between boson and fermions called
\emph{supersymmetry}
(see \cite[Section~\ref{pt-sec:supersymmetry}]{BBS-rg-pt} for details)
is of the same form as $V_0$ of \refeq{Vtil0def}, plus an additional term
$\tau_{\nabla\nabla}$
(the use of summation by parts to avoid inclusion of the term
$\tau_{\nabla\nabla}$
is discussed in \cite[Section~\ref{pt-sec:ourflow}]{BBS-rg-pt}).

Essentially based on this observation,
the general approach of Wilson \cite{WK74} suggests that, for $d=4$
the map $Z_j \mapsto Z_{j+1}$
is qualitatively well approximated by a $3$-dimensional map
$(g_j,z_j,\mu_j) \mapsto (g_{j+1},z_{j+1},\mu_{j+1})$ such that
\begin{equation} \label{e:Zjapprox}
  Z_j  \approx e^{-V_j(\Lambda)}
\end{equation}
where $V_j$ has the same form as $V_0$ with \emph{renormalised} coupling constants $g_j,z_j,\mu_j$.
Furthermore, Wilson's general approach suggests
that, given $g_0>0$, one can find $z_0,\mu_0$ such that $(g_j,z_j,\mu_j) \to 0$,
i.e., $Z_j \approx 1$ as $j \to\infty$.
This is the famous observation of \emph{infrared asymptotic freedom}.
In the context of Section~\ref{sec:ga}, it
corresponds to the approximation of the interacting model by the non-interacting model.
Our method gives a rigorous meaning to the approximation,
by introducing a codimension-$3$ error coordinate $K_j$ 
which rigorously keeps track of all errors in the approximation.
The introduction of $K_j$ trades the \emph{linear but nonlocal} evolution of $Z_j$
given by \eqref{e:Z0def}
for a \emph{nonlinear but local} evolution of $(V_j,K_j)$.
The evolution of $(V_j,K_j)$ is addressed in Sections~\ref{sec:rg-map}--\ref{sec:bulk-flow}
below.

Our construction of the evolution map in \cite{BS-rg-step} makes important use of
the finite-range property of the covariance decomposition.
This can be contrasted with
the \emph{block spin} method
\cite{GK85,GK86},
in which the
fluctuation covariances $C_j$ are chosen such the fields
$\phi_j$ are constant over the blocks in $\Bcal_j(\Lambda)$.
Block spin covariances decay exponentially, but do not have the finite-range property \eqref{e:finite-range}.
In our setup, fields are only \emph{approximately} constant over blocks by \eqref{e:scaling-estimate},
but this is compensated by the independence property \eqref{e:Exfac} which
allows for an effective construction of a renormalisation group map,
by using independence rather than cluster expansion.

\subsection{Interaction functional}
\label{sec:if}

Motivated by the discussion in Section~\ref{sec:polymers},
we define $\Vcal = \R^3$ and identify $\Vcal$ with the vector space of
\emph{local field polynomials}
\begin{align}
    \label{e:Vterms}
    &
    V
    =
    g \tau^{2} + \nu \tau + z \tau_{\Delta}
    ,
\end{align}
where $(g,\nu,z)  \in \R^3$.
Rather than the coupling constant $\nu = \nu_j$ we often use its rescaled version
\begin{equation} \label{e:munu}
    \mu_j = L^{2j} \nu_j
    .
\end{equation}
The reason for this rescaling is that, according to the heuristics
discussed around \eqref{e:monomrel}, the sizes of
$\tau^2,\tau_{\Delta}$ are comparable with that of $L^{-2j}\tau$ rather than with
that of $\tau$.
There is no distinction between $\mu_0$ and $\nu_0$, so the
two are interchangeable.
The parameters $(g,z,\nu)$ in $V$ are designed to track the
relevant and marginal directions in the dynamical system determined
by the renormalisation group flow.
The norm we use on $\Vcal$
is the scale-dependent
maximum norm on $\R^3$ with respect to the rescaled coordinate $\mu$,
i.e.,
\begin{equation}
\lbeq{Vcalnorm}
    \|V\|_{\Vcal} = \max\{ |g|, |z|, |\mu| \} = \max\{ |g|, |z|, L^{2j} |\nu| \}.
\end{equation}

Given a polymer $X \in \Pcal_j(\Lambda)$ as in Definition~\ref{def:blocks},
we write $V(X)$ for the form
\begin{equation} \label{e:VX}
  V(X) = \sum_{x\in X} (g \tau_x^2 + \nu \tau_x + z \tau_{\Delta,x}).
\end{equation}
In \cite[Section~\ref{pt-sec:Vpt-def}]{BBS-rg-pt},
for $X \in \Pcal_j(\Lambda)$
we define a field functional $W_j(V,X) \in \Ncal$
as an explicit quadratic function of $V \in \Vcal$.
Given $V \in \Vcal$,  $B \in {\cal B}_j(\Lambda)$, and $X \in {\cal P}_j(\Lambda)$,
we define the interaction functional
$I_j : \Pcal_j(\Lambda) \to \Ncal(\Lambda)$ by
\begin{align}
\label{e:Idef}
    I_{j} (V,B)
    & = e^{-V(B)}(1+W_j(V,B)),
    \quad\quad
    I_{j} (V,X)
    =
    \prod_{B\in \Bcal_{j}(X)}I_{j} (V,B)
    .
\end{align}
For simplicity, we often write $I_{j} (X) = I_{j} (V,X)$, with the
coordinate
$V$ left
implicit.
The precise definitions of $W_j(V,X)$
and $I_{j}(V,B)$ do not play a major role in this paper, but they do play a role
in results from \cite{BBS-rg-pt,BS-rg-IE,BS-rg-step} upon which our analysis depends.
The interaction functional $I_j$ has the following important properties.
  In their statement, the enlargement $B^+$ of a block $B$ is the union of $B$ and its neighbouring blocks.
\begin{itemize}
\item
  Field locality: $I_j(B) \in \Ncal(B^+)$
  for each block $B \in \Bcal_j(\Lambda_N)$;
\item
  Symmetry: $I_j$ is
  supersymmetric
  and Euclidean invariant;
\item
  Block factorisation: $I_j(X) = \prod_{B \in \Bcal_j(X)} I_j(B)$.
\end{itemize}
By definition, $W_0=0$, so $I_{0} (V,X)=e^{-V (X)}$ for all $X \subset \Lambda$.
With $V_0$ defined by \eqref{e:Vtil0def}, it follows from \eqref{e:Z0def} that
\begin{equation}
  Z_0 = I_0(\Lambda) = I_0(V_0, \Lambda).
\end{equation}

The definition \eqref{e:Idef} is motivated in
\cite[Section~\ref{pt-sec:WPjobs}]{BBS-rg-pt}.
In short, given $V$ as in \eqref{e:Vterms}, it is shown in
\cite[Section~\ref{pt-sec:WPjobs}]{BBS-rg-pt}
that there is an explicit choice $\Vpt$ of the same form as $V$
but with renormalised coupling constants,  such that
as a formal power series in $V$,
\begin{equation}
\label{e:ExIEapprox}
  \Ex_{j+1}\theta I_j(V,\Lambda) = I_{j+1}(\Vpt,\Lambda) + O(V^3).
\end{equation}
This shows that the form of $I$ is stable in the sense that
integrating out a fluctuation field 
is approximately the
same as replacing $V$ by $\Vpt$.
However, the error in \eqref{e:ExIEapprox} fails to be uniform in the
size of $\Lambda$.
To obtain an approximation that is uniform in $\Lambda$, we use the
localised form \refeq{Idef}.

The factorised version \eqref{e:Idef}, together with our choice of the
side length of $B$ to be larger than the range of the covariance,
allows us to take advantage of independence of fields on polymers that
do not touch.
An expression of this is \eqref{e:Exfac}, which implies that whenever
$X_1,X_2 \in \Pcal_j(\Lambda)$ do not touch, and $F(X_i) \in
\Ncal(X_i)$ for $i=1,2$, then
\begin{equation} \label{e:Efac}
  \Ex_{j} \theta (F(X_1)F(X_2)) = (\Ex_j\theta F(X_1))(\Ex_j\theta F(X_2))
  .
\end{equation}
This factorisation property can be combined with a perturbative calculation
on individual polymers, together with a careful control of errors
in the perturbative calculation.  Large parts of
\cite{BBS-rg-pt,BS-rg-IE,BS-rg-step} are concerned
with such matters.

\subsection{Circle product and error coordinate}
\label{sec:circ}

Given $F_1, F_2 :{\cal P}_j \to \Ncal$, we define their \emph{circle product}
$F_1\circ F_2 :{\cal P}_j \to \Ncal$ by
\begin{equation} \label{e:circprod}
    (F_1 \circ F_2)(Y) = \sum_{X\in {\cal P}_j: X \subset Y} F_1(X) F_2(Y \setminus X),
    \quad \text{for $Y \in \Pcal_j$.}
\end{equation}
The terms corresponding to $X=\varnothing$ and $X=Y$ are included
in the summation on the right-hand side, and we assume that
$F(\varnothing)=1$ for all $F: \Pcal_j \to \Ncal$.
The circle product depends on the scale $0 \leq j \leq N$, but
we leave this dependence implicit.
It is an associative and commutative product, since $\Ncal$ has these properties.
The identity element for the circle product is
$\1_\varnothing(X) = \1_{X=\varnothing}$, i.e.,
$(F\circ \1_{\varnothing})(Y) = F(Y)$ for all $F$ and $Y$.

Let $K_0 : \Pcal_0 \to \Ncal$ be the identity element $K_0 = \1_{\varnothing}$.
Then $Z_{0}=I_0(V_0,\Lambda)$ of \eqref{e:Z0def} is
also given by
\begin{equation}
\label{e:Zinit}
    Z_0 = I_0(\Lambda) =  (I_0 \circ K_0)(\Lambda)
    .
\end{equation}
Our procedure is to maintain this form,
\begin{gather}\label{e:ZIK}
    Z_j =  (I_j \circ K_j)(\Lambda)
\quad \text{with} \quad
    K_{j}:\Pcal_{j}(\Lambda) \rightarrow \Ncal(\Lambda),
\end{gather}
in the recursion $Z_j \mapsto Z_{j+1} = \Ex_{j+1}\theta Z_j$ of
\eqref{e:Z0def}, with the initial condition given by \eqref{e:Zinit}.
With $I_j=I_j(V_j)$,
the action of $\Ebold_{j+1}\theta$ on $Z_{j}$ is then expressed as
the \emph{renormalisation group map}
\begin{equation}
\label{e:RGmap}
     ( V_j, K_j) \mapsto ( V_{j+1},K_{j+1}).
\end{equation}
The map \refeq{RGmap} depends on the scale $j$ and the volume parameter $N$.
To achieve this, given
$V_j \in \Vcal$ and $K_j : \Pcal_j(\Lambda) \to \Ncal(\Lambda)$,
we seek to define  $V_{j+1}\in \Vcal$ and
$K_{j+1} : \Pcal_{j+1}(\Lambda) \to \Ncal(\Lambda)$ in such a way that
\begin{equation}
    \label{e:Kspace-objective}
    Z_{j+1}
    =
    \Ex_{j+1}\theta (I_j \circ K_j)(\Lambda)
    =
    (I_{j+1} \circ K_{j+1})(\Lambda)
    .
\end{equation}
Then $Z_{j} = (I_{j}\circ K_{j})(\Lambda)$ retains its form under progressive integration.
The representation \refeq{ZIK} is by no means unique: given $Z_j$ there are many
choices of  $I_j,K_j$
such that \eqref{e:Kspace-objective} holds. For example, a trivial choice is given by
$K_j(X) = \1_{X=\Lambda} Z_j(X)$ along with \emph{any} $I_j$ for which
$I_j(\varnothing)=1$.

It is crucial to make a careful choice of the representation
in order to obtain useful estimates valid at all scales.
A choice of the map $(V_j,K_j) \mapsto (V_{j+1},K_{j+1})$ that leads to good estimates
is explicitly constructed in \cite{BS-rg-step}.
This choice is such that $K_j = O(\|V_j\|^3)$ holds in a certain precise sense, so
that $K_j$ can be regarded as a third-order error term.
We discuss this in more detail in the next section, deferring substantial
details to \cite{BS-rg-step}.

Substantial details underlying our analysis are
developed in the series
of papers \cite{BBS-rg-flow,BBS-rg-pt,BS-rg-norm,BS-rg-loc,BS-rg-IE,BS-rg-step}.
One important aspect, which is worthy of mention despite not playing a visible
role in the present paper, is the operator $\LT$ defined and analysed in \cite{BS-rg-loc}.
The operator $\LT$ extracts from an arbitrary element of $\Ncal$ its relevant and marginal
parts, in the form of a local polynomial in the fields.
Although
$\LT$ does not make a direct appearance in the present paper, it plays a crucial role in
results from
\cite{BBS-rg-pt,BS-rg-IE,BS-rg-step} upon which we rely here.
The local polynomial $V_{j+1}$ is created from $(V_j,K_j)$ using $\LT$, which is used
to incorporate the relevant and marginal parts of $W_j$ and
$K_j$ into $V_{j+1}$.

\section{Renormalisation group map}
\label{sec:rg-map}

In this section, we discuss the definition and properties of the
renormalisation group map $(V_j, K_j) \mapsto (V_{j+1},K_{j+1})$ in
more detail.  Our discussion summarises results of
\cite{BBS-rg-pt,BS-rg-IE,BS-rg-step}, where further details can be
found.  The \emph{observable} fields $\sigma,\bar\sigma$ that appear
in those references play no role in this paper, and can be set equal
to zero for our present needs.  The observable fields are needed for
the analysis of critical correlation functions in
\cite{BBS-saw4,ST-phi4}.

Both \emph{local} and \emph{global} aspects are important for the
repeated application of the renormalisation group map.  In this
section, we focus on the local aspect, which concerns a single
application under appropriate assumptions on $(V_j,K_j)$.
We give precise meaning to the notion that $V_j$ is a second-order
approximation of $Z_j$ and that $K_j$ is a third-order remainder.
The global aspect is discussed in Section~\ref{sec:bulk-flow}; this is
the requirement that the assumptions on $(V_j,K_j)$ continue to remain
valid as the two scales $j$ and $N$ increase indefinitely.  It is the
global aspect that requires the careful choice of the initial
parameters $z_0$ and $\nu_0$ in Theorem~\ref{thm:suscept-diff}.

\subsection{Perturbative quadratic flow}
\label{sec:flow-approx}

In \cite[\eqref{pt-e:varphiptdef}]{BBS-rg-pt}, an explicit quadratic
map $\varphi_{\pt,j}^{\smash{(0)}}: \R^3 \to \R^3$ is defined and
discussed.  This map defines $V_{\pt,j+1}^{\smash{(0)}} =
\varphi_{\pt,j}^{\smash{(0)}}(V_j)$ in such a way that the formal
power series identity \refeq{ExIEapprox} is satisfied.  The map
$\varphi^{\smash{(0)}}_{\pt,j}$ is a second-order approximation to the
$V$-component of the full renormalisation group map, $(V_j, K_j)
\mapsto (V_{j+1}, K_{j+1})$.
It gives the dominant contribution to the renormalisation group map,
and it must be understood in detail.

It is useful to re-express
the $V$-component of the renormalisation group map
in \emph{transformed coordinates}.
In \cite[Proposition~\ref{pt-prop:transformation}]{BBS-rg-pt}, we define
quadratic polynomials $T_j : \R^3 \to \R^3$ and consider the
change of variables
$V_j \mapsto \smash{\Vch_j}=T_j(V_j)$.
The transformation $T_j$ satisfies
\begin{equation} \label{e:TVV2}
  T_0(V) = V, \quad\quad T_j(V) = V + O(\|V\|^2),
\end{equation}
with error estimate uniform in $j$.
Since the $T_j$ are polynomials,
this implies that they are invertible in a neighbourhood of $0$ that is independent of $j$.
The composite maps
$T_{j+1}\circ\varphi_{\pt,j}^{\smash{(0)}}\circ T_j^{-1}$
are equal, up to an error $O(\|V\|^3)$, to $\bar\varphi_j: \R^3 \to \R^3$,
given explicitly
by $(\gbar_{j+1},\zbar_{j+1},\mubar_{j+1}) = \bar\varphi_j(\gbar_j,\zbar_j,\mubar_j)$ with
\begin{align}
\label{e:gbar}
  \gbar_{j+1}
  &=
   \gbar_j
  -
  \beta_j  \gbar_j^{2}
  ,
  \\
\label{e:zbar}
  \zbar_{j+1}
  &=
  \zbar_j
  - \theta_j  \gbar_j^{2}
  ,
  \\
\label{e:mubar}
  \mubar_{j+1}
  &=
  L^2 \mubar_j ( 1- \gamma \beta_j \gbar_j)
  + \eta_j  \gbar_j
  -
  \xi_j  \gbar_j^{2}
   - \pi_j  \gbar_j  \zbar_j,
\end{align}
where $\beta_j,\theta_j,\eta_j,\xi_j,\pi_j$
are real coefficients defined precisely in
\cite[\eqref{pt-e:nuplusdef},\,\eqref{pt-e:betadef}--\eqref{pt-e:xipidef}]{BBS-rg-pt}.
These coefficients, and also those of the transformations $T_j$,
depend continuously on the mass $m^2$ appearing in the covariance decomposition of
Section~\ref{sec:decomposition},
but are independent of the size  $N$ of the torus.

The effect of the transformation $T_j$ is to \emph{triangularise} the evolution equation to second order:
the  $\gbar$-equation does not depend on $\zbar$ or $\mubar$, the $\zbar$-equation depends only on $\gbar$,
and the $\mubar$-equation depends both on $\gbar$ and $\zbar$.
This second-order triangularisation is the natural coordinate system
to study the evolution of $V_j$.
For example, we emphasise the explicit occurrence of the parameter
$\gamma = \frac{1}{4}$ in \eqref{e:mubar} which, in
Section~\ref{sec:pfmr}, gives rise to the power of the logarithm in
Theorem~\ref{thm:suscept}. This parameter is not apparent in the
original equations before transformation by $T_j$ (cf.\
\cite[\eqref{pt-e:nupta}]{BBS-rg-pt} where an additional term
proportional to $g\mu$ appears).

The sequence $(\beta_j)_{0\le j < \infty}$ plays a key role in the
analysis.  It is defined in \cite[\eqref{pt-e:betadef}]{BBS-rg-pt} by
\begin{equation} \label{e:betadef}
  \beta_j =
  8\sum_{x \in \Z^d} \left( w_{j+1,x}^2 -  w_{j,x}^2\right) ,
  \quad \text{with } w_{j,x} = \sum_{i=1}^j C_{i;0,x}.
\end{equation}
The $m^2$-dependence of $\beta_j$ is suppressed in the notation.
Since each $C_i$ is positive-definite, the above definition implies
that $\beta_j > 0$ for all $j$.  Moreover, it is shown in
\cite[Lemma~\ref{pt-lem:betalim}]{BBS-rg-pt} that
$\lim_{j\to\infty}\beta_j =\log L /\pi^2$ for $m^2=0$, so $\beta_j$ is
bounded away from $0$ for sufficiently large $j$.  On the other hand,
for $m^2 >0$, by \refeq{scaling-estimate}, $\beta_j$ decays extremely
rapidly to $0$ for $j \geq j_m$ where $j_m$ is the \emph{mass scale},
i.e., the smallest $j$ such that $L^{2j}m^2 \geq 1$.  The mass scale
is natural for our specific sequence \refeq{betadef}, but we apply
below a general dynamical system result from \cite{BBS-rg-flow}, in
which there is no explicit mass parameter with which to define a mass
scale.

To prepare for the application of the main result of \cite{BBS-rg-flow},
given $\Omega >1$ we define
the
$\Omega$-\emph{scale} $\jm \in \N \cup \{\infty \}$, by
  \begin{equation}
  \lbeq{mass-scale}
    \jm = \jm(m^2) = \inf \{ k \geq 0: |\beta_j| \leq \Omega^{-(j-k)} \|\beta\|_\infty
    \text{ for all $j$} \}
    .
  \end{equation}
For the remainder of the paper, we fix $\Omega > 1$ arbitrarily, e.g., $\Omega=2$.
According to \cite[Proposition~\ref{pt-prop:rg-pt-flow}]{BBS-rg-pt}, $j_m$ and
$\jm$ are equivalent for the sequence \refeq{betadef}, in the sense that
$|j_m-j_\Omega|$ is bounded uniformly as $m^2 \downarrow 0$.  Thus we may regard
the mass scale and $\Omega$-scale as essentially equivalent.  We define
\begin{equation}
\lbeq{chidef}
    \chi_j = \chi_j(m^2) = \Omega^{-(j-\jm)_+}.
\end{equation}
In \cite[Proposition~\ref{pt-prop:rg-pt-flow}]{BBS-rg-pt}, it is verified that the coefficients $\theta_j,\eta_j,\xi_j,\pi_j$
are bounded by $O(\chi_j)$
and that they depend continuously on $m^2$. In particular, it is shown that the following two
assumptions are satisfied.

\medskip\noindent \textbf{Assumption~(A1).}
  \emph{The sequence $\beta$:}
  The sequence $(\beta_j)$ is bounded, namely
  $\betamax = \sup_{j\in\N_0} |\beta_j| < \infty$.
  There exists $c>0$ such that $\beta_j \ge c$ for all but $c^{-1}$ values of $j \le \jm$.

\medskip\noindent \textbf{Assumption~(A2).}
  \emph{The other parameters of $\bar\varphi$:}
  Each of
  $\theta_j$, $\eta_j$, $\xi_j$, and $\pi_j$
  is bounded in absolute value by $O(\chi_j)$,
  with a constant that is independent of both $j$ and $\jm$.
\medskip

The above
Assumption~(A1) is the same as \cite[Assumption~(A1)]{BBS-rg-flow}, and
the above Assumption~(A2) is a specialised form of \cite[Assumption~(A2)]{BBS-rg-flow}
(where in the notation of \cite{BBS-rg-flow},
$\zeta_j = 0$, $\upsilon_j^{\smash{gg}} = \xi_j$, $\upsilon_j^{\smash{gz}} = \pi_j$,
$\upsilon_j^{\smash{g\mu}} = L^2 \gamma \beta_j$,
and $\upsilon_j^{\smash{zz}} = \upsilon_j^{\smash{z\mu}} = 0$).
We apply Assumptions~(A1--A2)
in the proof of Proposition~\ref{prop:flow-flow} below,
where we apply the main result of \cite{BBS-rg-flow}.

\medskip

We interpret the \emph{scale} parameter $j$ as \emph{time},
and thus view $\bar\varphi = (\bar\varphi_j)_{j\in\N_0}$ as a time-dependent dynamical system on $\R^3$ which we call the \emph{quadratic flow}.
The exponential decay of the coefficients beyond $\jm$ has the interpretation that the evolution essentially stops at $j=\jm$.
The dynamical  system $\bar\varphi_j$ has the fixed point $(0,0,0)$.
This fixed point is not hyperbolic:
there are two unit eigenvalues of $D\bar\varphi_j(0)$ corresponding to the $g$- and $z$-equations.
In \cite{BBS-rg-flow}, such non-hyperbolic fixed points
are studied in a general infinite-dimensional setting.
The flow of $\bar\varphi$ near the fixed point $0$ is elementary, and can be understood
via explicit analysis using the triangular structure of the equations
\eqref{e:gbar}--\eqref{e:mubar}.
In particular, the following proposition is a special case of
\cite[Proposition~\ref{flow-prop:approximate-flow}]{BBS-rg-flow}.
In its statement, $z_\infty$ denotes the limit $z_\infty = \lim_{j \to\infty}z_j$,
and similarly for $\mu_\infty$.

\begin{prop} \label{prop:approximate-flow}
  If $\gbar_0>0$ is sufficiently small,
  then there exists a unique global flow $\bar V = (\bar V)_{j\in \N_0} =
  (\gbar_j, \zbar_j, \mubar_j)_{j\in\N_0}$
  of $\bar\varphi$,
  i.e., a solution to $\bar V_{j+1} = \bar\varphi_j(V_j)$ defined for all $j \in\N_0$, with initial condition
  $\gbar_0$ and
  final condition
  $(\zbar_\infty , \mubar_\infty) = (0, 0)$.
  This flow satisfies, for any real $p \in [1,\infty)$,
  \begin{equation} \lbeq{approximate-flow}
    \chi_j \gbar_j^p
    = O\left(\frac{\gbar_0}{1+\gbar_0j}\right)^p,
    \quad
    \zbar_j = O(\chi_j \gbar_j) ,
    \quad
    \mubar_j = O(\chi_j \gbar_j),
  \end{equation}
  with constants independent of $j_\Omega$ and $\gbar_0$, and dependent on $p$ in
  the first bound.
  Furthermore, $\bar V_j$ is continuously
  differentiable in the initial condition $\gbar_0$
  and continuous in the mass parameter $m^2$,
  for every $j \in \N_0$.
\end{prop}

The main difficulty in the analysis of the
renormalisation group map \refeq{RGmap}
lies in the control of the non-perturbative coordinate
$K_j : \Pcal_j(\Lambda) \to \Ncal(\Lambda)$, and we address this next.

\subsection{Non-perturbative coordinate}
\label{sec:flow-infvollim}

The coordinate $V_j$ is identified with an element of the space $\Vcal = \R^3$, which is independent of $\Lambda$.
The remainder coordinate $K_j$, however, is a map $K_j : \Pcal_j(\Lambda) \to \Ncal(\Lambda)$
and therefore does depend on $\Lambda$.
In this section, we discuss the definition of $K_j$ in more detail,
summarising full details given in \cite{BS-rg-step}.
Of particular importance are locality properties of $K_j$ which allow for
the definition of a natural
limiting space as $\Lambda_N \uparrow \Z^d$.
We require some understanding of this
\emph{infinite volume limit}.
In Section~\ref{sec:fintori}, we use results about the infinite volume limit
to obtain estimates for finite $\Lambda_N$.

To discuss both the finite and infinite volume simultaneously,
let $\volume$ denote either $\Lambda_N$ or $\Z^d$, and set $N(\Lambda_N) = N$ and $N(\Z^d) = \infty$.
We interpret the inequality $j\leq N(\volume)$ to mean that $j < \infty$ if $N(\volume) = \infty$.
Given a finite set $X\subset\volume$, recall that $\Ncal(X)$
denotes the space of even differential forms
defined in Section~\ref{sec:intforms}.
We identify $\Ncal(X)$ as a subalgebra of $\Ncal(Y)$ if $X \subset Y$.
We define $\Ncal(\volume)$ to be the union over $\Ncal(X)$ for all finite subsets $X\subset\volume$.
As in \cite[Definition~\ref{step-def:blocks1}]{BS-rg-step}, we say that a polymer
$X \in \Pcal_j(\volume)$ is a \emph{small set} if it is connected and if $|X|_j$, the number of
$j$-blocks it contains, satisfies $|X|_j \leq 2^d$. We write $\Scal_j$ for the set of all small sets in $\Pcal_j(\volume)$.
Given any polymer $X \in \Pcal_j(\volume)$, we define its
\emph{small set neighbourhood} $X^\square$ by
\begin{equation} \label{e:smallsetnbh}
  X^\square = \bigcup_{Y \in \Scal_j: X \cap Y \neq \varnothing} Y.
\end{equation}
We write ${\rm Comp}_j( X) \subset \Pcal_j(\volume)$ for the set of connected components
of $X\in \Pcal_j(\volume)$.

From \cite[Definition~\ref{step-def:Kspace}]{BS-rg-step},
we recall the definition of the space $\Kspace_j(\volume)$,
as follows.  The non-perturbative coordinate $K_j$ is an element of the space
$\Kspace_j(\Lambda)$.

\begin{defn} \label{def:Kspace}
For $j \le N(\volume)$, let $\Kspace_{j} = \Kspace_{j} (\volume)$ be the
real vector space of functions $K :
\Pcal_j (\volume) \to \Ncal (\volume)$ with the properties,
\begin{itemize}
\item Field locality: $K(X) \in \Ncal (X^{\Box},\volume)$ for each connected $X\in\Pcal_j$,
\item Symmetry: $K$ is supersymmetric and Euclidean invariant,
\item Component Factorisation: $K (X) = \prod_{Y
\in {\rm Comp}_j( X)}K (Y)$ for all $X\in\Pcal_j$.
\end{itemize}
\end{defn}

The symmetries mentioned in Definition~\ref{def:Kspace} are discussed
above \cite[Definition~\ref{step-def:Kspace}]{BS-rg-step}, and also in
\cite{BBS-rg-pt,BS-rg-loc,BS-rg-IE}.  They do not play an explicit
role in this paper, but they are needed in results we use from
\cite{BBS-rg-pt,BS-rg-loc,BS-rg-IE,BS-rg-step}.  We do not discuss
them further here.

The analysis of \cite{BBS-rg-pt,BS-rg-loc,BS-rg-IE,BS-rg-step} permits
the inclusion of \emph{observables} which break certain symmetries and
are employed in \cite{BBS-saw4,ST-phi4} to study the decay of critical
correlation functions.  Observables are not needed in the present
paper. We only need the case $\sigma = \bar\sigma = 0$ of results in
\cite{BBS-rg-pt,BS-rg-loc,BS-rg-IE,BS-rg-step}, we identify $K$ and
$\pi_\varnothing K$, with $\pi_* K = 0$ in the notation of those
papers, and make no further references to observables here.  In
particular, Definition~\ref{def:Kspace} has been specialised to the
case where observables are not present.

\subsection{Norms}
\label{sec:flow-norms}

For the analysis of $K_j$, we employ a family of norms on the spaces $\Kspace_j(\volume)$,
both for $\volume =\Lambda_N$ and $\volume = \Z^d$.  These norms, which are defined
in \cite[\eqref{step-e:9Kcalnorm}]{BS-rg-step},
give rise to Banach spaces $\Wcal_j(\sgen,\volume) \subset \Kspace_j(\volume)$
indexed by a parameter $\sgen = (\mgen^2, \ggen) \in [0,\delta) \times (0,\delta)$
with fixed small $\delta >0$.
The parameter $\sgen$ is chosen depending on the mass $m^2$ and on $g_0$.
The use of $\mgen^2$ allows us the option to
choose the norm parameter close to but not necessarily equal to
the mass $m^2$ of the covariance $(-\Delta +m^2)^{-1}$, which permits variation of
$m^2$ without changing the norm.  The choice of $\ggen$ is discussed below.

In the present paper, we do not need and therefore
do not recall the definition of the norm from
\cite[\eqref{step-e:9Kcalnorm}]{BS-rg-step},
needed for the propagation of estimates of $K$ from one
scale to the next.
We provide some general indication
of the origin of these norms in the following discussion,
with reference to \cite{BS-rg-norm,BS-rg-IE,BS-rg-step} for further details.

We discuss the norms for the case $\volume=\Lambda$, the extension to $\volume = \Z^d$
is discussed in \cite{BS-rg-step}.
Consider first a single differential form $F \in\Ncal(\Lambda)$ which,
as in Section~\ref{sec:formsfunc}, we regard as a function of the boson
field $\phi$ and the fermion field $\psi$. For example, the degree-$0$ part $F^0$ is a function $F^0: \C^\Lambda \to \Lambda$ of $\phi$,
in the ordinary sense. The degree-$1$ part of $F$ can be written as
$\sum_{x\in\Lambda} F_{x}(\phi) \psi_x + \sum_{\bar x \in \Lambda} F_{\bar x}(\phi)\bar\psi_{\bar x}$,
and each term can be viewed as an ordinary function of $\phi$ times a monomial of degree $1$ in $\psi$. The coefficient $F_x(\phi)$
has the interpretation of a derivative of $F$ with respect to $\psi_x$ evaluated at $\psi=0$.
The semi-norm $\|F\|_{T_{0,j}}$ provides a means to measure
the size of these derivatives of $F$ with respect to $\phi$ and $\psi$ at $\phi =0$.
Definitions of these and other norms, and development of their properties,
is given in \cite{BS-rg-norm}.

To prove Theorem~\ref{thm:suscept}, we only make use of the
specialised set-up of \cite[Section~\ref{IE-sec:reg}]{BS-rg-IE},
omitting the parts that concern observables.
In particular, we consider four copies of $\Lambda$, denoted
$\Lambda_b$, $\bar\Lambda_b$, $\Lambda_f$, and $\bar\Lambda_f$,
corresponding to the components of $\phi,\bar\phi,\psi,\bar\psi$,
respectively, and set $\Lambdabold_b =\Lambda_b \sqcup\bar\Lambda_b$,
$\Lambdabold_f = \Lambda_f \sqcup \bar\Lambda_f$, $\Lambdabold =
\Lambdabold_b \sqcup \Lambdabold_f$, and we let $\vec\Lambdabold^*$
denote the set of finite ordered sequences of $\Lambdabold$.  A
\emph{test function} is a function $g: \vec\Lambdabold^* \to \R$.

In \cite[Section~\ref{IE-sec:reg}]{BS-rg-IE}
(see also \cite[Section~\ref{norm-sec:tf}]{BS-rg-norm}), we define linear spaces of test
functions $\Phi_j = \Phi_j(\ell_j)$.
We set $\ell_j = \ell_0 L^{-j[\phi]}$ for an appropriate constant $\ell_0$
(large and $L$-dependent), fix an integer $p_\Phi \ge\frac{1}{2}d+2$, and
fix an integer $p_\Ncal \ge 10$.
The $\Phi_j(\ell_j)$ norm of a test function $g$ is defined by
\begin{equation} \label{e:Phinormdef}
  \|g\|_{\Phi_j(\ell_j)} = \sup_{z\in \vec\Lambdabold^*: p_b(z) \leq p_\Ncal}
  \sup_{\alpha:
  |\alpha|
  \leq p_\Phi}
  \ell_j^{-z}   L^{j|\alpha|}  |\nabla^\alpha g_z|
  ,
\end{equation}
where $p_b(z) \leq p_\Ncal$ means that $z$ has at most $p_\Ncal$ boson
components, $\ell_j^{-z}$ denotes $\ell_j^{-q}$ when $z=(z_1,\ldots,
z_q)$ with each $z_i$ in a copy of $\Lambda$, and the restriction
$|\alpha| \le p_\Phi$ limits the number of finite-difference
derivatives to be at most $p_\Phi$ in each of the $q$ components of
$z$.  The number of fermion components of $z$ is unrestricted.  For
example, for a test function $\phi : {\Lambdabold}_b \to \R$ with only
one boson component,
\begin{equation} \label{e:phiPhinorm}
  \ell_0^{-1} L^{-j[\phi]} L^{j|\alpha|_1}|\nabla^\alpha \phi_x| \leq  \|\phi\|_{\Phi_j}
\end{equation}
for any multi-index $\alpha$ such that $|\alpha|_1 \leq p_\Phi$.
According to the discussion in Section~\ref{sec:fieldsize},
the fluctuation field at scale $j$, i.e., the Gaussian field with covariance $C_j$,
when regarded as a test function in $\Phi_j(\ell_j)$ thus typically has norm $O(1)$.
The scaling used in the definition of the norm has been designed to make this happen.

In \cite{BS-rg-norm}, a pairing between forms and test functions is defined  by
\begin{equation} \label{e:pairing}
  \langle F, g\rangle_0
  =
  \sum_{z \in {\vec\Lambdabold}^*: p_b(z) \leq p_\Ncal} \frac{1}{z!} F_z(0) g_z.
\end{equation}
The pairing can be thought of as a kind of generalised Taylor
expansion, to all orders in the fermion field, and to order $p_\Ncal$
in the boson field.  We define a semi-norm on $\Ncal(\Lambda)$ by
\begin{equation} \label{e:Tphinormdef}
  \|F\|_{T_{0,j}} =
  \|F\|_{T_{0,j}(\ell_j)} =
  \sup_{g\in\Phi_j : \|g\|_{\Phi_j} \le 1} |\langle F, g \rangle_0|.
\end{equation}
The pairing of $F$ with test functions $g$ of norm $1$
is intended to mimic how large $F$ should be when evaluated at
typical \emph{small} fields $(\bar\phi,\phi,\bar\psi,\psi)$, whose size is dictated by the scale-$j$ covariance.

It is also essential to have some control over \emph{large} fields;
this issue is addressed in detail in \cite{BS-rg-IE}.
As motivation, consider the probability measure on $\R$ given by
\begin{equation} \label{e:phi41d}
  \mu(ds) \propto e^{-gs^4} \, ds,
\end{equation}
which has mean $0$ and variance proportional to $g^{-1/2}$.  More informally,
typically $|s| = O(g^{-1/4})$.
The expectation $\Ex_{j+1}\theta(I_j \circ K_j(\Lambda))$ contains a factor $e^{-g_j|\phi|^4}$
which suggests that typically $|\phi| = O(g_j^{\smash{-1/4}})$.
Since we do not know {\it a priori} the value of $g_j$,
we design the norm $\|K_j\|_{\Wcal_j(\sgen,\volume)}$ so that,
in addition to taking small fields into account, it also measures the size of
$K_j$ when evaluated on
fields of size
$|\phi| = O(\ggen_j^{-1/4})$ with $\ggen_j$ an approximate guess
for $g_j$.
We then show \emph{a posteriori} that the guess does what is required.
The sequence
$\gbar_j = \gbar_j(m^2,g_0)$ determined by $\gbar_0=g_0$ and the recursion \eqref{e:gbar}
is a natural candidate for $\ggen_j$, but is one
that introduces dependence of the norms on the two parameters $m^2$ and $g_0$.
The dependence on $g_0$ cannot be avoided, but it is convenient to use a family of norms
that minimises the dependence on $m^2$.
To accomplish this, we set
\begin{equation} \label{e:ggendef}
  \ggen_j(m^2,g_0) =
  \gbar_j(0,g_0) \1_{j \le j_m} + \gbar_{j_m}(0,g_0) \1_{j > j_m},
\end{equation}
where the mass scale $j_m$ is the smallest integer $j$ such that $L^{2j}m^2 \ge 1$.
Thus $\ggen_j(m^2,g_0)$ is the massless sequence $\gbar_j(0,g_0)$ with its evolution shut down at
the mass scale.
We show in Lemma~\ref{lem:gbarmcomp} that
$\ggen_j(m^2,g_0) = \gbar_j(m^2,g_0) +O(\gbar_j^2(m^2,g_0))$, so the
sequences $\ggen_j$ and $\gbar_j$ are close to being equal.

A property of the norms that we use in
the proof of Theorem~\ref{thm:suscept-diff} (see Section~\ref{sec:suscept-diff-pf})
is the fact given by \cite[\eqref{step-e:T0dom}]{BS-rg-step} that
the $\Wcal_N$ norm dominates the $T_{0,N}$ semi-norm,
in the sense that
\begin{equation} \label{e:T0dom}
  \|F(\Lambda)\|_{T_{0,N}} \leq \|F\|_{\Wcal_N(\sgen,\Lambda)}
  .
\end{equation}

\subsection{Specification of renormalisation group map}
\label{sec:rgiv}

For both choices $\volume = \Lambda_N$ and $\volume=\Zd$,
explicit definitions of maps \eqref{e:RGmap}  are constructed in \cite{BS-rg-step}.
For $\volume = \Lambda$, the maps achieve the objective
\eqref{e:Kspace-objective}.  The significance for $\volume=\Zd$ is discussed
in Section~\ref{sec:rg-iv} below.
To simplify the notation, we write
\refeq{RGmap} as $(V,K) \mapsto (V_+,K_+)$, typically dropping subscripts $j$ and writing $+$ in place of $j+1$,
and also leaving the dependence of the maps on the mass parameter $m^2$ implicit.
As discussed in Section~\ref{sec:flow-norms}, it is necessary to make
assumptions on $(V,K)$, in order for these definitions to be well-defined and to obtain
useful estimates for $(V_+,K_+)$.
These assumptions are stated in terms of domains, defined as follows.

Given $C_\DV>1$, $\DVa>0$, $\delta >0$, and
$\sgen=(\mgen^2, \ggen) \in  [0,\delta) \times (0,\delta)$, let
\begin{equation} \label{e:domRG}
  \domRG_j(\sgen, \volume)
  =
  \{ (g, z, \nu): C_\DV^{-1} \ggen < g < C_\DV \ggen, \;
  |z|, |\mu| < C_\DV \ggen \}
  \times
  B_{\Wcal_j(\sgen, \volume)}(\DVa \chigen_j\ggen^3)
  ,
\end{equation}
where $B_X(r)$ is the open ball of radius $r$ around the origin in the
Banach space $X$, and $\chigen_j = \chi_j(\mgen^2)$ with $\chi_j$
defined by \refeq{chidef}.  Roughly speaking, the domain
\eqref{e:domRG} permits small $g > 0$ that is bounded away from zero,
with $z,\mu = O(g)$, and with $K$ bounded in a precise but non-trivial
fashion by $O(g^3)$.  The domain $\domRG_j(\sgen,\volume)$ is equipped
with the norm of $\Vcal \times \Wcal_j(\sgen,\volume)$.

The parameter $\mgen^2$ that appears in the definition of the domain
\refeq{domRG} would ideally be set equal to the mass $m^2$ appearing
in the covariance, but this would be problematic for the discussion of
continuity of the renormalisation group map as a function of $m^2$.
Thus we decouple $m^2$ from the domain by using $\mgen^2$ in the
domain and requiring $m^2$ to be close to but not necessarily equal to
$\mgen^2$.  For this, we recall the interval $\Iint_j$ defined in
\refeq{massint}, and set
\begin{equation}
\lbeq{Itilint}
    \Igen_j = \Igen_j(\mgen^2) =
    \begin{cases}
    [\frac 12 \mgen^2, 2 \mgen^2] \cap \Iint_j & (\mgen^2 \neq 0)
    \\
    [0,L^{-2(j-1)}] \cap \Iint_j & (\mgen^2 =0).
    \end{cases}
\end{equation}
For $j<N$, these intervals are illustrated in Figure~\ref{fig:intervals}.
\begin{figure}
  \begin{center}
    \input{intervals.pspdftex}
  \end{center}
  \caption{The dashed line illustrates the interval $\Iint_j = [0,\delta)$ for $j < N$.
      The intervals $\Igen_j (\mgen^2)$ are given by $[\frac12 \mgen^2,2\mgen^2]$ if $\mgen^2 \in (0,\frac12 \delta)$
      and are illustrated below the dashed line. As $\mgen^2 \downarrow 0$, the length of these
      intervals shrinks to $0$, but $\Igen_j(0)$ is given by the non-empty interval $[0,L^{-2(j-1)}]$.
    }
\label{fig:intervals}
\end{figure}

For $\sgen = (\mgen^2,\ggen) \in [0,\delta)\times (0,\delta)^2$, the maps
$V_+ = V_{+,\volume}$ and $K_+ = K_{+,\volume}$ with mass $m^2 \in \Igen_+(\mgen^2)$
are maps
\begin{equation}
  \label{e:VKplusmap}
  V_+
  : \domRG(\sgen, \volume)  \to \Vcal,
  \quad\quad
  K_+
  : \domRG(\sgen, \volume)  \to \Kspace_{+}(\volume),
\end{equation}
such that, given $(V,K) \in \domRG_{j}(\sgen, \Lambda)$,
\begin{equation}
\lbeq{EIK}
  \Ex_{+}\theta (I(V) \circ K)(\Lambda) = (I_{+}(V_+(V,K)) \circ K_+(V,K))(\Lambda).
\end{equation}
There is no analogue of \refeq{EIK} for the case $\volume = \Zd$, and we postpone
discussion of infinite volume to Section~\ref{sec:rg-iv}.
The \emph{renormalisation group map} is defined by \refeq{VKplusmap}.
As mentioned in Section~\ref{sec:circ}, there are trivial choices of maps
that make \refeq{EIK} hold.  The maps of \cite{BS-rg-step} provide a nontrivial and
good choice, a choice which obeys the useful
estimates discussed below.

In the present paper, we do not need or use the precise definitions of
the maps $(V_+, K_+)$, so we discuss only their existence and certain
of their important properties from \cite{BS-rg-step}.  Since the maps
are given by explicit formulas, not depending on the parameter $\sgen$
appearing their domains in \eqref{e:VKplusmap}, the maps viewed as
maps on different domains coincide on their intersections, and we will
apply this property without further comment.

The map $V_+$ is a perturbation of the map $\varphi_\pt^{(0)}$
discussed in Section~\ref{sec:flow-approx}, and it is convenient to describe it
in terms of the difference
\begin{equation}
  R_+(V,K) = V_+(V,K) - \varphi_\pt^{(0)}(V).
\end{equation}
The following theorem, which is illustrated by
Figure~\ref{fig:Kplusdom}, provides estimates for the maps $R_+,K_+$,
combining \cite[Theorems~\ref{step-thm:mr-R}--\ref{step-thm:mr},
\ref{step-thm:Kmcont}]{BS-rg-step} into a single statement.
In particular, we apply the result of
\cite[Theorem~\ref{step-thm:mr-R}]{BS-rg-step}
in the form given in \cite[\eqref{step-e:Rmain-g}]{BS-rg-step}.
For the statement of the theorem, we view $R_+,K_+$ as maps jointly on
$(V,K,m^2) \in \domRG(\sgen) \times \Igen_+(\mgen^2)$.
The $L^{p,q}$ norm is the operator norm of a multi-linear operator
from $\Vcal^p \times \Wcal^q$ to $\Vcal$ or to $\Wcal_{+}$, for $R_+$
or $K_+$, respectively, and $\chigen=\chi_j(\mgen^2)$ where $\chi_j$
is given by \eqref{e:chidef}.  We have replaced occurrences of
$\chi^{3/2}$ and $\chi^{1/2}$ from \cite{BS-rg-step} by $\chi$ here;
this amounts to an unimportant redefinition of the parameter $\Omega$.
Also, in \refeq{RKplusmaps}, $\sgen_+=(\mgen^2,\ggen_+)$ for any
choice of $\ggen_+ \in [\frac 12 \ggen, 2\ggen]$ (as required by
\cite[\eqref{step-e:gbarmono}]{BS-rg-step}).  It is straightforward to
verify that the sequence defined by \refeq{ggendef} obeys this
constraint.  For application of a single renormalisation group step in
Theorem~\ref{thm:step-mr-fv}, there is considerable flexibility in the
definition of $\ggen$.  However, indefinite iteration of the
renormalisation group map requires a careful choice of the sequence
$\ggen_j$, and \refeq{ggendef} provides a choice that works.

\begin{figure}
  \begin{center}
    \input{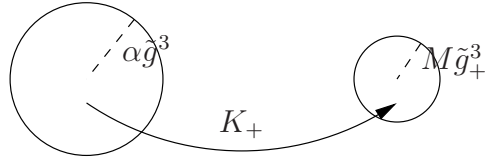}
  \end{center}
  \caption{
  The map $K_+$ of Theorem~\ref{thm:step-mr-fv}
  maps the centred ball of radius
      $\DVa\ggen^3$ to a ball of smaller radius $M\ggen_+^3$ (in our application $\alpha = 4M$).
    }
\label{fig:Kplusdom}
\end{figure}

\begin{theorem} \label{thm:step-mr-fv}
  Let $d =4$ and $\volume = \Lambda_N$.
  Let $C_\DV$ and $L$ be sufficiently large, and let $p,q\in \N_0$.
  There exist $M>0$ (depending on $p,q$), $\delta >0$,  and $\kappa = O(L^{-1})$ such that
  for $\ggen \in (0,\delta)$ and $\mgen^2 \in \Iint_+$, and with the domain
  $\domRG$ defined using any $\DVa> M$, the maps
  \begin{equation}
  \lbeq{RKplusmaps}
    R_+:\domRG(\sgen,\Lambda) \times \Igen_+(\mgen^2) \to \Vcal,
    \quad
    K_+:\domRG(\sgen,\Lambda) \times \Igen_+(\mgen^2) \to \Wcal_{+}(\sgen_+,\Lambda)
  \end{equation}
  are analytic in $(V,K)$,
  and satisfy the estimates
  \begin{align}
\label{e:Rmain-g}
    \|D_V^p D_K^q R_+\|_{L^{p,q}}
    & \le
    \begin{cases}
    M
    \chigen\ggen_{+}^{3-p} & (p\ge 0,\, q=0)\\
    M
     \ggen_{+}^{1-p-q} & (p\ge 0,\, q = 1,2)\\
    \rlap{$0$}\hspace{3.3cm}  & (p\ge 0,\, q \ge  3),
    \end{cases}
\\
\lbeq{DVKbd}
    \|D_{V}^pD_{K}^{q}K_+\|_{L^{p,q}}
    &\le
    \begin{cases}
    M  \chigen \ggen_{+}^{3-p}
    &
    (p \ge 0)
    \\
    \rlap{$\kappa$}\hspace{3.3cm}
    & (p=0,\, q=1)
    \\
    M  \ggen_{+}^{-p}
    (
    \chigen
    \ggen_{+}^{10/4}
    )^{1-q}
    &
    (p \ge 0,\, q \ge 1)
    .
    \end{cases}
  \end{align}
  In addition, $R_+,K_+$, and every Fr\'echet derivative in $(V,K)$,
  when applied as a multilinear map to directions $\dot{V}$ in $\Vcal^{p}$
  and $\dot{K}$ in $\Wcal^{q}$, is jointly continuous in all arguments,
  $m^{2}, V,K, \dot{V}, \dot{K}$.
\end{theorem}

We emphasise that when $\mgen^2=0$ in Theorem~\ref{thm:step-mr-fv}, the
continuity statement includes right-continuity at $m^2=0$.
The bounds for $(p,q)=(0,0)$ show that $R_+$ and $K_+$ are third-order in $\ggen$ and
thus third-order in $V$.
The $(p,q)=(0,1)$ bound of
\eqref{e:DVKbd} shows that the $K$-coordinate is contracting, since $\kappa < 1$.
Thus, viewed as a dynamical system, $(V,K)\mapsto (V_+,K_+)$ has two centre directions $g$ and $z$,
one expanding direction $\mu$, and an infinite-dimensional contracting direction $K$.
Theorem~\ref{thm:step-mr-fv} is the deepest ingredient in our analysis; its proof
occupies most of \cite{BS-rg-IE,BS-rg-step}.  The
development of the norms defined in \cite{BS-rg-norm}
culminates in the $\Wcal$ norm needed for Theorem~\ref{thm:step-mr-fv}, and the
operator $\LT$ defined in \cite{BS-rg-loc} is designed to achieve the contraction
in $K$.

\subsection{Renormalisation group map in infinite volume}
\label{sec:rg-iv}

In our context, the renormalisation group map is most naturally defined to be a map
in the setting of a torus, since a defining property is that it should preserve the circle
product under expectation as in \refeq{EIK}.  We have no analogue of \refeq{EIK}
for the infinite volume $\Zd$, for which we have defined neither the expectation nor
the circle product.  Nevertheless, there is a natural definition
of a map $(V,K) \mapsto (V_+,K_+)$ which is set on $\Zd$ rather than on a torus $\Lambda$,
as an appropriate inductive limit of the corresponding maps on the family of all tori.
The infinite volume map has the advantage that it is defined for all scales
$j \in \N$, with no restriction due to finite volume.  In particular, we can
study the limit $j \to \infty$, which we use in an important way to apply
the main result of \cite{BBS-rg-flow} to the
dynamical system defined by the renormalisation group.

The definition of the renormalisation group map on $\Zd$ is discussed
in detail in \cite[Section~\ref{step-sec:rg-iv}]{BS-rg-step}.  There a
correspondence is established between elements of the family of spaces
$\Kcal_j(\Lambda)$ indexed by $\Lambda$ which obey a compatibility
property called Property~$(\Zd)$, and the space $\Kcal_j(\Zd)$.  In
brief, given $X \in \Pcal_j(\Zd)$ and any $\Lambda$ whose period is at
least twice the diameter of $X$, there is a map $\iota: X \to \Lambda$
such that nearest neighbours in $X$ are mapped to nearest neighbours
in $\Lambda$ and the preimages of nearest neighbours in $\Lambda$ are
nearest neighbours in $X$.  The above-mentioned correspondence permits
identification of $K_{j,\Lambda}(\iota X)$ with $K_{j,\Zd}(X)$ for all
$\Lambda$ whose period is at least twice the diameter of $X$.  The
diameter restriction causes information to be lost in going from the
finite volume $K$ to its infinite volume version, as there is no
infinite volume counterpart to $K(U)\in \Kcal(\Lambda)$ whenever $U$
is comparable in size to $\Lambda$, and in particular if $U$
``wraps around'' the torus $\Lambda$.  This does not cause
difficulties since our analysis shows that the scale-$j$ norm of
$K(U)$ decays exponentially in the number of scale-$j$ blocks in $U$,
and the significant information is maintained by the correspondence.

It is shown in \cite[Proposition~\ref{step-prop:KplusZd}]{BS-rg-step}
that
\begin{equation}
\lbeq{Kiv}
    \text{if the family $(K_\Lambda)$ has Property~$(\Zd)$ then so does the family
    $(K_{+,\Lambda})$},
\end{equation}
and this permits the definition of a map $(V,K_\Zd) \mapsto K_{+,\Zd}$.
Moreover, it is shown in
\cite[Proposition~\ref{step-prop:VZd}]{BS-rg-step} that the map $V_+$ can be regarded
as a map on $(V,K_\Zd)\in \domRG(\Zd)$, and in addition
\begin{equation}
\lbeq{Viv}
    \text{if the family $(K_\Lambda)$ has Property~$(\Zd)$ then
    $V_{+,\Lambda}(V,K_\Lambda)=V_{+,\Zd}(V,K_{\Zd})$ for all $\Lambda$}.
\end{equation}
In particular, this shows that the loss of information in going from $(K_\Lambda)$
to $K_\Zd$ does not affect the flow of coupling constants.

The following theorem is proved in \cite[Theorem~\ref{step-thm:VKZd}]{BS-rg-step}.
In particular, it shows that $R_{+,\Zd}$ and $K_{+,\Zd}$
(in the space $\Wcal(0,\ggen)$)
are right-continuous
at $m^2=0$ for all $j<\infty$.

\begin{theorem}
\label{thm:step-mr-iv}
All statements of Theorem~\ref{thm:step-mr-fv}
hold, including domains and estimates,
with the same parameters and constants, when $\volume = \Lambda$ is replaced by
$\volume = \Zd$.
\end{theorem}

\subsection{Renormalisation group map in transformed variables}
\label{sec:trans}

In Section~\ref{sec:flow-approx}, we introduced the transformation $T_j$,
and discussed its effect on the perturbative quadratic part of the renormalisation group flow.
We now show that its effect on the nonperturbative part
is insignificant.
The transformed renormalisation group flow is important in proofs in Sections~\ref{sec:bulk-flow}--\ref{sec:pfmr}.

We distinguish versions of the renormalisation group coordinates and maps in
transformed variables by writing, e.g., $\Vch = T(V)$, and $(\Vch_+,\Kch_+)$ for versions
of the maps $(V_+,K_+)$ which act on transformed variables $(\Vch,K)$.
Since the transformation acts trivially for $j=0$, i.e.,  $T_0(V)=V$,
we do not distinguish between $V_0$ and $\Vch_0$ and use both interchangeably.
In more detail, the maps $\Vch_+,\Kch_+$ are defined by
\begin{equation} \label{e:VchKch}
  \Vch_+(\Vch,K) = T_{+}(V_+(T^{-1}(\Vch),K)),
  \qquad
  \Kch_+(\Vch,K) = K_+(T^{-1}(\Vch),K),
\end{equation}
and we also set
\begin{equation} \label{e:Rch}
  \Rch_+(\Vch,K) = \Vch_+(\Vch,K)  - \bar\varphi(\Vch).
\end{equation}
As in Section~\ref{sec:rgiv},
dependence on $m^2$ is left implicit in the notation, and the subscript $+$
is shorthand for scale $j+1$.

We extend $T,T^{-1}$ to act trivially on $K$, i.e., $T(V,K) = (T(V),K)$, $T^{-1}(V,K) = (T^{-1}(V),K)$,
and set $\domRGch = T(\domRG)$. Since $T$ is invertible for $V$ sufficiently small (which we assume),
also $T^{-1}(\domRGch) = \domRG$.

\begin{cor} \label{cor:step-mr-tr}
  All statements of Theorems~\ref{thm:step-mr-fv}--\ref{thm:step-mr-iv}
  (namely regularity, domains, estimates)
  hold when $R_+,K_+$ are replaced by $\Rch_+,\Kch_+$.
\end{cor}

\begin{proof}
  We first note that since $T$ is a polynomial with $T(V) = V+O(\|V\|^2)$, for bounded $p$,
  \begin{equation} \label{e:DTbd}
    \|D^pT(V)\| \leq O(1), \quad \|D^pT^{-1}(V)\| \leq O(1).
  \end{equation}
  The chain rule and \eqref{e:DTbd} imply
  \begin{equation} \label{e:Kchbdpf}
    \|D^p_{\Vch}D^q_K \Kch_+(\Vch,K)\|
    \leq C \sum_{k=1}^p \|D^k_V D^q_K K_+(T^{-1}(\Vch),K)\|,
  \end{equation}
  and the claim for $\Kch_+$ immediately follows from \eqref{e:DVKbd}
  since the bound on the deriviative of $K_+$ is largest for $k=p$.

  To bound the derivatives of $\Rch_+$, we write
  $\Rch_{+} = \Rch_\pt + \Rch_*$, where
  \begin{align}
    \Rch_\pt(\Vch) &= T_{+}(\varphi_\pt^{(0)}(T^{-1}(\Vch))) - \bar\varphi(\Vch)
    \\
    \Rch_*(\Vch,K) &=
    T_{+}\left(\varphi_\pt^{(0)}(T^{-1}(\Vch))
    +
    R_+(T^{-1}(\Vch),K)\right)
    -T_{+}(\varphi_\pt^{(0)}(T^{-1}(\Vch)))  \big).
  \end{align}
  It suffices to show that $\Rch_\pt,\Rch_*$ both satisfy the bounds claimed for $R_+$ in \eqref{e:Rmain-g}.
  As discussed below \eqref{e:TVV2},
  $T_{+} \circ \varphi_\pt^{(0)} \circ T^{-1} = \bar\varphi + O(\|V\|^3)$.
  Since $\varphi_\pt^{(0)}, \bar\varphi, T,T_{+}$ are polynomials,
  it follows that
  \begin{equation}
    D_{\Vch}^p (T_{+} \circ \varphi_\pt^{(0)} \circ T^{-1})
    = D_{\Vch}^p \bar\varphi + O(\|\Vch\|^{3-p}).
  \end{equation}
  In particular, since $\Rch_\pt$ is independent of $K$, this shows that $\Rch_\pt$ satisfies \eqref{e:Rmain-g}.
  For $\Rch_*$,
  we write
  the quadratic
  polynomial $T_{+}$ as $T_{+}(V) = V+B_{+}(V,V)$ with $B_+(V_1,V_2)$ bilinear,
  and we set $f=\varphi_\pt^{(0)}(T^{-1}(\Vch))$ and
  $r=R_+(T^{-1}(\Vch),K)$. Then
  \begin{equation}
    \Rch_* = r
    + B_{+}(r,f) + B_+(f,r) + B_{+}(r,r).
  \end{equation}
  The first term is bounded exactly as \eqref{e:Kchbdpf} and the
  others can be seen to be smaller in a similar way.
  This completes the proof.
\end{proof}

\section{Renormalisation group flow}
\label{sec:bulk-flow}

We now discuss the global aspect of the renormalisation group: the
enabling of indefinitely repeated application of the renormalisation
group map, as the scale $j$ increases.  It is this global aspect that
requires the careful choice of $(\nu_0^c,z_0^c)$ in
Theorem~\ref{thm:suscept-diff}, and that leads to the identification
of the critical point.
We restrict throughout Section~\ref{sec:bulk-flow} to the infinite volume
case, $\volume=\Z^d$.  In Section~\ref{sec:pfmr}, we return to the case
$\volume = \Lambda_N$.

\subsection{Existence and regularity of global flow}
\label{sec:flow-global}

We say that $(V_j,K_j)_{j\in\N_0}$ is a \emph{global flow} of the
infinite volume renormalisation group if
\begin{equation} \label{e:flowVK}
  (V_{j+1},K_{j+1}) = (V_{j+1}(V_j,K_j),K_{j+1}(V_j,K_j)) \quad \text{for all $j \in \N_0$}.
\end{equation}
In \refeq{flowVK},
on the left-hand side
$(V_{j+1},K_{j+1})$ is
an element in the sequence $(V_j,K_j)$, and on the right-hand side
it denotes the map $(V_+,K_+)=(V_{j+1},K_{j+1})$
of Section~\ref{sec:rg-map}; the interpretation should be clear from context.
We suppress the dependence on the mass parameter $m^2$ in our notation.

The following proposition, which sits at the centre of
our analysis, constructs a sequence with the desired properties.
For its statement, we fix the parameter $\DVa$ in the definition of
the domains $\domRG$ in \eqref{e:domRG} as $\DVa = 4M$, where
$M$ is the constant in Theorem~\ref{thm:step-mr-fv} (this choice is
convenient, but somewhat arbitrary).

\begin{prop} \label{prop:flow-flow}
  Let $d=4$ and $\volume = \Zd$, and
  let $\delta>0$ be sufficiently small.
  There are continuous functions of $(m^2,g_0)$, namely $z_0^c, \mu_0^c: [0,\delta)^2
  \to \R$, with $z_0^c(m^2,0)=\mu_0^c(m^2,0)=0$,
  continuously differentiable in $g_0 \in (0,\delta)$ with uniformly bounded derivatives,
  such that for all $(m^2,g_0)\in [0,\delta) \times (0,\delta)$, the global flow
  \eqref{e:flowVK} with mass parameter $m^2$ and initial condition
  given by
  \begin{equation} \label{e:flow-flow-ic}
    V_0 = (g_0, z_0^c(m^2,g_0), \mu_0^c(m^2,g_0)), \quad K_0=\1_{\varnothing},
  \end{equation}
  exists, and $(V_j,K_j) \in \domRG_j(s_j,\Z^d)$ for all $j\in\N_0$.
  Here $s_j=(m^2,\ggen_j(m^2,g_0))$ with $\ggen_j$ given by \refeq{ggendef}.
  In particular, $\|K_j\|_{\Wcal_j(s_j)} = O(\chi_j \gbar_j^3)$ and $\gch_j,g_j = O(\gbar_j)$.
    In addition, $\zch_j,\much_j,z_j,\mu_j = O(\chi_j\gbar_j)$.
\end{prop}

Proposition~\ref{prop:flow-flow} can be extended to
replace its initial condition $K_0 = \1_{\varnothing}$
by a more general condition on $K_0$.  The extension is not required for our present
purposes, so is not pursued here, but a more general initial condition
is needed for the analysis of weakly self-avoiding walk with nearest-neighbour
attraction \cite{BBS-saw-sa}.

The proof of Proposition~\ref{prop:flow-flow} is based an application of the main result
of \cite{BBS-rg-flow}, which concerns a class of non-hyperbolic dynamical systems.
For this, it is advantageous to work with the transformed variables $\Vch_j=T_j(V_j)$
defined in \refeq{TVV2}--\refeq{mubar}, as this produces a triangular system to
second order.  We therefore reformulate the renormalisation group map
$(V,K) \mapsto (V_{+,\Zd},K_{+,\Zd})$ of Section~\ref{sec:rg-iv} in terms of the
transformed variables.
This was anticipated in Section~\ref{sec:trans}, where maps $\Vch_+,\Kch_+,\Rch_+$
are defined.
Thus we define the \emph{evolution map}
\begin{align} \label{e:Phi}
  \Phi_j(\Vch_j,K_j)
  &= (\Vch_{j+1}(\Vch_j,K_j),\Kch_{j+1}(\Vch_j,K_j))
  \nnb &
  = (\bar\varphi_j(\Vch_j) + \Rch_{j+1}(\Vch_j,K_j), \Kch_{j+1}(\Vch_j,K_j)).
\end{align}
The map $\bar\varphi$ is given by \eqref{e:gbar}--\eqref{e:mubar},
and its global flow with initial condition $\gbar_0$ and final condition
$(\zbar_\infty,\mubar_\infty)=(0,0)$ is determined by Proposition~\ref{prop:approximate-flow}.
Precise statements about the domains of the $\Phi_j$ are deferred to Section~\ref{sec:pfflow}.

We interpret $\Phi_j$
as the time-dependent evolution map of a dynamical system, which is an in\-fi\-nite-dimensional
perturbation of the
3-dimensional dynamical system
$\bar\varphi$ on $\Vcal = \R^3$.
The maps $(\Phi_j)$ are between
different spaces, since $K_j$ and $K_{j+1}$ are functions of polymers
defined in terms of blocks
of different side lengths $L^j$ and $L^{j+1}$, respectively, but
this aspect is unimportant.
Although we have not defined $(V_+,K_+)$
at $(V,K)=(0,0)$ in Section~\ref{sec:rgiv},
it is natural to extend it to act trivially on
$(0,0)$, which can thus be regarded a \emph{fixed point} of the dynamical system $\Phi$.
We are interested in the behaviour of $\Phi$ near this fixed point. 
In particular,
we wish to construct a sequence $(\Vch_j,K_j)$ that satisfies the flow equation
\begin{equation} \label{e:Phi-flow}
  (\Vch_{j+1},K_{j+1}) = \Phi_j(\Vch_j,K_j),
\end{equation}
with good estimates on the sequence $(\Vch_j,K_j)$ and with boundary conditions
\begin{equation} \label{e:Phi-bc}
  K_0=\1_\varnothing, \quad \gch_0>0 \; \text{small}, \quad \zch_\infty = \much_\infty = 0.
\end{equation}
The condition $\gch_0>0$ is related to the stability problem
discussed around \eqref{e:phi41d}; because of it the fixed point
can be approached from one side only.

\subsection{Non-hyperbolic dynamical system}

\renewcommand{\I}{{\sf N}}
\renewcommand{\b}{b}
\renewcommand{\u}{}
\newcommand{\domK}{a}
\newcommand{\domV}{h}

The dynamical system \refeq{Phi-flow} is
not hyperbolic near the fixed point $(0,0)$, due to the two unit eigenvalues of $\bar\varphi$
corresponding to the variables $g,z$.
To study it, we apply the main result of \cite{BBS-rg-flow}, which considers
a general class of dynamical systems $\Phi = (\Phi_j)$ with non-hyperbolic fixed point and with
contractive coordinate $K_j$ lying in
a given sequence of Banach spaces $\Wcal_j$.
The result of \cite{BBS-rg-flow} shows that, under appropriate assumptions on $\Phi_j$,
there exists a global flow satisfying the boundary conditions \eqref{e:Phi-bc}.
We now recall this result in the form we require.
Because we apply the result to the renormalisation group flow in the transformed variables
discussed in Section~\ref{sec:flow-approx},
we state it in the notation of
transformed variables $(\gch,\zch,\much,K)$ rather than
the variables $(g,z,\mu,K)$ used in \cite{BBS-rg-flow}.

The dynamical system studied in \cite{BBS-rg-flow} involves a quadratic
flow map $\bar\varphi$, of which the flow defined by \refeq{gbar}--\refeq{mubar} is
an instance.
This flow map in \cite{BBS-rg-flow} is required to obey
the Assumptions~(A1--A2)  stated in  Section~\ref{sec:flow-approx},
with a fixed $\Omega >1$ which defines $\jm$ as in \refeq{mass-scale}.
Any such quadratic flow has a unique global flow which obeys the
boundary conditions of our interest, and this global flow obeys the estimates of
Proposition~\ref{prop:approximate-flow}; this is proved in
\cite[Proposition~\ref{flow-prop:approximate-flow}]{BBS-rg-flow}.
We have already seen that the specific quadratic flow \refeq{gbar}--\refeq{mubar}
obeys Assumptions~(A1--A2) and Proposition~\ref{prop:approximate-flow}.
In this section, we consider any quadratic flow that obeys
Assumptions~(A1--A2), not necessarily our specific example \refeq{gbar}--\refeq{mubar},
and we replace $\Phi$ of \refeq{Phi} by the more general map
\begin{equation}
  \Phi_j(\Vch_j,K_j)
  = (\bar\varphi_j(\Vch_j) + \rho_j(\Vch_j,K_j), \psi_j(\Vch_j,K_j)),
\end{equation}
where $(\rho_j,\psi_j)$ satisfy Assumption~(A3) stated below.

To formulate the theorem,
let  $\Vcal =\R^3$, let $(\Wcal_j)_{j\in\N_0}$ be a sequence of Banach spaces, and set $X_j = \Vcal \oplus \Wcal_j$
(compared with \cite{BBS-rg-flow} we have reversed the order of the components $\Vcal,\Wcal_j$).
We define domains $D_j \subset X_j$ on which $(\rho_j,\psi_j)$ is assumed to be defined,
and an assumption which states estimates for $(\rho_j,\psi_j)$, as follows.
The domain and estimates depend on the initial
condition $g_0$ and an external parameter $m^2$ (here we write $m^2$ instead of $m$
as in \cite{BBS-rg-flow}).
For parameters $\domK,\domV>0$ and sufficiently small $g_0>0$, let $(\gbar_j,\zbar_j,\mubar_j)_{j\in\N_0}$
be the sequence determined by Proposition~\ref{prop:approximate-flow} with initial condition
$\gbar_0=g_0$ and mass $m^2$,
and define the domain $D_j = D_j(m^2, g_0, \domK,\domV) \subset X_j$  by
\begin{align}
  \lbeq{Djdef}
  D_j
  =
  \{
  x_j
  =
  (\gch_j,\zch_j,\much_j,K_j)
  \in X_j :
  |\gch_j - \gbar_j| &\leq \domV \gbar_j^2 |\log \gbar_j|,\nnb
  |\zch_j - \zbar_j| &\leq \domV \chi_j\gbar_j^2 |\log \gbar_j|,\nnb
  |\much_j - \mubar_j| &\leq \domV \chi_j\gbar_j^2 |\log \gbar_j|,\nnb
  \|K_{j}\|_{\Wcal_j} &\leq \domK \chi_j\gbar_j^3
  \}.
\end{align}
Assumptions~(A1--A2) for the quadratic flow are supplemented with the following
assumption for the perturbation $(\rho_j,\psi_j)$.
Constants in Assumption~(A3) carry subscripts ``3'' to indicate
that they arise in (A3).

\medskip\noindent \textbf{Assumption~(A3).}
  \emph{The perturbation:}
  The maps
  $\rho_j : D_j \to \Vcal \subset X_{j+1}$ and
  $\psi_j : D_j \to \Wcal_{j+1} \subset X_{j+1}$
  are three times
  continuously
  Fr\'echet differentiable,
  and
  there exist $\kappa \in (0,\Omega^{-1})$,
  $R \in (0,\domK(1-\kappa\Omega))$,
  $M>0$
  such that,
  for all $x_j = (V_j,K_j) \in D_j$,
  \begin{align}
\label{e:Rmain-g-flow}
    \|D_V^p D_K^q  \rho_j(x_j)\|_{L^{p,q}}
    & \le
    \begin{cases}
    M
    \chi_{j+1}\gbar_{j+1}^{3-p} & (p=0,1; \, q=0)\\
    M (\gbar_{j+1}^2 |\log \gbar_{j+1}|)^{-p} (\chi_{j+1}\gbar_{j+1}^3)^{1-q}
    & \Big(\begin{array}{ll} &p=0,\,q=1 \\ &\text{ or } p+q=2, 3\end{array}\Big),
    \end{cases}
\\
\lbeq{DVKbd-flow}
    \|D_{V}^pD_{K}^{q}\psi_j(x_j)\|_{L^{p,q}}
    &\le
    \begin{cases}
    R  \chi_{j+1} \gbar_{j+1}^{3}
    &
    (p =  q=0; \, K_j=0)
    \\
    M  \chi_{j+1} \gbar_{j+1}^{2}
    &
    (p = 1; \, q=0)
    \\
    \rlap{$\kappa$}\hspace{3.3cm}
    & (p=0,\, q=1)
    \\
    M (\gbar_{j+1}^2 |\log \gbar_{j+1}|)^{-p} (\chi_{j+1}\gbar_{j+1}^3)^{1-q}
    & (p+q=2, 3)
    .
    \end{cases}
  \end{align}

The following theorem, which is illustrated by
Figure~\ref{fig:phaseport}, is a restatement of
\cite[Theorem~\ref{flow-thm:flow}]{BBS-rg-flow} (with
\cite[Remark~\ref{flow-rk:Nrad}]{BBS-rg-flow} for part~(iii)), and its
corollary is a restatement of
\cite[Corollary~\ref{flow-cor:masscont}]{BBS-rg-flow}.  We make the
simplifying assumption that $K_0=0$, which corresponds to our
application with $K_0(X)=0$ for all non-empty polymers $X$.  In the
theorem, the sequence $(\bar K_j)_{j\in\N_0}$ is defined inductively
by $\bar K_{j+1} = \psi_j(\bar K_j, \bar V_j)$, with $\bar K_0=0$.
The existence of this sequence is established in
\cite[Lemma~\ref{flow-lem:small-ball}]{BBS-rg-flow}, which also gives
\begin{equation}
  \lbeq{Kbarbd}
  \|\bar K_j\|_{\Wcal_j} \leq \domK_* \chi_j \gbar_j^3
\end{equation}
for any $\domK_* \in (R/(1-\kappa\Omega),\domK]$,
if $\gbar_0$ is chosen sufficiently small (depending on $a_*$).

\begin{figure}
  \begin{center}
    \input{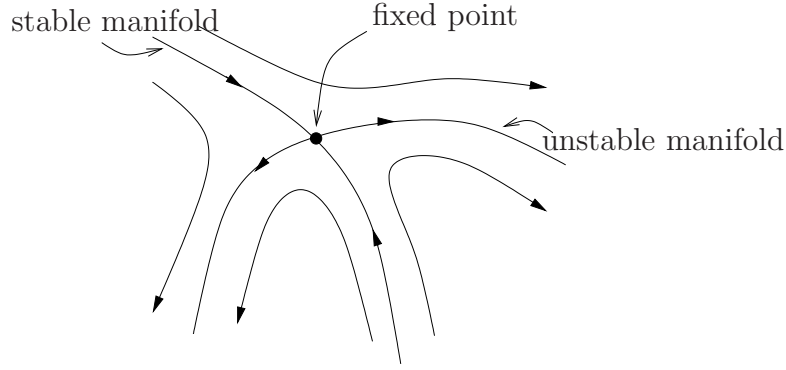}
  \end{center}
  \caption{Schematic phase portrait of the global flow of Theorem~\ref{thm:flow4}.
  In the renormalisation group
  flow  of  Proposition~\ref{prop:flow-flow},
    the portion of the stable manifold near the fixed point $(V=(0,0,0),K=0)$ restricted to $K=0$ consists of the
    points $V = (g_0,\mu_0^c(g_0), z_0^c(g_0))$, $g_0 \in [0,\delta]$.
    }
\label{fig:phaseport}
\end{figure}

\begin{theorem}
\label{thm:flow4}
  Fix a sequence of Banach spaces $\Wcal_j$.
  Suppose that (A1--A3) hold
  with parameters given by
  $(\domK,\domV,\kappa,\Omega,M,R)$ and with $\gbar_0=\ggen_0$.
  Let $\domK_* \in (R/(1-\kappa\Omega),\domK)$, $\b \in (0,1)$.
  There exists $\domV_* > 0$ such that for any $\domV \geq \domV_*$,
  there exists $g_*>0$ such that if $\tilde g_0 \in (0,g_*]$,
  there exists a neighbourhood  $\I = \I(\tilde g_0) \subset \R$ of $\ggen_0$
  such that the following conclusions hold.

  \smallskip\noindent
  (i)
  For initial condition $g_0 \in \I$,
  there exists a global flow $\xch$ of $\Phi=(\bar\varphi+\rho,\psi)$ with $(\zch_\infty,\much_\infty)=(0,0)$
  such that, with $\bar x$ the unique flow of $\bar\Phi=(\bar\varphi,\psi)$ determined by the same
  boundary conditions,
  \begin{align}
     \lbeq{VVbar2app}
     |\gch_j - \gbar_j|
     &\leq \b \domV \gbar_j^2 |\log \gbar_j|,
     \\
     \lbeq{VVbar3app}
     |\zch_j - \zbar_j|
     & \leq \b \domV \chi_j \gbar_j^2 |\log \gbar_j|     ,
     \\
     \lbeq{VVbar4app}
     |\much_j -\mubar_j|
     &\leq \b \domV \chi_j \gbar_j^2 |\log \gbar_j|,
     \\
     \lbeq{VVbar1app}
     \|K_j -\bar K_j\|_{\Wcal_j}
     &\leq \b (\domK-\domK_*) \chi_j \gbar_j^3
    .
  \end{align}
  The sequence $\xch$
  is the unique solution to
  \eqref{e:Phi-flow} which obeys
  these boundary conditions
  and the bounds \refeq{VVbar2app}--\refeq{VVbar1app}.

  \smallskip\noindent
  (ii)
  For every $j \in \N_0$, the map
  $(\Vch_j,K_j): \I \to \Vcal \oplus \Wcal_j$ is $C^1$
  and obeys
  \begin{equation} \lbeq{derivbd}
    \ddp{\zch_0}{g_0} = O(1), \quad \ddp{\much_0}{g_0} = O(1).
  \end{equation}

  \smallskip\noindent
  (iii)
  The neighbourhood  $\I(\tilde g_0)$ can be taken to be an interval
  centred at $\tilde g_0$, whose length $\varepsilon(\tilde g_0) > 0$
  depends only on $\ggen_0$, the constants in (A1--A3), and $a_*,b,h$,
  and with $\varepsilon(\ggen_0)$  bounded below away from $0$ uniformly on compact
  subsets of $\ggen_0>0$.
\end{theorem}

In particular, \eqref{e:Kbarbd} and \eqref{e:VVbar1app} imply that
$\|K_j\|_{\Wcal_j} \leq a \chi_j\gbar_j^3$ for all $j \in \N_0$.
Theorem~\ref{thm:flow4} concerns a single dynamical system $\Phi = (\Phi_j)_{j\in\N_0}$.
Let $\Mext$ be a metric space of external parameters and assume now that the $\Phi_j$
depend continuously on an external parameter $m^2 \in \Mext$, in the sense that the
$\Phi_j$ are continuous maps $X_{j} \times \Mext \to X_{j+1}$.
We say that $\Phi$ satisfies (A1--A3) uniformly if
$\Phi_j(\cdot, m^2)$ satisfies (A1--A3) for parameters
$(\domK,\domV,\kappa,\Omega,R,M)$
independent of $m^2 \in \Mext$.
For the proof of Theorem~\ref{thm:suscept-diff}, we apply the following extension
to Theorem~\ref{thm:flow4},
with $\Mext \subset [0,\delta)$.

\begin{cor}
  \label{cor:flow4-masscont}
   Assume that the $\Phi_j$ depend continuously on an external parameter $m^2 \in \Mext$
   and that Assumptions~(A1--A3) hold uniformly in $m^2$.
   Let $x = x(m^2,g_0) = (V(m^2,g_0),K(m^2,g_0))$
   be the global flow for external parameter $m^2$ and initial condition
   $g_0 \in \I(\ggen_0)$ guaranteed by Theorem~\ref{thm:flow4}.
   Then $x_j$ 
   is continuous in $(m^2,g_0)$ for each $j \in \N_0$,
   i.e., $(m^2,g_0) \mapsto V_j(m^2,g_0)$ is a continuous function into $\R^3$,
   and
   $(m^2,g_0) \mapsto K_j(m^2,g_0)$ is a continuous function into $\Wcal_j$.
\end{cor}

\newcommand{\Wcalt}{\Wcal^{\smash{0}}}

\subsection{Proof of Proposition~\ref{prop:flow-flow}}
\label{sec:pfflow}

We prove Proposition~\ref{prop:flow-flow}
by applying Theorem~\ref{thm:flow4} to the infinite volume evolution map
\eqref{e:Phi} for the transformed variables, with the verification of
(A3) via Theorem~\ref{thm:step-mr-iv}.
The proof is divided into three main steps:

\smallskip\noindent
\emph{Step 1.} Construction of the maps $z_0^c,\mu_0^c$ and proof of the estimates for $(V_j,K_j)$ for all $j\in\N$.
\\
\emph{Step 2.} Proof of continuity of $z_0^c,\mu_0^c$ in the interior of their domains, i.e., $(m^2,g_0) \in (0,\delta)^2$.
\\
\emph{Step 3.} Proof that $z_0^c,\mu_0^c$ are continuous on $[0,\delta)^2$.
\smallskip

For the regularity statements of steps~2--3, a difficulty is that
Theorem~\ref{thm:step-mr-iv} involves the restrictions
$\mgen^2 \in \Iint_+ = [0,\delta)$
and $m^2 \in \Igen_+(\mgen^2)$, which maintain the mass $\mgen^2$ appearing in the norm
close to the mass $m^2$ appearing in \refeq{Phi} through its implicit
dependence on the original covariance $(-\Delta+m^2)^{-1}$.
To prove
mass continuity of $z_0^c,\mu_0^c$ in $m^2 \in [0,\delta)$, we must prove some
compatibility in this respect.
Similarly, Theorem~\ref{thm:flow4} permits only local variation of $g_0$,
whereas we must prove continuity in $g_0\in [0,\delta)$.
This is illustrated by Figure~\ref{fig:covering}.
To deal with these difficulties, although we aim to prove
$(V_j,K_j) \in \domRG_j(s_j)$
with
$s_j=(m^2,\ggen_j(m^2,g_0))$, we consider a wider
class of  parameters and spaces.
To simplify the notation, we write
\begin{equation} \label{e:Wcaltdef}
  \Wcalt_j(\mgen^2,\ggen_0) = \Wcal_j(\mgen^2, \ggen_j(\mgen^2,\ggen_0),\Z^d),
\end{equation}
where here $(\mgen^2, \ggen_0)$ defines the sequence $(\ggen_j)$ via \eqref{e:ggendef}.

\begin{figure}
  \begin{center}
    \input{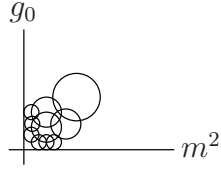}
  \end{center}
  \caption{
      To construct a continuous solution on $(m^2,g_0) \in (0,\delta)^2$, we first construct solutions
      in small neighbourhoods of a fixed $(\mgen^2, \ggen_0)$, and show compatibility of these solutions.
    }
\label{fig:covering}
\end{figure}

\subsubsection{Preliminaries: the sequence \texorpdfstring{$\ggen$}{gtilde} and norms}

The proof uses the properties
established in the next lemma, for the sequence $\ggen_j$ defined in \refeq{ggendef}.
The proof of Lemma~\ref{lem:gbarmcomp} uses elementary calculus only.

\begin{lemma} \label{lem:gbarmcomp}
  (i) There exists $\delta>0$ such that uniformly in $(m^2,g_0) \in (0,\delta)^2$
  and in $j \in\N_0$,
  \begin{equation} \label{e:ggengbarequiv2}
    \ggen_j(m^2,g_0) = \gbar_j(m^2,g_0) + O(\gbar_j^2(m^2,g_0)).
  \end{equation}
  (ii) For sufficiently small $\theta > 0$,
  if $|m^2-\hat m^2| \leq \theta\hat m^2$ 
  and $|g_0 -\hat g_0| \leq \theta \hat g_0$, 
  then
  \begin{equation} \label{e:gbarmg0bd}
    |\ggen_j(m^2,g_0) - \ggen_j(\hat m^2, \hat g_0)|
    \leq
    (\theta+O(g_0)) \ggen_j(\hat m^2, \hat g_0).
  \end{equation}
  (iii) For $j \in\N_0$,
  $\ggen_j(m^2,g_0)$ is monotone increasing in each of $m^2$ and $g_0$.
\end{lemma}

\begin{proof}
  (i)
  We use dots to denote derivatives with respect to $m^2$.  By
  \cite[Lemma~\ref{pt-lem:betalim}(b)]{BBS-rg-pt},
  $\dot\beta_l = O(L^{2l})$ (with $L$-dependent constant).
  Differentiation of \eqref{e:gbar} gives
  $\dot\gbar_{j+1} = \dot\gbar_j (1-2\beta_j \gbar_j) - \dot\beta_j \gbar_j^2$.
  We solve the recursion and then apply
  \cite[Lemma~\ref{flow-lem:elementary-recursion}(iii)]{BBS-rg-flow} to obtain
  \begin{equation} \label{e:gbarj-mderiv}
    \dot\gbar_j
    = - \sum_{l=0}^{j-1} \prod_{k=l}^j (1-2\beta_j\gbar_j) \dot\beta_l \gbar_l^2
    = - (1+O(g_0)) \sum_{l=0}^{j-1} \left(\frac{\gbar_j}{\gbar_l}\right)^2 \dot\beta_l \gbar_l^2
    = O(\gbar_j^2 L^{2j})
    .
  \end{equation}
  By continuity, $\gbar_j$ achieves its maximum on $[0,\hat m^2]$ at some
  $m_*^2$ (depending on $j$).
  For $j \leq j_m$, by integrating the derivative we obtain
  \begin{equation} \label{e:gbar0m1}
    |\gbar_j(0)-\gbar_j(m^2)|
    \leq
    O\left(  \gbar_j^2(m_*^2)\right)
    .
  \end{equation}
  Thus $\gbar_j(m^2_*) \leq \gbar_j(0) + O(\gbar_j(m^2_*)^2)$.
  Since $\gbar_j(m^2_*) = O(\gbar_0)$
  (by \cite[Lemma~\ref{flow-lem:elementary-recursion}(i)]{BBS-rg-flow}),
  it follows that
  $\gbar_j(m^2_*) \leq (1+O(\gbar_0)) \gbar_j(0)$, and hence
  $\gbar_j(m^2) = \gbar_j(0) + O(\gbar_j(0)^2)$ for $j \leq j_m$ as claimed.

  For $j \geq j_m$, we iterate \eqref{e:gbar}, and use $\gbar_l \leq \gbar_{j_m}$ for $l \geq j_m$
  (which is immediate from $\beta_k \geq 0$)
  and $\sum_{l=j_m}^\infty |\beta_l| = O(1)$
  (by \cite[Lemma~\ref{pt-lem:wlims}]{BBS-rg-pt}), to obtain
  \begin{equation}
    \gbar_j
    = \gbar_{j_m} \prod_{l=j_m}^{j-1} (1-\beta_l \gbar_l)
    = \gbar_{j_m} \exp\left(\sum_{l=j_m}^{j-1} O(\beta_l \gbar_l) \right)
    = \gbar_{j_m} \exp(O(\gbar_{j_m}))
    = \gbar_{j_m} + O(\gbar_{j_m}^2),
  \end{equation}
  and the proof of (i) is complete.

  \smallskip
  \noindent
  (ii)
  Using $\log(1+t) \leq t$ for $t \geq 0$, we first observe that for sufficiently small
  $\theta >0$,
  \begin{equation}
    |j_{\hat m} - j_m| \leq 1 + |\log_{L^2} \hat m^2 - \log_{L^2} m^2|
    \leq 1+ \frac{|\hat m^2 - m^2|}{m^2 \log L^2} \leq 1+O(\theta) < 2.
  \end{equation}
  Thus, since $j_m$ is an integer, $|j_{\hat m} - j_m| \leq 1$.
  Let $j_m^*(j) = j_m$ if $j \geq j_m$ and $j_m^*(j) = j$ otherwise. Then also
  $|j_m^* - j_{\hat m}^*| \leq |j_m - j_{\hat m}| \leq 1$.
  By \eqref{e:ggendef} and the triangle inequality,
  \begin{equation}
  \begin{aligned}
    \ggen_j(g_0, m^2) - \ggen_j(\hat g_0, \hat m^2)
    &=
    \gbar_{j_m^*}(g_0) - \gbar_{j_{\hat m}^*}(\hat g_0)
    \\&
    \leq
    |\gbar_{j_m^*}(g_0) - \gbar_{j_{\hat m}^*}(g_0)|
    +|\gbar_{j_{\hat m}^*}(g_0) - \gbar_{j_{\hat m}^*}(\hat g_0)|
    .
  \end{aligned}
  \end{equation}
  By \eqref{e:gbar},
  $|\gbar_{j_m^*} - \gbar_{j_{\hat m}^*}| \le O(\gbar_{j_{\hat m}^*}^2)$.
  Also, it is shown in \cite[Lemma~\ref{flow-lem:elementary-recursion}(iv)]{BBS-rg-flow}
  that if $|g_0-\hat g_0| \leq \theta \hat g_0$ then
   \begin{equation} \label{e:gbarg0g0tildebd}
     |\gbar_j(g_0) - \gbar_j(\hat g_0)|  \leq \theta(1+O(g_0))\gbar_j(g_0).
   \end{equation}
   The combination of the above estimates gives
  \begin{equation}
    |\ggen_j(g_0, m^2) - \ggen_j(\hat g_0, \hat m^2)|
    \leq
    (\theta+O(g_0)) \ggen_j(\hat g_0,\hat m^2)
  \end{equation}
  which is \eqref{e:gbarmg0bd}.

  \smallskip
  \noindent
  (iii)
  The sequence $\gbar$ is monotone increasing in $g_0$
  (by \cite[\eqref{flow-e:gbarprime}]{BBS-rg-flow}),
  and thus the sequence $\ggen$ is also monotone increasing in $g_0$.
  Also, the sequence $\gbar$ is decreasing in $j$ since $\beta_k \ge 0$ for all $k$,
  while $j_m$ is decreasing in $m^2$,
  so $\ggen$ is increasing in $m^2$.
\end{proof}

In preparation for step~2, we make the following observation.
It is proved in \cite[Lemma~\ref{step-wishlist-lem:Wcalnormequiv}]{BS-rg-step}
(with $1+O(\ggen-\ggen') \leq 2$)
that there exists $\delta>0$ such that,
whenever $\mgen^2 \leq \mgen'^2$ and $\ggen' \leq \ggen \leq (1+\delta)\ggen'$,
\begin{equation} \label{e:normequiv}
  \|\cdot\|_{\Wcal_j(\mgen^2,\ggen)}
  \leq
  2 \|\cdot\|_{\Wcal_j(\mgen'^2,\ggen')}
  .
\end{equation}
Therefore, by Lemma~\ref{lem:gbarmcomp}(ii-iii),
there exists $\delta>0$ such that
if $\mgen^2 \leq m^2$ and $g_0 \leq \ggen_0 \leq (1+\delta)g_0$,
then, for all $g_0 \in (0,\delta)$,
\begin{equation} \label{e:normequivg0m}
  \|\cdot\|_{\Wcalt_j(\mgen^2, \ggen_0)} \leq 2 \|\cdot\|_{\Wcalt_j(m^2,g_0)},
\end{equation}
with $\Wcal^0$ defined in \eqref{e:Wcaltdef}.

In preparation for step~3, we recall
from \cite[Lemma~\ref{step-wishlist-lem:Wcalnormequiv}]{BS-rg-step}
that
$\|\cdot\|_{\Wcal_j(\mgen,\ggen)} = \|\cdot\|_{\Wcal_j(0,\ggen)}$ if $j \le \jm$.
In addition,
if $j \le j_m$ then $\ggen_j(m^2,g_0) = \ggen_j(0,g_0)$ by definition
in \eqref{e:ggendef}.
Therefore the scale-$j$ norm is independent of $m^2$ as long as $j \le \min\{\jm,j_m\}$.
As discussed below \refeq{mass-scale}, the mass scale $j_m$ differs from $\jm$
by at most a constant, so we conclude that there is a constant $c$ such that
\begin{equation} \label{e:normindep}
  \|\cdot\|_{\Wcalt_j(m^2,g_0)} = \|\cdot\|_{\Wcalt_j(0,g_0)} \quad
  \text{for $m^2 \leq c L^{-2j}$.}
\end{equation}

\subsubsection{Step 1: construction of the maps \texorpdfstring{$z_0^c,\mu_0^c$}{z0c, mu0c}}

\begin{proof}[Step 1]
We construct the maps $z_0^c,\mu_0^c$ and show that for initial conditions \eqref{e:flow-flow-ic},
\eqref{e:flowVK} is valid and $(V_j,K_j) \in \domRG_j(s_j)$ for all $j\in\N_0$.
As explained in Section~\ref{sec:trans}, this is equivalent to showing that there exist
$(\Vch_j,K_j) \in \domRGch_j(s_j)$ satisfying \eqref{e:Phi-flow}.
We also show the estimates stated for $V_j,\Vch_j$.

We fix $(\mgen^2,\ggen_0) \in [0,\delta) \times (0,\delta)$.
To apply Theorem~\ref{thm:flow4} to the maps \eqref{e:Phi} with the $\Wcal_j$ spaces given
by $\Wcalt_j(\mgen^2,\ggen_0)$,
we must verify Assumptions~(A1--A3) with this choice.
Assumptions~(A1--A2) for the maps $\bar\varphi$
have already been seen to be satisfied in Section~\ref{sec:flow-approx}.
We apply Theorem~\ref{thm:step-mr-iv}
to verify (A3) with $\Wcal_j=\Wcalt_j(\mgen^2,\ggen_0)$.
Application of Theorem~\ref{thm:flow4} with a single value $(\mgen^2,\ggen_0)$
then produces a local solution, i.e.,
sequences $(\Vch_j,K_j) \in \domRGch_j$ defined for
each $(m^2,g_0)$ in a neighbourhood of $(\mgen^2,\ggen_0)$.
In steps~2--3, we subsequently show
that these local solutions can be combined
into a single continuous solution of $(m^2, g_0) \in [0,\delta)^2$.
The details are as follows.

Recall  the domain $\domRG_j$ of \refeq{domRG}.
We write $\domRG_j^{0}(\mgen^2, \ggen_0)$ instead of
$\domRG_j(\mgen^2, \ggen_j(\mgen^2,\ggen_0))$.
Recall also the domain  $D_j$ of \refeq{Djdef}.
Explicitly, with $\chigen_j=\chi_j(\mgen^2)$,
\begin{align}
\label{e:domRG-2}
  \domRG_j^{0}(\mgen^2,\ggen_0)
  &=
  \{ (g_j, z_j, \mu_j): C_\DV^{-1} \ggen_j < g_j < C_\DV \ggen_j, \;
  |z_j|, |\mu_j| < C_\DV \ggen_j \}
  \nnb
  & \qquad
  \times
  B_{\Wcalt_j(\mgen^2,\ggen_0)}(\DVa \chigen_j\ggen_j^3)
  ,
  \\
  \lbeq{Djdef-2}
  D_j(\mgen^2,\ggen_0)
  &=
  \{
  (\gch_j,\zch_j,\much_j) :
  |\gch_j - \gbar_j| \leq \domV \gbar_j^2 |\log \gbar_j|,\;
  \nnb & \qquad\qquad\qquad\qquad
  |\zch_j - \zbar_j|, |\much_j - \mubar_j| \leq \domV \chigen_j\gbar_j^2 |\log \gbar_j|
  \}
  \nnb
  & \qquad
  \times
  B_{\Wcalt_j(\mgen^2,\ggen_0)}( \domK \chigen_j\gbar_j^3).
\end{align}
Set $a=\DVa$.
By Lemma~\ref{lem:gbarmcomp}(i), for any fixed $C_\DV>1$ and $h>0$
(with $\delta$ chosen small depending on $h$), the domains of
\refeq{domRG-2}--\refeq{Djdef-2} obey $D_j \subset T_j(\domRG_j^{0})$.
Therefore, for any $(\Vch,K) \in D_j$, by Corollary~\ref{cor:step-mr-tr},
  $\rho=\Rch_+,\psi=\Kch_+$ obey the bounds of Theorem~\ref{thm:step-mr-iv}
(namely those of Theorem~\ref{thm:step-mr-fv}).
For convenience, we recall the
bounds for $\Kch_+$ from \refeq{DVKbd}, which are more delicate than those for $\Rch_+$:
\begin{align}
\lbeq{DVKbd-9}
    \|D_{V}^pD_{K}^{q}\Kch_+\|_{L^{p,q}}
    &\le
    \begin{cases}
    M  \chigen_{j+1} \ggen_{j+1}^{3-p}
    &
    (p \ge 0)
    \\
    \rlap{$\kappa$}\hspace{3.3cm}
    & (p=0,\, q=1)
    \\
    M  \ggen_{j+1}^{-p}
    (
    \chigen_{j+1}
    \ggen_{j+1}^{10/4}
    )^{1-q}
    &
    (p \ge 0,\, q \ge 1)
    .
    \end{cases}
\end{align}
As stated above the statement of Proposition~\ref{prop:flow-flow}, we
have made the choice $\DVa=4M$.  Also, since $\kappa = O(L^{-1})$, we
can and do assume that $\kappa < \half \Omega^{-1}$.  Then
\refeq{DVKbd-9} implies the bounds of Assumption~(A3) for $\Kch$,
namely
\begin{align}
\lbeq{DVKbd-flow-9}
    \|D_{V}^pD_{K}^{q}\Kch_+\|_{L^{p,q}}
    &\le
    \begin{cases}
    R' \chi_{j+1} \gbar_{j+1}^{3}
    &
    (p =  q=0; \, K_j=0)
    \\
    M' \chi_{j+1} \gbar_{j+1}^{2}
    &
    (p = 1; \, q=0)
    \\
    \rlap{$\kappa$}\hspace{3.3cm}
    & (p=0,\, q=1)
    \\
    M' (\gbar_{j+1}^2 |\log \gbar_{j+1}|)^{-p} (\chi_{j+1}\gbar_{j+1}^3)^{1-q}
    & (p+q=2, 3)
    ,
    \end{cases}
\end{align}
with constants $M' = 2M$ (to absorb a factor $2$ in replacing
$\ggen_j$ by $\gbar_j = \gbar_j(1+O(\gbar_j))$, and $R'=2M = \frac12
\domK < \domK(1-\kappa \Omega)$ (the factor $2$ is again to replace
$\ggen_j$ by $\gbar_j$, the second equality is our choice $a=4M$, and
the inequality follows since $\kappa\Omega < \frac12$).  Note that the
powers in \refeq{DVKbd-flow-9} are less restrictive than in
\refeq{DVKbd-9}.  The conclusion of the previous discussion is that,
for any $(\mgen^2,\ggen_0) \in [0,\delta) \times (0,\delta)$, the maps
$\Phi(\cdot, m^2)$ with $m^2 \in \Igen_+(\mgen^2)$ satisfy (A3) with
norms $\Wcalt_j(\mgen^2,\ggen_0)$ and parameters
$(\domK,\domV,\kappa,\Omega,R',M')$ depending on the constants in
Theorem~\ref{thm:step-mr-iv}, but not on any of
$\mgen^2,m^2,\ggen_0,g_0$.

We choose $b=\frac14$ (somewhat arbitrarily), $h>h_*$ large enough,
$\delta \leq g_*$, and apply Theorem~\ref{thm:flow4} to conclude that
for each $(\mgen^2,\ggen_0) \in (0,\delta)^2$, there is a
neighbourhood $N(\mgen^2,\ggen_0) = \Igen_+(\mgen^2) \times
\I(\ggen_0) \subset (0,\delta)^2$ of $(\mgen^2,\ggen_0)$ such that,
given $(m^2,g_0) \in N(\mgen^2,\ggen_0)$, there are solutions to
\eqref{e:Phi-flow} with maps $\Phi_j = \Phi_j(\cdot,m^2)$ that satisfy
the boundary conditions \eqref{e:Phi-bc} specified by the parameter
$g_0$.  We denote these solutions by
$\xch^{d}(\mgen^2,\ggen_0;m^2,g_0)$ where the argument
$(\mgen^2,\ggen_0)$ indicates the dependence on the Banach spaces
$\Wcalt_j(\mgen^2,\ggen_0)$ to which we apply Theorem~\ref{thm:flow4},
and where $(m^2,g_0)$ refers to the mass parameter $m^2$ of $\Phi_j$
and the initial condition $g_0$.

Choosing $(\mgen^2, \ggen) = (m^2,g)$,
we define a map
$\xch_j^c: [0,\delta) \times (0,\delta) \to \Vcal \oplus \Kspace_j(\Z^d)$ by
\begin{equation} \label{e:xss}
  \xch_j^c(m^2,g_0) = \xch_j^{d}(m^2,g_0;m^2,g_0) \in \Vcal \oplus \Wcalt_j(m^2,g_0) \subset \Vcal \oplus \Kspace_j(\Z^d)
  .
\end{equation}
We extend this map to $[0,\delta)^2$ by setting $\xch^{c}_j(m^2,0) = (0,0)$.
In particular, restriction to the $(z_0,\mu_0)$-component produces two real-valued maps
$z_0^c: [0,\delta)^2 \to \R$ and $\mu_0^c : [0,\delta)^2 \to \R$ such that
$z_0^c(m^2,0)=\mu_0^c(m^2,g_0)=0$, and
the flow with initial condition \eqref{e:flow-flow-ic} satisfies $(V_j,K_j) \in \domRG_j(s_j)$ for all $j\in\N_0$,
as desired.

In particular, this implies $\gch_j = O(\gbar_j)$.
Since $\zbar_j, \mubar_j = O(\chi_j\gbar_j)$, by Proposition~\ref{prop:approximate-flow},
\eqref{e:VVbar3app}--\eqref{e:VVbar4app} also show $\zch_j,\much_j = O(\chi_j \gbar_j)$.
By \eqref{e:TVV2}, it then follows that $g_j = O(\gbar_j)$.
Furthermore, by \cite[\eqref{pt-e:zch-def}--\eqref{pt-e:much-def}]{BBS-rg-pt},
$\much_j = \mu_j + O(\mu_j)^2$ and $\zch_j = z_j + O(z_j\mu_j) + O(\mu_j^2)$, and
$z_j,\mu_j = O(\chi_j \gbar_j)$ follows.
\end{proof}

\subsubsection{Step 2: regularity in the interior}

In preparation for step~2 of the proof,
we make the following two observations:

\smallskip\noindent (a)
The neighbourhood $N(\mgen^2,\ggen_0) = \Igen_+(\mgen^2) \times \I(\ggen_0)$
in the proof
  of (i--ii) can be taken instead to be of the form
  $N(\mgen^2,\ggen_0)
  = (\mgen^2-\varepsilon(\mgen^2,\ggen),\mgen^2+\varepsilon(\mgen^2,\ggen))
  \times (\ggen-\varepsilon(\mgen^2,\ggen),\ggen+\varepsilon(\mgen^2,\ggen))$,
  with $\varepsilon(\mgen^2,\ggen)>0$ bounded away from $0$ uniformly on compact subsets of $(0,\delta)^2$.

\smallskip\noindent (b)
The flexibility to choose any $b \in (0,1)$
in Theorem~\ref{thm:flow4} can be used
to enhance the statement of
Theorem~\ref{thm:flow4}, as follows.
We apply Theorem~\ref{thm:flow4} with two choices of $b$, namely
$b_1=\frac 12$ and $b_2=\frac 14$.
The values of $h_*,g_*$ in Theorem~\ref{thm:flow4} depend on $b$,
and we assume that
$h > \max\{h_*(b_1),h_*(b_2)\}$
but otherwise the choice of $h$ is arbitrary, and
$\delta \leq \min\{ g_*(b_1), g_*(b_2) \}$.
As in the proof of (i--ii), for $(\mgen^2,\ggen_0) \in (0,\delta)^2$, there exists
a neighbourhood $N(\mgen^2,\ggen_0)$ such that for each $(m^2,g_0) \in N(\mgen^2,\ggen_0)$, there exists
a solution that satisfies \eqref{e:VVbar2app}--\eqref{e:VVbar1app} with $b=b_2 =\frac14$
and $\Wcal_j=\Wcalt_j(\mgen^2, \ggen_0)$,
by applying the existence statement of Theorem~\ref{thm:flow4} with the latter parameters.
Moreover, any solution that satisfies \eqref{e:VVbar2app}--\eqref{e:VVbar1app}
with $b=b_1 =\frac12$
and $\Wcal_j = \Wcalt_j(\mgen^2,\ggen_0)$
is unique, by applying the uniqueness statement of Theorem~\ref{thm:flow4} with
$b=b_1 =\frac12$,
and hence must, in fact, satisfy \eqref{e:VVbar1app} with the smaller value of the
parameter $b=b_2 =\frac14$ and therefore be identical to the solution
produced when $b=b_2$.

\begin{proof}[Step 2]
We show that $z_0^c,\mu_0^c$ are continuous in
$(m^2,g_0) \in (0,\delta)^2$,
differentiable in $g_0 \in (0,\delta)$, with uniformly bounded $g_0$-derivative.

We set $b=b_2 = \frac 14$ as in step~1.
We fix some $(\hat m^2,\hat g_0) \in (0,\delta)^2$ and show that
$z_0^c,\mu_0^c$ are
continuous in both variables
and differentiable in $g_0$ at this $(\hat m^2,\hat g_0)$.
By Corollary~\ref{cor:flow4-masscont}, for any $(\mgen^2,\ggen_0)\in (0,\delta)^2$,
$\xch_j^{d}(\mgen^2,\ggen_0; m^2, g_0)$ is continuous in $(m^2,g_0)\in N(\mgen^2,\ggen_0)$,
as a map taking values in the Banach space $\Vcal \oplus \Wcalt_j(\mgen^2,\ggen_0)$.
Moreover, by Theorem~\ref{thm:flow4}(ii), this map is differentiable in $g_0$, for fixed $(\mgen^2,\ggen_0,m^2)$,
and its derivative is bounded uniformly in $(\mgen^2,m^2,\ggen_0,g_0)$.
It therefore suffices to show that, given $(\hat m^2, \hat g_0)$, there exists $(\mgen^2,\ggen_0)$ such that
$(\hat m^2,\hat g_0) \in N(\mgen^2,\ggen_0)$, and that for all $(m^2,g_0)$ in a neighbourhood of $(\hat m^2,\hat g_0)$,
$\xch_j^{d}(m^2,g_0;m^2,g_0) = \xch_j^{d}(\mgen^2,\ggen_0;m^2,g_0)$.
Such a neighbourhood is constructed in the next paragraph.
Assuming its existence, the continuity in $(m^2,g_0)$ and the differentiability in $g_0$
of $\xch_j^c(m^2,g_0)$ follow from the already established properties of $\xch_j^{d}(\mgen^2,\ggen_0; \cdot, \cdot)$.

We now choose $(\mgen,\ggen_0)$ and the neighbourhood of $(\hat m^2,\hat g_0)$,
discussed in the previous paragraph.
By item~(a) above, given $(\hat m^2,\hat g_0)$,
the radius of $N(\mgen^2,\ggen_0)$ is uniformly bounded below by some $\hat\varepsilon >0$ on the compact set
$\tilde m^2 \in [\half \hat m^2, \hat m^2]$, $\tilde g_0 \in [\hat g_0, (1+\delta)\hat g_0]$.
We assume that $\hat\varepsilon < \max\{\delta,\half\}$ and
choose $\mgen^2 = \hat m^2 - \half \hat\varepsilon$ and $\ggen_0 = \hat g_0 + \half \hat\varepsilon$.
Then
$N(\mgen^2,\ggen_0)
\supset [\hat m^2 - \frac14 \hat\varepsilon,\hat m^2+ \frac14 \hat\varepsilon]
\times [\hat g_0 - \frac14 \hat\varepsilon, \hat g_0+ \frac14 \hat\varepsilon] = Q$.
By definition, $\mgen^2 \leq m^2$ and $g_0 \leq \ggen_0 \leq (1+\delta)g_0$
hold for all $(m^2,g_0) \in Q$. Thus, for $(m^2,g_0) \in Q$, by \eqref{e:normequivg0m},
\begin{equation}
\begin{aligned}
\lbeq{Kmmgg}
  &\|K^c(m^2,g_0;m^2,g_0)-\bar K(m^2,g_0)\|_{\Wcalt_j(\mgen^2,\ggen_0)}
  \\&\qquad \leq
  2\|K^c(m^2,g_0;m^2,g_0)-\bar K(m^2,g_0)\|_{\Wcalt_j(m^2,g_0)}
  \\
  & \qquad
  \leq 2 b_2 (a-a_*) \gbar_j^3
  = b_1 (a-a_*) \gbar_j^3
  ,
\end{aligned}
\end{equation}
and the bounds \eqref{e:VVbar2app}--\eqref{e:VVbar4app} hold with $b=b_2<b_1$.
As discussed in item~(b) above, solutions for which \refeq{Kmmgg} and
\eqref{e:VVbar2app}--\eqref{e:VVbar4app} hold with $b_1$ are unique, from which we
conclude that $\xch^{d}(\mgen^2,\ggen_0;m^2,g_0) =
\xch^{d}(m^2,g_0;m^2,g_0)$ as desired.
\end{proof}

\subsubsection{Step 3: regularity at the boundary}

\begin{proof}[Step 3]
We show that $\z_0^c,\mu_0^c$ are
continuous as $g_0 \downarrow 0$ or as $m^2 \downarrow 0$.

The bounds $\mu_0^c, z_0^c = O(g_0)$, which hold uniformly in $m^2$, imply that $\mu_0^c,z_0^c \to 0$
as $(m^2,g_0) \to (\hat m^2,0)$ for some $\hat m^2 \in [0,\delta)$.
The continuity in the limit $g_0 \downarrow 0$ is immediate from this.

For continuity of $\mu_0^c,z_0^c$ as $(m^2,g_0) \to (0,\hat g_0)$
for some $\hat g_0 \in (0,\delta)$,
we show that $\xch_0(m^2, g_0) \to\xch_0(0,\hat g_0)$
by adapting the argument used
in the proof of Corollary~\ref{cor:flow4-masscont}
(given in \cite[Corollary~\ref{flow-cor:masscont}]{BBS-rg-flow}), as follows.
First, we recall that
\begin{align}
  |\gch_j(m^2,g_0) - \bar g_j(m^2,g_0)| \lbeq{VVbar2mp}
  &\leq \b \domV \gbar_j(m^2)^2 |\log\gbar_j(m^2,g_0)|
  ,
  \\
  |\much_j(m^2,g_0) - \bar \mu_j(m^2,g_0)| \lbeq{VVbar3mp}
  &\leq \b \domV \chi_j(m^2) \gbar_j(m^2,g_0)^2 |\log\gbar_j(m^2,g_0)|
  ,
  \\
  |\zch_j(m^2,g_0) - \bar z_j(m^2,g_0)| \lbeq{VVbar4mp}
  &\leq \b \domV \chi_j(m^2) \gbar_j(m^2,g_0)^2 |\log\gbar_j(m^2,g_0)|
  ,
  \\
  \|K_j(m^2,g_0) - \bar K_j(m^2,g_0)\|_{\Wcalt_j(m^2,g_0)} \lbeq{VVbar1mp}
  &\leq \b (\domK-\domK_*) \chi_j(m^2) \gbar_j(m^2,g_0)^3
  .
\end{align}
By Proposition~\ref{prop:approximate-flow},
$\bar V_j(m^2,g_0)$ is continuous in $(m^2,g_0^2) \in [0,\delta]$, and thus in particular
$\bar V_0(m^2,g_0^2)$ is uniformly bounded for $(m^2,g_0^2) \in [0,\delta]^2$.
With \refeq{VVbar2mp}--\refeq{VVbar4mp}, we see that
$\Vch_0(m^2,g_0)$ is therefore also uniformly bounded.
Thus, for every sequence $(m^2,g_0) \to (0,\hat g_0)$, $\Vch_0(m^2,g_0)$ has a limit point.
It suffices to show that any such limit point is equal to $\Vch_0(0,\hat g_0)$.
To show this uniqueness, we fix an arbitrary limit point $V_0^*$
and a sequence $(m^2, g_0) \to (0,\hat g_0)$ such that $\Vch_0(m^2, g_{0}) \to V_0^*$.
We also set $K_0^*=K_0=0$.

Then we define $x_j^* = (V_j^*,K_j^*)$ by inductive application of $\Phi_j(\cdot,0)$
starting from $x_0^* = (V_0^*,K_0^*)$, as long as $x_j^* \in D_j(0,\hat g_0)$.
For $m^2=0$, the continuity interval in Theorem~\ref{thm:step-mr-fv} is
$\Igen_{j+1}(0) = [0,L^{-2j}]$ (recall \refeq{massint} and \refeq{Itilint}).
Since $\Vch_0(m^2,g_0) \to V_0^*$, it follows by
induction and the continuity of
the maps
$\rho_j=\Rch_{j+1}: \Wcalt_j(0, \hat g_0) \times [0,L^{-2j}] \to \Vcal$
and
$\psi_j=\Kch_{j+1}: \Wcalt_j(0, \hat g_0) \times [0,L^{-2j}] \to \Wcalt_{j+1}(0,\hat g_0)$,
that $x_j(m^2,g_{0}) \to x_j^*$ in $\Vcal \times \Wcalt_j(0,\hat g_0)$.
By an analogous induction, using continuity of $\bar V_j$ and $\psi_j$,
it follows that  $\bar K_j(m^2, g_{0}) \to \bar K_j(0,\hat g_0)$ in $\Wcalt_j(0,\hat g_0)$.
Since $\chi_j(m^2) \to \chi_j(0) = 1$, we can now take the limit of
\refeq{VVbar2mp}--\refeq{VVbar1mp} along the sequence $(m^2,g_{0}) \to (0,\hat g_0)$
and obtain, with \eqref{e:normindep},
\begin{align}
  |g_j^* - \bar g_j(0, \hat g_0)|
  &\leq \b \domV \gbar_j(0, \hat g_0)^2 |\log\gbar_j(0, \hat g_0)|
  ,
  \\
  |\mu_j^* - \bar \mu_j(0, \hat g_0)|
  &\leq \b \domV
  \gbar_j(0,\hat g_0)^2 |\log\gbar_j(0, \hat g_0)|
  ,
  \\
  |z_j^* - \bar z_j(0, \hat g_0)|
  &\leq \b \domV
  \gbar_j(0,\hat g_0)^2 |\log\gbar_j(0, \hat g_0)|
  ,
  \\
  \|K_j^* - \bar K_j(0, \hat g_0)\|_{\Wcalt_j(0,\hat g_0)} \lbeq{KKbarm}
  &\leq \b (\domK-\domK_*)
  \gbar_j(0, \hat g_0)^3
  .
\end{align}
This shows inductively that $x_j^{*}$ does remain in $D_j(0,\hat g_0)$ for all $j$,
so the above inductions can be carried out indefinitely.
By the uniqueness assertion of Theorem~\ref{thm:flow4}, $x_j^* = \xch_j(0,\hat g_0)$.
In particular, $V_0^* = \Vch_0(0,\hat g_0)$, so $\Vch_0(m^2,g_0) \to \Vch_0(0,\hat g_0)$
as $(m^2,g_0) \to (0,\hat g_0)$.
This shows that $z_0^c,\mu_0^c$ are continuous as $m^2 \downarrow 0$, as claimed.
\end{proof}

\section{Proof of main result}
\label{sec:pfmr}

It is shown in Section~\ref{sec:chvar} that the main result,
Theorem~\ref{thm:suscept}, is a consequence of Theorem~\ref{thm:suscept-diff}.
We now prove Theorem~\ref{thm:suscept-diff}.
Functions  $z_0^c$ and $\nu_0^c=\mu_0^c$ with
the regularity properties required by Theorem~\ref{thm:suscept-diff}
were constructed  in  Proposition~\ref{prop:flow-flow}.
We  now show that there exists $\delta>0$ such that
for $m^2,g_0,\hat g_0 \in (0,\delta)$,
\begin{align} \label{e:chi-m-pf}
  \hat\chi \left( m^2,g_0,\nu_0^c(m^2,g_0),z_0^c(m^2,g_0) \right)
  &= \frac{1}{m^2} ,
  \\
  \label{e:chiprime-m-pf}
  \ddp{\hat\chi}{\nu_0} \left(m^2,g_0,\nu_0^c(m^2,g_0),z_0^c(m^2,g_0) \right)
  &\sim - \frac{1}{m^4} \frac{c(\hat g_0)}{(\hat g_0\bubble_{m^2})^{\gamma}}
  \quad \text{as $(m^2,g_0)  \to (0,\hat g_0)$},
\end{align}
with $c(g_0)$ continuous and equal to $1+O(g_0)$.
This will complete the proof of Theorem~\ref{thm:suscept-diff}.
From this, an elementary analysis yields also Theorem~\ref{thm:nuc}.

\subsection{Estimates for finite volume}
\label{sec:fintori}

Let $(z_0^c(m^2,g_0), \mu_0^c(m^2,g_0))$ be the functions of
Proposition~\ref{prop:flow-flow}, let
\begin{equation}
\lbeq{V0c}
  V_{0}^c(m^2,g_0)
  = (g_0, z_0^c(m^2,g_0), \mu_0^c(m^2,g_0)),
\end{equation}
and let $V_j$ be the sequence determined by this
initial condition in Proposition~\ref{prop:flow-flow}.
Proposition~\ref{prop:flow-flow} considered the
the \emph{infinite volume flow}, and we now wish to consider the finite
volume flow.  The following proposition, which is the basis for
the proof of Theorem~\ref{thm:suscept-diff}, constructs the finite volume flow.
Its $V_j$ component is \emph{identical} to the above infinite volume sequence
of $V_j$, though of course only $j \le N$ is meaningful for $\Lambda_N$.
In particular, the sequence $(V_j)_{0\le j \le N}$ is independent of the
volume, in the sense that it is the same for $\Lambda_N$
as it is on any $\Lambda_{N'}$ with $N'>N$.

\begin{prop} \label{prop:KjNbd}
  Let $d=4$, $N\in \N$, $g_0\in (0,\delta)$ and $m^2 \in [0,\delta)$.
  Let $\ggen_j=\ggen_j(m^2,g_0)$ and $s_j = (m^2,\ggen_j)$.
  Let $K_0=\1_\varnothing$, and let $V_j$ be the sequence
  in the infinite volume global flow
  determined by the initial condition \refeq{V0c}.
  Then, with the additional proviso that for the case $j+1=N$ we restrict to
  $m^2 \in [\delta L^{-2(N-1)},\delta)$,
  the finite volume renormalisation group flow $(V_{j},K_{j,N}) \mapsto (V_{j+1},K_{j+1,N})$
  of Theorem~\ref{thm:step-mr-fv} exists for all $j<N$, with
  $(V_j,K_j) \in \domRG(s_j,\Lambda_N)$.
  In particular,
  \begin{align}
  \label{e:KjNbd}
    \|K_{j,N}\|_{\Wcal_j(s_j,\Lambda_N)} &\leq O(\chi_j \gbar_j^3)
    \quad
    \quad
    (j \le N),
  \end{align}
  $\gch_j,g_j = O(\gbar_j)$, and $\zch_j,\much_j,z_j,\mu_j = O(\chi_j\gbar_j)$.
\end{prop}

\begin{proof}
  The proof is by induction on $0\leq j\leq N$.
  We write $V_{j,N}$ for the perturbative coordinate, and show in the course
  of the proof that it is equal to $V_j=V_{j,\Zd}$ as claimed.
  The induction hypothesis is
  that \refeq{KjNbd} holds,
  that $V_{j,N}=V_{j,\Zd}$, and
  that the family $(K_{j,\Lambda})$ has Property~$(\Zd)$ discussed in Section~\ref{sec:rg-iv}.
  For $j=0$, the fact that \refeq{KjNbd} holds is true vacuously since $K_0=\1_\varnothing$,
  the fact that $V_{0,N}=V_{0,\Zd}$ holds by definition, and the fact
  that $(K_{0,\Lambda})$ has Property~$(\Zd)$ is again vacuous.
  To advance the induction, it follows from \refeq{Kiv}--\refeq{Viv}
  that $V_{j+1,N}=V_{j+1,\Zd}$, and
  that $(K_{j+1,\Lambda})$ has Property~$(\Zd)$.
  By 
  Proposition~\ref{prop:flow-flow}, 
  $V_{j+1,N} = V_{j+1,\Z^d}$ lies in the set appearing in the
  domain $\domRG_j(m^2,\ggen_j,\Z^d)$
  of \refeq{domRG}, which is the same set as in the domain $\domRG_j(m^2,\ggen_j,\Lambda_N)$.
  Together with the induction hypothesis on $K_{j,N}$, we conclude that
  $(V_{j,N},K_{j,N}) \in \domRG_j(m^2,\ggen_j,\Lambda)$.
  By Theorem~\ref{thm:step-mr-fv} with $\mgen^2=m^2$, this implies that
  \refeq{KjNbd} holds when $j$ is replaced by $j+1$.
  (The restriction on $m^2$ at the last scale arises
  from $\Iint_N$ of \refeq{massint}.)
  This completes the induction and the proof.
\end{proof}

In particular, it follows from Proposition~\ref{prop:KjNbd} and \refeq{EIK} that,
with the initial choice $V_0=V_{0}^c(m^2,g_0)$,
and with $I_j=I_j(V_{j})$,
\begin{equation}
\lbeq{ZIKNa}
    Z_j=(I_j\circ K_{j,N})(\Lambda) \quad \quad (0 \le j \le N).
\end{equation}
This identity continues to hold for initial conditions $V_0 \neq V_0^c$,
as long as the right-hand side remains well-defined,
which is the case in an $N$-dependent neighbourhood of $V_0^c$, by continuity.
In addition to the renormalisation group flow $(V_j,K_j)$, it will be convenient
to use its transformed version, defined in Section~\ref{sec:trans}.
According to \eqref{e:flowVK}, the transformed flow $(\Vch_j,K_j) = (T_j(V_j),K_j)$ satisfies,
for all $j<N$,
\begin{align}
\lbeq{fvflow}
    (\Vch_{j+1},K_j)
    &
    =
    \big(\bar\varphi_j(\Vch_j) + \Rch_{j+1}(\Vch_j,K_j),
    \Kch_{j+1}(\Vch_j,K_j)
    \big),
\end{align}
where, by Corollary~\ref{cor:step-mr-tr},
the maps $\Rch_{j+1}$ and $\Kch_{j+1}$ obey the estimates \refeq{Rmain-g}--\refeq{DVKbd}
for $0 \le j <N$.

\subsection{Continuous functions of the renormalisation group flow}%
\label{sec:Fcont}

We now show how
the regularity statements of Theorem~\ref{thm:step-mr-fv}, which apply in the \emph{local}
domains $\domRG(\sgen) \times \Igen_+(\mgen^2)$, imply suitable
regularity statements of functions of the flow of
Proposition~\ref{prop:KjNbd},
on the larger domain $[0,\delta) \times (0,\delta)$.
We restrict attention here to the case of finite volume $\volume = \Lambda_N$.

\begin{defn} \label{def:flowcont}
  (i)
  A map $(V,K,m^2) \mapsto F(V,K,m^2)$ acting on a subset
  of $\Vcal \times \Kspace_j \times [0,\delta)$ with values in a Banach space $E$
  is a \emph{continuous function of the renormalisation group coordinates
  at scale-$j$},
  if its domain includes $\domRG_j(\sgen_j) \times \Igen_{j+1}(\mgen^2)$
  for all $\sgen_0 \in [0,\delta) \times (0,\delta)$, and if its
  restriction to the domain $\domRG_j(\sgen_j) \times \Igen_{j+1}(\mgen^2)$
  is continuous as a map $F: \domRG_j(\sgen_j) \times \Igen_{j+1}(\mgen^2) \to E$,
  for all $\sgen_0 \in [0,\delta) \times (0,\delta)$.
  We also say that $F$ is a $C^0$ map of the renormalisation group coordinates.
  \smallskip
  \\
  \noindent
  (ii) For $k \in \N$, a map $F$ is a
  \emph{$C^k$ map of the renormalisation group coordinates at scale-$j$},
  if it is a $C^0$ map  of the renormalisation group coordinates,
  its restrictions to the domains  $\domRG_j(\sgen_j) \times \Igen_{j+1}(\mgen^2)$
  are $k$-times continuously Fr\'echet differentiable in $(V,K)$,
  and every Fr\'echet derivative in $(V,K)$,
  when applied as a multilinear map to directions $\dot{V}$ in $\Vcal^{p}$
  and $\dot{K}$ in $\Wcal^{q}$, is jointly continuous in all arguments,
  $m^{2}, V,K, \dot{V}, \dot{K}$.
\end{defn}

Two examples of $C^k$ maps of the renormalisation group coordinates are $F(V,K,m^2) = V$,
and the map $R_+$ of Theorem~\ref{thm:step-mr-fv}.
The map $K_+$ is a \emph{not} a $C^k$ map in the above sense, since it does not act with a common
target space $E$ (its image space $\Wcal_+$ depends on $\sgen$).
In all our applications, $E$ is finite-dimensional.

For $s_0 = (m^2,g_0) \in [0,\delta) \times (0,\delta)$,
recall the definition of $V_0^c(s_0)$ from \refeq{V0c}.
Let $(V_j,K_j) = (V_j(m^2,V_0),K_j(m^2,V_0))$ be the flow of $(V_+,K_+)$
with initial condition $V_0$ (not necessarily equal to $V_0^c$) and mass parameter $m^2$.
By Proposition~\ref{prop:KjNbd}, this flow exists for all $j < N$ if $V_0 = V_0^c$,
and also exist for $j=N$ if $m^2 \in [\delta L^{-2(N-1)}, \delta)$.
Given a $C^k$ map $F$ of the renormalisation group coordinates at scale-$j$,
let $D_{V_0}^k F(s_0)$ denote the $k^{\rm th}$ derivative of $F(V_j(m^2,V_0),K_j(m^2,V_0),m^2)$ with respect to the initial condition $V_0$,
evaluated at $V_0=V_0^c(s)$.
The following proposition shows that $D_{V_0}^kF$ exists and is continuous in $s_0$.

\begin{prop} \label{prop:Fcont}
  Let $j < N(\volume)$, $k \in \N_0$,
  and let $F$ be a $C^k$ map of the renormalisation group coordinates at scale-$j$.
  Then, for every $p \leq k$,
  all $s_0 \in [0,\delta) \times (0,\delta)$, the derivative $D_{V_0}^p F(s_0)$ exists, and
  \begin{equation} \label{e:Fcont}
    \text{$s_0 \mapsto D_{V_0}^p F(s_0)$ is a continuous map $[0,\delta) \times (0,\delta) \to L^p(\Vcal, E)$,}
  \end{equation}
  where $L^p(\Vcal,E)$ is the space of $p$-linear maps from $\Vcal$ to $E$ with the operator norm.
\end{prop}

In particular, with $k=0$,
continuous maps of the renormalisation group coordinates at scale-$j$,
in the sense of Definition~\ref{def:flowcont},
induce continuous functions in the ordinary sense of the parameters $(m^2,g_0)$
of the renormalisation group flow.

To give an example of an application of Proposition~\ref{prop:Fcont}, it follows from
Corollary~\ref{cor:step-mr-tr} that
\begin{equation} \label{e:gchrec-0}
  \gch_{j+1}
  = \gch_j - \beta_j \gch_j^2 + r_j \quad\text{with}\quad r_j=O(\chicCov_j\gch_j^3).
\end{equation}
It also follows from Theorem~\ref{thm:step-mr-fv} that the function
$F(V_j,K_j,m^2) = g_{j+1}$ is a $C^\infty$ map of the renormalisation group coordinates
at scale-$j$.  Therefore, by Proposition~\ref{prop:Fcont}, we conclude
the continuity of $g_j$ in $(m^2,g_0) \in [0,\delta) \times (0,\delta)$, for all $j\le N$.
Furthermore, since $g_j = O(g_0)$, the continuity extends to
$(m^2,g_0) \in [0,\delta)^2$.  From \refeq{gchrec-0} and the continuity of $\beta_j$
in $m^2 \in [0,\delta)$, we therefore also have continuity of $r_j=r_j(m^2,g_0)$ in
$(m^2,g_0) \in [0,\delta)^2$, for all $j<N$.

\begin{proof}[Proof of Proposition~\ref{prop:Fcont}]
Let $\sgen_0 =(\mgen^2,g_0) \in [0,\delta) \times (0,\delta)$,
and set $\sgen_j = (\mgen^2, \ggen_j(\mgen^2,\ggen_0))$.
For $s_0 \in [0,\delta) \times (0,\delta)$ sufficiently close to $\sgen_0$,
we write $[D_{V_0}^k(V_j, K_j)](s_0)$ for the $k^{\rm th}$ derivative of $(V_j(m^2,V_0),K_j(m^2,V_0))$
with respect to the initial condition $V_0$, evaluated at $V_0=V_0^c(s_0)$ of \eqref{e:V0c},
and with derivative taken in $\Vcal \times \Wcal_j(\sgen_j)$.
In particular, $[D_{V_0}^0(V_j,K_j)](s_0) = (V_j(s_0), K_j(s_0))$ is the
renormalisation group flow with initial condition
determined by $s_0$
in Proposition~\ref{prop:KjNbd}.
To see that the derivatives exist, we
claim that for every $\sgen_0 = (\mgen^2,\ggen_0) \in [0,\delta) \times (0,\delta)$,
$k \in \N_0$, and $j < N$,
\begin{equation} \label{e:VKcontk}
  \text{$s_0 \mapsto D_{V_0}^k(V_j,K_j)$ is
    continuous in $L^k(\Vcal,\Vcal \times \Wcal_j(\sgen_j))$ as $s_0 \to \sgen_0$.}
\end{equation}

We prove \refeq{VKcontk} as follows.  First, we know from  Proposition~\ref{prop:KjNbd}
that $(V_j(\sgen_0),K_j(\sgen_0)) \in \domRG_j(\sgen_j)$,
and that $z_0^c,\mu_0^c$ are continuous.  Also,
the maps $(V_+,K_+)$ are continuous by Theorem~\ref{thm:step-mr-fv}.  Therefore,
there exist neighbourhoods $N_j=N_j(\sgen_0)$ of $\sgen_0$ with
  $(V_j(s_0),K_j(s_0)) \in \domRG_j(\sgen_j)$ for $s_0 \in N_j$,
  and such that $(V_j(s_0),K_j(s_0))$ is continuous in $s_0 \in N_j$ as a map with values in $\Vcal \times \Wcal_j(\sgen_j)$.
  In particular, for $j<N$,
  the case $k=0$ of $(V_j,K_j)$ as $s_0 \to \sgen_0$ in \eqref{e:VKcontk} follows.

  For $k > 0$, the existence and continuity of $D_{V_0}^k(V_j,K_j)$ in \eqref{e:VKcontk}
  is a consequence of the continuity the derivatives of the maps $V_+,K_+$ in
  $\domRG_j(\sgen_j) \times \Igen_{j+1}(\mgen^2)$ provided by
  Theorem~\ref{thm:step-mr-fv}.
  By induction, we assume that the claim holds for some $j<N$. It is vacuous for $j=0$.
  To advance the induction, we consider, e.g., $D_{V_0}K_{j+1}$.
  By the chain rule,
  \begin{equation} \label{e:Kmchain}
    D_{V_0}K_{j+1} = D_VK_{j+1}(V_j,K_j,m^2)D_{V_0}V_j + D_K K_{j+1}(V_j,K_j,m^2)D_{V_0}K_j.
  \end{equation}
  By the inductive assumption, each of $V_j,D_{V_0}V_j,K_j,D_{V_0}K_j$ is continuous as $s_0 \to \sgen_0$.
  The continuity statement of  Theorem~\ref{thm:step-mr-fv}
  advances the induction.
  An analogous argument applies to $D_{V_0}V_{j+1}$, and also to all higher derivatives.
  This completes the proof of \refeq{VKcontk}.

  The assumption that $F$ is a continuous map of the renormalisation group
  coordinates
  and
  \refeq{VKcontk}
  imply that
  $F(V_j(s_0),K_j(s_0),m^2) \to F(V_j(\sgen_0),K_j(\sgen_0),\mgen^2)$
  as $s_0 \to \sgen_0$. Since $\sgen_0 \in [0,\delta) \times (0,\delta)$ is arbitrary,
  the claim \eqref{e:Fcont} follows for $k=0$.
  Similarly, for $k>0$,
  the proof of \eqref{e:Fcont} follows from the chain rule and \refeq{VKcontk}.
\end{proof}

\subsection{Susceptibility and renormalisation group flow}
\label{sec:finscale}

As the first step in the analysis of the susceptibility, we
express $\hat\chi_N$ in terms of the renormalisation group coordinates.
Let $N \in \N$ and let $m^2 \in [\delta L^{-2(N-1)},\delta)$.
Recall from \eqref{e:chibarm} that
\begin{equation}
\label{e:chibarm-bis}
  \hat\chi_N
  = \frac{1}{m^{2}} + \frac{1}{m^4} \frac{1}{|\Lambda|}  D^2Z_N^0(0, 0; 1,1),
\end{equation}
with $Z_N^0$ the degree-$0$ part of the form $Z_N = \Ex_C\theta Z_0$.
Let $(V_j,K_{j,N})$ be the finite volume flow given by Proposition~\ref{prop:KjNbd}.
With the abbreviation $K_N=K_{N,N}$, it follows from
\refeq{ZIKNa}, together with the fact that
$\mathcal{B}_N(\Lambda_N)$ consists only of $\varnothing,\Lambda$, that
\begin{equation}
\lbeq{ZIKN}
  Z_N = (I_N \circ K_N)(\Lambda)
  = I_N(\Lambda) + K_N(\Lambda)
  .
\end{equation}
Therefore,
\begin{align}
\label{e:chibarm-bisIK}
  \hat\chi_N
  =
  \frac{1}{m^{2}} +  \frac{1}{m^4}\frac{1}{|\Lambda|}  D^2I_N^0(0, 0; 1,1)
  + \frac{1}{m^4}\frac{1}{|\Lambda|}  D^2K_N^0(0, 0; 1,1).
\end{align}

By \eqref{e:Idef}, $I_N(\Lambda) = e^{-V_N(\Lambda)}(1+W_N(\Lambda))$,
so $I_N^0(\Lambda) = e^{-V_N^0(\Lambda)}(1+W_N^0(\Lambda))$ where $V_N^0(\Lambda)$ and $W_N^0(\Lambda)$
are the degree-$0$ parts of the forms $V_N(\Lambda)$ and $W_N(\Lambda)$, respectively.
Thus
\begin{equation}
\lbeq{D2Icalc}
   D^2 I_N^0(\Lambda; 0, 0; 1, 1)
  =  D^2 e^{-V_N^0}(\Lambda;0, 0; 1, 1)
  + D^2W_N^0(\Lambda; 0,0; 1, 1)
  ,
\end{equation}
since cross-terms cancel when $\phi=0$
because $W_N^0$ is a polynomial in $\phi$ with no monomials
of degree below two.
The first term on the right-hand side of \eqref{e:D2Icalc} can be evaluated by direct calculation,
using \eqref{e:VX} and \eqref{e:taudef}--\eqref{e:addDelta}, to give
\begin{equation}
  \label{e:D2eVN}
   D^2 e^{-V_N^0}(\Lambda;0, 0; 1, 1)
  = - \sum_{x,y} \nu_N \delta_{xy} 1_x1_y
  -  \sum_{x,y} z_N (-\Delta_{xy}) 1_x1_y
  = - \nu_N |\Lambda|
  ,
\end{equation}
since the quartic term $\tau^2$ does not contribute, and $\Delta 1 = 0$.
This gives the identity
\begin{equation}
  \label{e:chibarm-bis2}
  \hat\chi_N
  =
  \frac{1}{m^{2}} -  \frac{\nu_N}{m^4}
  + \frac{1}{m^4}\frac{1}{|\Lambda|}  D^2W_N^0(0, 0; 1,1)
  + \frac{1}{m^4}\frac{1}{|\Lambda|}  D^2K_N^0(0, 0; 1,1).
\end{equation}

By Proposition~\ref{prop:flow-flow}, $\nu_N = O(\chi_N L^{-2N} \gbar_N)$,
which converges exponentially to $0$, so the second term in
\eqref{e:chibarm-bis2} does not contribute to the limit $N\to\infty$.
We will prove \eqref{e:chi-m-pf} by showing that the two rightmost terms
in \eqref{e:chibarm-bis2} are even
smaller, by factors $O(\gbar_N)$ and $O(\gbar_N^2)$, respectively,
so that the limit is $m^{-2}$ as desired.

For \eqref{e:chiprime-m-pf}, the constant term $m^{-2}$ in \eqref{e:chibarm-bis2}
does not contribute
because its derivative is zero.
We denote the derivative of $(V_N,K_N)$ with respect to $\nu_0$,
evaluated at $V_0 = V_0^c$, by $V_N' = (g_N', z_N', \mu_N')$ and $K_N'$.
We will show that $-m^{-4}\nu_N'$ in the derivative of \eqref{e:chiprime-m-pf}
is asymptotic to the right-hand side of \refeq{chiprime-m-pf},
and that the derivatives of the last two terms in \eqref{e:chibarm-bis2}
are again smaller by $O(\gbar_N)$ and $O(\gbar_N^2)$, respectively.
To execute this strategy, we use the following sequence of lemmas.
The bubble diagram $\bubble_{m^2}=8B_{m^2}$ defined in \eqref{e:freebubble},
 with its logarithmic divergence in $d=4$, enters
via Lemmas~\ref{lem:betasum}--\ref{lem:ginfty}.
The power $\gamma = \frac 14$ for the logarithmic correction
to the susceptibility in $d=4$ arises
in Lemma~\ref{lem:gzmuprime}.
Lemmas~\ref{lem:ginfty}--\ref{lem:gzmuprime}
involve extensions of arguments used in \cite{BBS-rg-flow}.

\begin{lemma} \label{lem:betasum}
  For $m^2 >0$,
  \begin{equation}
    \sum_{j=0}^{\infty} \beta_j
    =
    \bubble_{m^2}.
  \end{equation}
\end{lemma}

\begin{proof}
By the definition of $\beta_j$ in \refeq{betadef},
  \begin{equation}
    \sum_{j=0}^{k-1} \beta_j
    =
    8\sum_{x\in\Z^d} w_k(x)^2
    ,
  \end{equation}
since the left-hand side is a telescoping sum.
Since the terms in the covariance decomposition are positive definite,
the Fourier transforms satisfy $0\leq \hat w_k \leq \hat C$, where
$C=(-\Delta_\Zd + m^2)^{-1}$ is the Green function on $\Zd$.
By the Parseval relation, the dominated converge theorem,
and \eqref{e:freebubble},
\begin{equation}
  \lim_{k \to \infty}
  \sum_{x\in\Z^d} w_k(x)^2
  =
  \lim_{k \to \infty}
  \int_{[-\pi,\pi]^d} \hat w_k(p)^2
  \frac{dp}{(2\pi)^d}
  =
  \int_{[-\pi,\pi]^d} \hat C(p)^2
  \frac{dp}{(2\pi)^d}
  =
  B_{m^2}
  .
\end{equation}
Since $\bubble_{m^2} = 8 B_{m^2}$, the proof is complete.
\end{proof}

For the subsequent analysis, we note the following facts.
First, in the upper bounds of \refeq{Rmain-g}--\refeq{DVKbd},
$\ggen_j$ can be replaced by $\gch_j$ or $\gbar_j$ due to Lemma~\ref{lem:gbarmcomp}, e.g.,
\begin{equation} \label{e:gchrec}
  \gch_{j+1}
  = \gch_j - \beta_j \gch_j^2 + r_j
  \quad\text{with}\quad r_j=O(\chicCov_j \gch_j^3)
  = O(\gbar_j^3).
\end{equation}
As argued below \refeq{gchrec-0}, $r_j=r_j(m^2,g_0)$ is continuous on $[0,\delta)^2$.
Secondly, it follows from \refeq{gchrec} that
for any continuously differentiable function $\psi: (0,\infty) \to \R$,
\begin{equation}
  \label{e:gbarsumbisx}
  \sum_{l=j}^{k} (\beta_l \gch_l^2 - r_l) \psi(\gch_l)
  = \int_{\gch_{k+1}}^{\gch_{j}} \psi(t) \; dt
  + O\left(\int_{\gch_{k+1}}^{\gch_{j}} t^2 |\psi'(t)| \; dt
  \right).
\end{equation}
The formula \eqref{e:gbarsumbisx} is proved in
\cite[\eqref{flow-e:gbarsumbis}]{BBS-rg-flow} for the special case $r_j=0$,
but the same proof applies when $r_j = O(\chi_j \gch_j^3)$.
We also use the fact,
proved in \cite[Lemma~\ref{flow-lem:elementary-recursion}(ii)]{BBS-rg-flow}, that
\begin{equation}
  \lbeq{chigbd-bis}
  \sum_{l=j}^k \chi_l \bar g_l^n \leq C_{n}
  \begin{cases}
    |\log \bar g_{k}| & (n = 1)
    \\
    \chi_j \bar g_j^{n-1} & (n > 1)
    .
  \end{cases}
\end{equation}

\begin{lemma} \label{lem:ginfty}
  Let $d=4$.
  For $(m^2,g_0) \in (0,\delta)^2$,
  the limit $\gch_\infty = \lim_{j \to \infty} \gch_j$  exists,
  is continuous in $(m^2,g_0)$, and extends continuously to $[0,\delta)^2$.
  For $\hat g_0 \in (0,\delta)$,   as $m^2 \downarrow 0$ and $g_0 \to \hat g_0$,
  \begin{equation} \label{e:ginfty}
    \gch_\infty \sim \frac{1}{\bubble_{m^2}}
    .
  \end{equation}
\end{lemma}

\begin{proof}
As mentioned above, the remainder $r_j = r_j(m^2,g_0)$ in \refeq{gchrec}
is a continuous function of $(m^2,g_0)\in [0,\delta)^2$.
The solution to the recursion \refeq{gchrec} is given by
\begin{equation}
    \gch_{j+1} = g_0 \prod_{k=0}^j (1-\beta_k \gch_k - \gch_k^{-1}r_k)
    = g_0 \exp \left( \sum_{k=0}^j \log (1-\beta_k \gch_k - O(\chi_k\gbar_k^2)) \right).
\end{equation}
By the dominated convergence theorem and \eqref{e:chigbd-bis}, the
limit
\begin{equation}
  \gch_\infty
  = \lim_{j\to\infty} \gch_{j} =
  g_0 \exp\left( \sum_{k=0}^{\infty} \log(1-\beta_k \gch_k + O(\chi_k\gch_k^2)) \right)
\end{equation}
exists and is continuous in $(m^2,g_0) \in [0,\delta)^2$.
It remains to prove
\eqref{e:ginfty}.
  For this, we set $\psi(t) = t^{-2}$ in \eqref{e:gbarsumbisx},
  and use that $\sum_{l=0}^{k-1} r_l \gch_l^{-2} = O(\sum_{l=0}^{k-1}\chi_l \gbar_l) = O(\log \gch_l)$
  by \eqref{e:chigbd-bis},
  to obtain
  \begin{equation} \label{e:gjinv}
    \gch_{k}^{-1}
    = g_0^{-1} + \sum_{j=0}^{k-1} \beta_j + O(\log \gch_{k})
    .
  \end{equation}
  Therefore, by Lemma~\ref{lem:betasum}, if $(m^2,g_0)\in (0,\delta)^2$ then
  \begin{equation} \label{e:ginftyasym}
    \gch_\infty^{-1} +O( |\log \gch_\infty|)
    = g_0^{-1} + \bubble_{m^2}.
  \end{equation}
  In particular $\gch_\infty \to 0$ as $m^2 \downarrow 0$ or $g_0 \downarrow 0$, and
  $\gch_\infty$ is right-continuous also at $m=0$ and $g_0=0$.
  Finally, \refeq{ginfty} follows from \refeq{ginftyasym}.
  This completes the proof.
\end{proof}

In the proof of the following lemma, we use the fact that
\begin{equation}
  \label{e:prodbdx}
  \prod_{k=j}^l (1 - \gamma \beta_k \gch_k)^{-1} =
  \left( \frac{\gch_{j}}{\gch_{l+1}} \right)^\gamma (c_j+O(\chi_l \gbar_l)),
\end{equation}
with $c_j = 1+O(\chi_j\gbar_j)$ a continuous function of $(m^2,g_0) \in [0,\delta)^2$.
The product formula \refeq{prodbdx} is
a consequence of the recursion relation \refeq{gchrec};
its proof is identical to that of
\cite[Lemma~\ref{flow-lem:elementary-recursion}(iii-a)]{BBS-rg-flow}
where the same statement is proved with $r_j =0$.
As in \cite[Lemma~\ref{flow-lem:elementary-recursion}(iii-a)]{BBS-rg-flow},
the continuity of $c_j$ follows
 from the continuity of $\beta_j$ and
$\gch_j$ (the latter was established below \eqref{e:gchrec-0}).
In the next lemma, we use primes to denote derivatives with respect to the
initial value $\nu_0^c=\mu_0^c$.

\begin{lemma} \label{lem:gzmuprime}
  Let $d = 4$.
  Let  $(m^2,g_0) \in [0,\delta) \times (0,\delta)$,
  $(z_0,\mu_0)=(z_0^c,\mu_0^c)$,
  and $s_j = (m^2, \ggen_j(m^2,g_0))$.
  There exists a continuous function $c:[0,\delta)^2 \to \R$,
  which satisfies
  $c(m^2,g_0)=1+O(g_0)$,
  such that for all $j \in \N_0$:
  \begin{equation} \label{e:mugzprime}
    \much_j'
    =
    L^{2j}
    \left(\frac{\gch_j}{g_0}\right)^\gamma(c(m^2,g_0) + O(\chi_j \gch_j))
    ,
    \quad
    \gch_j' = O\left(\much_j' \gch_j^{2}  \right),
    \quad
    \zch_j' = O\left(\chi_j \much_j' \gch_j^{2} \right).
  \end{equation}
  Also, for $N \in \N_0$
  and $(m^2,g_0) \in [\delta L^{-2(N-1)},\delta) \times (0,\delta)$,
  \begin{equation} \label{e:Kprime}
    \|K_j'\|_{\Wcal_j(s_j,\Lambda_N)} = O\left(\chi_j \much_j' \gch_j^2 \right),
  \end{equation}
  with $K_j$ the sequence of Proposition~\ref{prop:KjNbd} corresponding to $(m^2,g_0)$.
\end{lemma}
\begin{proof}
For $s_0=(m^2,g_0) \in [0,\delta)^2$,
let $V_j = V_j(s_0)$ be the infinite sequence
of Proposition~\ref{prop:KjNbd} with initial condition specified by $s_0$.
Let $\Vch_j = T_j(V_j) = (\gch_j, \zch_j,\much_j)$ be the transformed flow.
For $m^2 \in [\delta L^{-2(N-1)},\delta)$,
let $K_j = K_j(s_0)$ be the sequence given by Proposition~\ref{prop:KjNbd},
so that $(\Vch_j,K_j)$ is a finite volume flow of $(\Vch_+,\Kch_+)$,
as in \eqref{e:fvflow}.
(In fact, $K_j$ is defined for all $m^2 \in [0,\delta)$ if $j < N$.)
Let $(\Vch_j',K_j')$ denote the sequence of derivatives
along this solution,
with respect to the initial condition $\much_0 = \mu_0$,
and with the derivative $K_j'$ taken in the space $\Wcal_j(s_j) = \Wcal_j(s_j,\Lambda_N)$.
Since the sequence $V_j$ is independent of $\Lambda_N$ by Proposition~\ref{prop:KjNbd},
the sequence $\Vch_j'$ is also independent of $N$ and can be extended inductively
to all $j \in\N_0$ by choosing $N>j$.

We define $\Pi_j = \Pi_j(m^2,g_0)$ by
\begin{equation}
  \Pi_j = L^{2j}\prod_{l=0}^{j-1}(1-\gamma \beta_l \gch_l)
  .
\end{equation}
By \eqref{e:prodbdx}, there is a continuous
function $\Gamma_\infty (m^2,g_0) = O(g_0)$ such that
\begin{equation}  \label{e:prodid}
  \Pi_j =
  L^{2j}
  \left( \frac{\gch_{j}}{\gch_{0}} \right)^\gamma (1+\Gamma_\infty(m^2,g_0)+O(\chi_j\gch_j))
  .
\end{equation}
We also define $\Sigma_j = \Sigma_j(m^2,g_0)$ by
\begin{equation}
  \label{e:muPiSig}
  \much_j' = \Pi_j (1+ \Sigma_j) \quad (j \geq 0), \qquad
  \Sigma_{-1} = 0
  .
\end{equation}

We make the inductive assumption that for $j<N$
there exist $M_1 \gg M_2 \gg 1$ such that
\begin{equation} \label{e:induct1}
  |\Sigma_{j}-\Sigma_{j-1}| \leq O(M_1+M_2) \chi_j \gbar_j^2,
  \quad
  |\check g_j'|, |\check z_j'| \leq M_1 \chi_j \Pi_j \gbar_j^2 ,
  \quad
  \|K_j'\|_{\Wcal_j(s_j,\volume)} \leq M_2 \chi_j \Pi_j \gbar_j^2.
\end{equation}
Since $(\gch_0',\zch_0',\much_0',K_0') = (0,0,1,0)$,
the inductive assumption \eqref{e:induct1}  is true for $j=0$.
To advance the induction, we begin by applying
\eqref{e:prodid}--\eqref{e:muPiSig} to conclude that, if $L \gg 1$, if
$\Omega \leq L $, and if $g_0$ is sufficiently small, then
\begin{equation}
  |\much_j'| \le 2 \Pi_j
  ,\quad
  \chi_j \Pi_{j}\gbar_j^2 \leq \half \chi_{j+1} \Pi_{j+1} \gbar_{j+1}^2
  .
\end{equation}
We also use differentiated versions of the flow equation
\eqref{e:fvflow}.
By the chain rule, with $F = \Rch_{j+1}$ or $F=\Kch_{j+1}$,
\begin{equation} \label{e:chain}
  F'(\Vch_j,K_j) = D_{\Vch}F(\Vch_j,K_j)\Vch_j' + D_KF(\Vch_j,K_j)K_j'.
\end{equation}
By the estimates of Theorem~\ref{thm:step-mr-fv} (which apply as
discussed below \refeq{fvflow}) and by \eqref{e:induct1},
this gives
\begin{align}
  \|D_VF(\Vch_j,K_j)\Vch_j'\| &\leq O(\chi_j \gbar_j^2)(M_1 \gbar_j^2+2)\Pi_j \leq O(\chi_j \Pi_j\gbar_j^2),
  \\
  \|D_K\Rch_{j+1}(\Vch_j,K_j)K_j'\| &\leq O(M_2)\chi_j \Pi_j\gbar_j^2,
  \\
  \|D_K\Kch_{j+1}(\Vch_j,K_j)K_j'\| &\leq M_2\chi_j \Pi_j\gbar_j^2,
\end{align}
where the norms on the left-hand sides are those of the appropriate
$\Vcal,\Wcal(s_j)$ spaces.
This implies, for $M_2 \gg 1$,
\begin{equation} \label{e:phiprime}
  \|\Rch_{j+1}'(\Vch_j,K_j)\| \leq O(M_2) \chi_j \Pi_j\gbar_j^2,
  \quad
  \|\Kch_{j+1}'(\Vch_j,K_j)\| \leq 2M_2\chi_j\Pi_j\gbar_j^2.
\end{equation}

For $\check \mu$, the induction is advanced using the recursion \eqref{e:fvflow}
with \eqref{e:mubar}, \eqref{e:phiprime}, and
\begin{equation}
  (\eta_j \gch_j)', (\xi_j\gch_j^2)', (\gch_j\zch_j)' = O(M_1 \chi_j \Pi_j \gbar_j^2).
\end{equation}
It follows that
\begin{align} \label{e:muchO}
  \check\mu_{j+1}'
  &=
  L^2 \check\mu_j'(1-\gamma\beta_j \gch_j)
  + O\left((M_1+M_2)\chi_j \Pi_j  \gbar_j^2 \right)
  \nnb
  &
  = \Pi_{j+1}(1+\Sigma_j) + O\left((M_1+M_2)\chi_{j+1} \Pi_{j+1}  \gbar_{j+1}^2 \right)
  .
\end{align}
This enables us to advance
the induction for $\much$, namely the first estimate of \eqref{e:induct1}.
The advancement of the induction for $\check g$ and $\check z$ is similar,
as follows.
We use
the recursion relation \eqref{e:fvflow} with \eqref{e:gbar}--\eqref{e:zbar}
and \eqref{e:phiprime}, and choose $M_1 \gg M_2$ to obtain
\begin{align}
  |\gch_{j+1}'|, |\zch_{j+1}'|
  &\leq (M_1(1+O(\gbar_j))+O(M_2))\chi_j \Pi_j \gbar_j^2
  \nnb
  &\le 2 M_1 \chi_j \Pi_j \gbar_j^2
  \leq M_1 \chi_{j+1} \Pi_{j+1} \gbar_{j+1}^2.
\end{align}
This advances the induction for $\gch$ and $\zch$.

We now complete the proof,
having established that \eqref{e:induct1} holds for all $j < N$.
As discussed in the first paragraph of the proof, the bound
for $\Vch_j'$ in fact holds for all $j \in \N_0$.
Let $F(V,K,m^2) = V$.  By Proposition~\ref{prop:KjNbd} and Theorem~\ref{thm:step-mr-fv},
$F$ is a $C^\infty$ map of the renormalisation group coordinates at scale-$j$,
in the sense of Definition~\ref{def:flowcont}, for all $j\in \N_0$.  Therefore, by
Proposition~\ref{prop:Fcont},
$V_j'(s_0)$  is continuous in $s_0 \in [0,\delta) \times (0,\delta)$ for each $j \in \N_0$.
The same is therefore true for
$\Vch_j'(s_0)$.
As a consequence, since $\Pi_j$ is continuous,
it follows that $\Sigma_j(s_0)$ is continuous in $[0,\delta) \times (0,\delta)$,
for each $j \in \N_0$.
Since $\sum_{j=1}^\infty\chi_j \gbar_j^2 =O(g_0)$ by \refeq{chigbd-bis},
it follows from \eqref{e:induct1} that the limit $\Sigma_\infty  = \lim_{j \to \infty}\Sigma_j
=\sum_{j=1}^\infty (\Sigma_j - \Sigma_{j-1})$ exists with $\Sigma_\infty = O(g_0)$.
In particular, $\Sigma_j$ and $\Sigma_\infty$ extend continuously as $g_0 \downarrow 0$.
Similarly, continuity of $\Sigma_\infty$ on $[0,\delta)^2$ follows from the dominated convergence theorem,
with \eqref{e:approximate-flow}.
Also, again by \refeq{chigbd-bis},
\begin{equation}
  \label{e:sumid}
  \Sigma_\infty-\Sigma_j =
  O\left(\sum_{k=j+1}^\infty \chi_k \gbar_j^2 \right) = O(\chi_j \gbar_j)
  .
\end{equation}
From \eqref{e:prodid}--\eqref{e:muPiSig} and \eqref{e:sumid}, we obtain
the equation for $\much_j'$ in \eqref{e:mugzprime}, with
$c(m^2,g_0) = (1 + \Sigma_\infty(m^2,g_0))(1+\Gamma_\infty(m^2,g_0))$.
This $c(m^2,g_0)$ is indeed continuous,
since $\Sigma_\infty$ and $\Gamma_\infty$ are.
With this, \eqref{e:induct1} implies the last two equations in
\eqref{e:mugzprime} and \eqref{e:Kprime}.
\end{proof}

Lemmas~\ref{lem:ginfty}--\ref{lem:gzmuprime} are stated in terms of the transformed
variables. However, since $\gch_j = g_j + O(g_j^2)$, by \eqref{e:TVV2},
Lemma~\ref{lem:ginfty} also implies $g_\infty \sim \gch_\infty \sim 1/\bubble_{m^2}$.
Moreover, we note that \eqref{e:mugzprime} remains true if
$\much_j'$ is replaced by $\mu_j'$.
In fact, by \cite[\eqref{pt-e:much-def}]{BBS-rg-pt}, there
exist constants $a_j = O(1)$  such that
\begin{equation}
  \much_j = \mu_j + a_j\mu_j^2.
\end{equation}
Since $\mu_j = O(\chi_j \gch_j)$, by Proposition~\ref{prop:KjNbd},
this indeed implies
\begin{equation} \label{e:mup}
  \mu_j' = \much_j'(1 +O(\mu_j))
  = L^{2j}  \left(\frac{\gch_j}{g_0}\right)^\gamma(c(m^2,g_0) + O(\chi_j \gch_j)).
\end{equation}
In particular, setting $c(\hat g_0) = c(0, \hat g_0)$, and by Lemma~\ref{lem:ginfty},
\begin{equation} \label{e:nupinf}
  \lim_{N\to\infty} \nu_N'
  =
  c(m^2,g_0) \left(\frac{\gch_\infty}{g_0}\right)^\gamma
  \sim
  \frac{c(\hat g_0)}{(\hat g_0 \bubble_{m^2})^\gamma}
  \quad \text{as $(m^2,g_0) \to (0,\hat g_0)$.}
\end{equation}
Similarly, by \cite[\eqref{pt-e:gch-def}--\eqref{pt-e:zch-def}]{BBS-rg-pt},
\begin{equation} \label{e:gzp}
  g_j'
  = O(g_j\mu_j'),
  \quad
  z_j'
  = O(\chi_j g_j \mu_j').
\end{equation}

\subsection{Proof of Theorem~\ref{thm:suscept-diff}}
\label{sec:suscept-diff-pf}

\begin{proof}[Proof of Theorem~\ref{thm:suscept-diff}]
Let $\nu_0^c=\mu_0^c$ and $z_0^c$
be the functions of Proposition~\ref{prop:flow-flow},
which as desired are continuous in $(m^2,g_0)$
and differentiable in $g_0$.
As discussed at the beginning of Section~\ref{sec:pfmr},
it suffices to show that for $\hat\chi$ and $\hat\chi'$ evaluated
at $( m^2,g_0,\nu_0^c,z_0^c)$,
\begin{equation} \label{e:chi-m-pf-1}
  \hat\chi
  = \frac{1}{m^2} ,
  \quad\quad\quad
  \hat\chi'
  \sim - \frac{1}{m^4} \frac{c(\hat g_0)}{(\hat g_0\bubble_{m^2})^{\gamma}}
  \quad \text{as $(m^2 ,g_0)\to (0,\hat g_0)$},
\end{equation}
with $c$ continuous and $c(g_0)=1+O(g_0)$.
We do this using the identity \refeq{chibarm-bis2}, which asserts that
\begin{equation}
\label{e:chibarm-bis3}
  \hat\chi_N
  =
  \frac{1}{m^{2}}
  +
  \frac{1}{m^4}
  \left(
  -   \nu_N
  +
  \frac{1}{|\Lambda|}  D^2W_N^0(0, 0; 1,1)
  +  \frac{1}{|\Lambda|}  D^2K_N^0(0, 0; 1,1)
  \right)
  ,
\end{equation}
and is valid for initial conditions $V_0$ in an $N$-dependent
neighbourhood of $V_0^c$, as discussed below \eqref{e:ZIKNa}. In
particular, this allows the identity to be differentiated in $V_0$ at
$V_0=V_0^c$.

According to \eqref{e:munu}, we have $\nu_j = L^{-2j} \mu_j$, so
Proposition~\ref{prop:flow-flow} implies
$\nu_N = O(\chi_N \gbar_N L^{-2N}) \to 0$.  We prove $\hat\chi=m^{-2}$ by showing that
the $W$ and $K$ terms in \eqref{e:chibarm-bis3} are even smaller than
this, in fact by factors $O(\chi_N\gbar_N)$ and $O(\chi_N\gbar_N^2)$,
respectively.  For $\hat\chi'$, by \eqref{e:nupinf}, the leading term
$-m^{-4}\nu_N'$ obeys the asymptotic formula desired for $\hat\chi'$.
Thus it suffices to show that $\nu_0$-derivatives of the $W$ and $K$
terms in \eqref{e:chibarm-bis3} are relatively smaller as $m^2
\downarrow 0$.

Recall the definitions of the $\Vcal$,
$\Phi_N$, and $T_{0,N}$ norms in \refeq{Vcalnorm},
\eqref{e:Phinormdef}, and \eqref{e:Tphinormdef}.
The bounds for $W$ are more elementary than those for $K_N$, and are developed in \cite{BS-rg-IE}.
By definition, $ W_j(V,\tilde V)$ is bilinear in
$(V,\tilde V)$, and by \cite[\eqref{IE-e:W-logwish}]{BS-rg-IE},
\begin{equation} \label{e:Wbilinbd}
  \|W_{N}(V(\Lambda),\tilde V(\Lambda))\|_{T_{0,N}} \leq O(\chi_N) \|V\|_{\Vcal}\|\tilde V\|_{\Vcal}.
\end{equation}
We write $W_N =W_N( V_N(\Lambda), V_N(\Lambda))$
and $W_N'  = \ddp{}{\nu_0} W_N$.
Differentiation gives
\begin{equation}
  W_N'
  =
  W_{N}(\Lambda; V_N,V_N')
  + W_{N}(\Lambda; V_N',V_N),
\end{equation}
and it then follows from \eqref{e:gzp} that
\begin{equation} \label{e:WNNbd}
  \|W_{N} \|_{T_{0,N}} \leq O(\chi_N g_N^2),
  \quad
  \|W_{N}' \|_{T_{0,N}} \leq O(\chi_N  g_N  \mu_N')
  .
\end{equation}
By definition of the $T_\phi$ norm
(see \eqref{e:Tphinormdef}, or \cite{BS-rg-norm} for full details),
for a differential form $F\in \Ncal(\Lambda)$ with degree zero part $F^0$,
\begin{equation}
  |D^2F^0(0, 0; f, f)|
  \leq 2 \|F\|_{T_{0,N}} \|f\|_{\Phi_N}^2.
\end{equation}
The norm of the constant test function $1 \in \Phi_N$ is
\begin{equation}
\lbeq{1norm}
  \|1\|_{\Phi_N} = \ell_N^{-1} \sup_x |1_x| = \ell_N^{-1}
  = O(L^{N[\phi]})
  ,
\end{equation}
where $\ell_N = \ell_0 L^{-N[\phi]} = \ell_0 L^{-N(d-2)/2}$.  Since
$|\Lambda| = L^{dN}$, the bound on $W_N$ of \eqref{e:WNNbd} therefore
gives
\begin{equation}
\label{e:WNbd}
  |\Lambda|^{-1}
  |D^2 W_{N}^0(\Lambda; 0, 0; 1, 1)|
  \leq
  2 |\Lambda|^{-1} \|W_N \|_{T_{0,N}} \|1\|_{\Phi_N}^2
  \leq O(\chi_N \bar g_N^2 L^{-2N}),
\end{equation}
and the right-hand side vanishes in the limit $N \to\infty$.
Similarly, the $\nu_0$-derivative is bounded by $O(\chi_N \bar g_N
L^{-2N}\mu_N') = O(\chi_N \gbar_N \nu_N')$, so it is smaller than the
leading contribution $\nu_N'$ by a factor $O(\chi_N\gbar_N) \to 0$.

Proposition~\ref{prop:KjNbd} and Lemma~\ref{lem:gzmuprime} provide bounds
on $K_N$ and $K_N'$ analogous to \eqref{e:WNNbd} with one more power of $g_N$,
namely
\begin{equation}
  \label{e:KNbd2}
  \|K_{N}\|_{T_{0,N}} \leq O(\chi_N g_N^3),
  \qquad
  \|K_{N}'\|_{T_{0,N}} \leq O(\chi_N g_N^2 \mu_N').
\end{equation}
Thus the contribution due to $K$ is also small relative to the contributions
due to $\nu$, and the proof is complete.
\end{proof}

\subsection{Proof of Theorem~\ref{thm:nuc}}

\begin{proof}[Proof of Theorem~\ref{thm:nuc}]
Let $g_0 = \tilde g_0(g,0)$, with $\tilde g_0$ the function of
Proposition~\ref{prop:changevariables}(ii).
Let $\nu_0^c,z_0^c$ be given by $\mu_0^c,z_0^c$ of
Proposition~\ref{prop:flow-flow}.
By \eqref{e:g0g},
\begin{equation}
  \nu_c(g) = \frac{\nu_0^c(0,g_0)}{1+z_0^c(0,g_0)}
  = \nu_0^c(0,g_0)+O(g_0^2).
\end{equation}
Since $g_0 = g + O(g^2)$  by Proposition~\ref{prop:changevariables},
it suffices to show that
$\mu_0^c(0,g_0) = - 2C(0) g_0 + O(g_0^2) = -{\sf a} g_0 + O(g_0^2)$
(all covariances have $m^2 =0$ in this proof).
By Proposition~\ref{prop:flow-flow},
$\mu_0^c = \much_0$ where (see \eqref{e:mubar} and \eqref{e:fvflow})
the sequence $\much_j$ satisfies
\begin{equation} \label{e:much}
  \much_{j+1} = L^{2}\much_j(1-\gamma\beta_j \gch_j) + \eta_j \gch_j + O(\chi_j \gbar_j^2),
\end{equation}
and where we have used $\gch_j,\zch_j,\much_j = O(\gbar_j)$ for the
higher order terms, but not for the linear $\gch,\much$ terms and not for the $\much\gch$ term.
We also recall from \cite[\eqref{pt-e:nuplusdef}, \eqref{pt-e:wbardef2}]{BBS-rg-pt}
that
\begin{equation} \label{e:etadef}
  \eta_j = 2L^{2(j+1)} C_{j+1;0,0}
  .
\end{equation}

By Proposition~\ref{prop:flow-flow}, the sequence $\much_j$ is bounded,
so infinite iteration of \eqref{e:much} gives
\begin{equation} \label{e:mubar0}
  \much_0
  =
  - \sum_{l=0}^\infty
  \left( L^{-2(l+1)} \prod_{k=0}^l (1 - \gamma\beta_k\gch_k)^{-1} \right)
  (\eta_l \gch_l + O(\chi_l\gbar_l^2)).
\end{equation}
By \eqref{e:prodbdx}, we obtain from
\eqref{e:etadef}--\eqref{e:mubar0}
that
\begin{equation}
  \much_0 =
  - 2 (1+O(g_0))g_0^\gamma \sum_{l=0}^\infty
  \left( C_{l+1;0,0} \gch_l^{1-\gamma} + O(L^{-2l}\chi_l\gbar_l^{2-\gamma})\right).
\end{equation}
Since $C(0) = \sum_{l=0}^\infty C_{l+1;0,0}$, this gives
\begin{equation} \label{e:mubar0r}
\begin{aligned}
  \much_0
  &= - 2C(0) g_0 (1+O(g_0))
  - 2(1+O(g_0))g_0^\gamma\sum_{l=0}^\infty C_{l+1;0,0} (\gch_l^{1-\gamma}-g_0^{1-\gamma})
  \\ & \qquad\qquad
  + g_0^\gamma \sum_{l=0}^\infty  O(L^{-2l}\gbar_l^{2-\gamma}).
\end{aligned}
\end{equation}
We show that the last two terms are $O(g_0^2)$.  By
\eqref{e:gbarsumbisx} with $\psi(t) = (1-\gamma)^{-1} t^{-\gamma}$,
\begin{equation}
  \gch_l^{1-\gamma} - g_0^{1-\gamma} = \int_{\gch_0}^{\gch_l} \psi(t) \; dt
  =  (1-\gamma)^{-1} \sum_{k=0}^{l-1} \beta_k \gch_k^{2-\gamma} + O(g_0^{2-\gamma}).
\end{equation}
Thus, by Fubini's theorem and $C_{l+1;0,0} = O(L^{-2l})$, the middle sum in \eqref{e:mubar0r} is bounded by a multiple of
\begin{align}
  g_0^{\gamma}\sum_{l=0}^\infty C_{l+1;0,0} (\gch_l^{1-\gamma}-\gch_0^{1-\gamma})
  &=
  (1-\gamma)^{-1}
  g_0^\gamma \sum_{k=0}^{\infty} \beta_k \gch_k^{2-\gamma} \sum_{l=k+1}^\infty C_{l+1;0,0} + O(g_0^2)
  \nnb
  &=
   (1-\gamma)^{-1}
  g_0^\gamma \sum_{k=0}^{\infty} \beta_k \gch_k^{2-\gamma} O(L^{-2k})
  + O(g_0^2)
  = O(g_0^2).
\end{align}
Similarly, the rightmost sum in \eqref{e:mubar0r} is bounded by $O(g_0^2)$, so
\begin{equation}
  \much_0
  = - 2C(0) g_0 + O(g_0^2),
\end{equation}
and the proof is complete.
\end{proof}

\setcounter{section}{0}
\renewcommand{\thesection}{\Alph{section}}
\section{Existence of critical value}

The following lemma proves existence of a critical value in $(-\infty,0]$.
For dimensions $d>2$, a simple proof rules out $-\infty$.  Although our main
results pertain only to $d \ge 4$, we nevertheless show that the critical
value is finite in all dimensions.

\begin{lemma}
\label{lem:csub} For all dimensions $d>0$, there exists a
\emph{critical value} $\nu_c=\nu_c(d,g) \in (-\infty,0]$ such that
$\chi(g,\nu) < \infty$ if and only if $\nu > \nu_c$.  For $d>2$,
$\nu_c \in [-2C_0(0)g,0]$, where $C_0(x) = (-\Delta^{-1}_{\Zd})_{0,x}$
is the Green function of the simple random walk.
\end{lemma}

\begin{proof}
We first show that
$c_{T+S} \leq c_{T}c_{S}$ for all $S,T \ge 0$.
To prove this, let
\begin{equation}
  I(S,T) = \int_S^T \!\! \int_S^T \1_{X_{S_1} = X_{S_2}} \; dS_1 \, dS_2.
\end{equation}
Then
\begin{equation}
    I(T) = I(0,T) \geq I(0,S) + I(S,T).
\end{equation}
By the Markov property, $I(0,T)$ and $I(T, T+S)$ are conditionally independent given $X(T)$.
Using translation-invariance, it therefore follows that
\begin{equation}
    \label{e:subadditivity}
    c_{T+S} \leq E(e^{-gI(0,T)} e^{-gI(T,T+S)})
    = E(e^{-gI(0,T)}) E(e^{-gI(T,T+S)})
     = c_T c_S.
\end{equation}
A standard lemma for subadditive functions now yields the existence
of a \emph{critical value} $\nu_c=\nu_c(g) \in [-\infty,\infty)$ such that $c_T^{1/T}
\to e^{\nu_c}$ and also $c_T \ge e^{\nu_c T}$ (see, e.g., \cite[Lemma~1.2.2]{MS93}).
The inequality $\nu_c \leq 0$ is obvious from $I(S,T) \geq 0$.
For $d>2$, we show that  $\nu_c \in [-2C_0(0)g,0]$, as follows.
By Jensen's inequality,
\begin{equation}
  c_{T} = E(e^{-gI(T)}) \geq e^{-g E(I(T))}.
\end{equation}
An elementary estimate shows that
$E(I(T)) \le 2TC_0(0)$, and the Green function
$C_0(x)$ is finite for $d>2$
(the estimate can be done as in the
discrete case, see, e.g., \cite[Lemma~1.4]{BS95}).
This shows that $\nu_c \in [-2C_0(0)g,0]$.

To prove that  $\nu_c > -\infty$ also in lower dimensions,
we proceed as follows.
Let $S_{T,n}$ be the event that the path followed by
$X[0,T]$ is an $n$-step strictly self-avoiding walk.  Then
\begin{align}
\lbeq{chilower}
    \chi(\nu)  & = \int_0^\infty E(e^{-gI(T)}) e^{-\nu T} dT
    \ge
    \sum_{n=0}^\infty
    \int_0^\infty E(e^{-gI(T)} \mid S_{T,n}) P(S_{T,n}) e^{-\nu T} dT.
\end{align}
Let $s_n$ denote the number of $n$-step (discrete-time) strictly
self-avoiding walks that start at the origin.  The connective constant
$\mu = \lim_{n\to \infty}s_n^{1/n}$ exists and lies in $[d,2d-1]$.
Let $Y_T$ denote the number of steps taken by $X$ during the interval
$[0,T]$; this is a Poisson random variable with mean $2dT$.  For any
$\mu'<\mu$ there is a constant $k$ such that
\begin{equation}
\lbeq{PS}
    P(S_{T,n}) = \frac{s_n}{(2d)^n}P(Y_T=n)
    =  \frac{s_n}{(2d)^n} \frac{e^{-2dT} (2dT)^n}{n!}
    \ge
    k\frac{e^{-2dT} (\mu' T)^n }{n!}.
\end{equation}
Also, by Jensen's inequality,
\begin{equation}
\lbeq{JI}
    E(e^{-gI(T)} \mid S_{T,n})
    \ge
    e^{ -g   E(I(T) \mid S_{T,n})}.
\end{equation}

To evaluate $E(I(T)\mid S_{T,n})$,
we first observe that on the event $S_{T,n}$ the $n$ jump times
of the walk are independent and uniformly distributed on $[0,T]$, and in particular
the lengths of the $n+1$ subintervals of $[0,T]$ determined by the jump times are
identically distributed.
Let $U_0=0$, $U_{n+1}=T$, and let $U_1,\ldots,U_n$  be
independent uniform random
variables  on $[0,T]$ with order statistics
$U_{(1)},\ldots,U_{(n)}$ (i.e., ordered from smaller to larger values).  Then
\begin{equation}
\lbeq{CSrev1}
    E(I(T)\mid S_{T,n})
    =
    E\left[ \sum_{i=1}^{n+1} (U_{(i)}-U_{(i-1)})^2 \right]
    =
    (n+1)E  U_{(1)}^2
    .
\end{equation}
The probability density function of $U_{(1)}$ is $f(x)= nT^{-n}(T-x)^{n-1}$, so
\begin{align}
\lbeq{CSrev}
    E(I(T)\mid S_{T,n})
    & =
    \frac{(n+1)n}{T^n} \int_0^T x^2 (T-x)^{n-1}dx
    =
    \frac{2T^2}{(n+2)}
    .
\end{align}

By \refeq{JI} and \refeq{CSrev},
\begin{equation}
    E(e^{-gI(T)} \mid S_{T,n})
    \ge
    e^{ -2g  T^2/(n+2)}
    \ge
    e^{ -2g  T^2/n}.
\end{equation}
Therefore, by \refeq{chilower}--\refeq{PS} and with $h=2g$,
\begin{equation}
\begin{aligned}
    \chi(\nu) & \ge
    k
    \sum_{n=1}^\infty
    \int_0^\infty e^{-hT^2/n} \frac{e^{-2dT}(\mu' T)^n}{n!} e^{-\nu T} dT
    \\ &
    =
    k
    \sum_{n=1}^\infty
    \int_0^\infty e^{-hnT^2}\frac{e^{-2dnT}(\mu' n T)^n}{n!} e^{-\nu n T }
    n dT.
\end{aligned}
\end{equation}
Using Stirling's formula and $\mu' T \ge 1$ in the restricted integration
domain, we obtain
\begin{align}
\lbeq{Stirling}
    \chi(\nu)
    & \ge
    k
    \sum_{n=1}^\infty
    \frac{ne^n}{\sqrt{2\pi n}}
    \int_{(\mu')^{-1}}^\infty e^{-n(hT^2 + (2d+\nu)T)}    dT
    .
\end{align}
We consider $\nu < -2d$ and complete the square to get
\begin{equation}
    hT^2 + (2d+\nu)T = h(T - c_\nu)^2 - hc_\nu^2,
    \quad
    \text{with $c_\nu = (|\nu|-2d)/(2h)$.}
\end{equation}
We further assume that $c_\nu > (\mu')^{-1}$, i.e., that
$|\nu|> 2d +2h (\mu')^{-1}$,
and use Laplace's method to get, as $n \to \infty$,
\begin{align}
    &
    \int_{(\mu')^{-1}}^\infty e^{-n(hT^2 + (2d+\nu)T)}    dT
    \sim
    {\rm const}
    \frac{e^{nhc_\nu^2}}{\sqrt{n}}.
\end{align}
With \refeq{Stirling}, this gives a divergent lower bound on the susceptibility,
so $\nu_c \in [-2d -2h (\mu')^{-1},0]$
and the proof is complete.
\end{proof}

Concerning explicit bounds on the critical value for $d=1,2$,
since $\mu'$ is an arbitrary number less than $\mu$,
the above proof gives $\nu_c \in [-2d -4g\mu^{-1},0]$.
Using $\mu=1$ for $d=1$, and $\mu^{-1} \le \frac 12$ for $d=2$, this gives
$\nu_c \in [-2 -4g,0]$ for $d=1$ and
$\nu_c \in [-4 -2g,0]$ for $d=2$.

\section*{Acknowledgements}

This work was supported in part by NSERC of Canada.
This material is also based upon work supported by the National Science
Foundation under agreement No.\ DMS-1128155.
RB gratefully acknowledges the support of the University of British Columbia,
where he was a PhD student while much of his work was done.
Part of this work was done away from the authors' home institutions,
and we gratefully acknowledge the support and hospitality of
the IAM at the University of Bonn and
the Department of Mathematics and Statistics at McGill University (RB),
the Institute for Advanced Study at Princeton and  Eurandom (DB),
and
the Institut Henri Poincar\'e and  the Mathematical Institute of Leiden
University (GS).
We thank Alexandre Tomberg for useful discussions, and an anonymous
referee for helpful comments.

\bibliography{../../bibdef/bib}
\bibliographystyle{plain}


\end{document}